\documentclass[
 amsmath,amssymb,
 aps, onecolumn,
pra,
]{revtex4-2}
\usepackage[left = 1 in, right = 1 in, top = 1 in]{geometry}
\usepackage{bm}
\usepackage{amsthm,mathrsfs}
\usepackage{complexity}
\usepackage{graphicx}
\usepackage{pgfplots}
\usepackage[font=footnotesize]{caption}
\usepackage{subfig}
\usepackage{hyperref}
\usepackage{qcircuit}
\usepackage{algorithm}
\usepackage{algpseudocode}
\usepackage{dsfont}
\usepackage{adjustbox}

\allowdisplaybreaks
\newtheorem{theorem}{Theorem}
\newtheorem{lemma}[theorem]{Lemma}
\newtheorem{definition}[theorem]{Definition}

\newtheorem{proposition}[theorem]{Proposition}
\newtheorem{claim}[theorem]{Claim}
\newcommand\bra[1] {
	\langle#1|}
\newcommand\ket[1] {
	|#1\rangle}
\newcommand{\braket}[2]{\langle #1 | #2 \rangle}
\newcommand{\ketbra}[2]{|#1\rangle\!\langle #2|}
\newcommand{\QFT}{\text{\normalfont{QFT}}}

\begin{document}

\title{Efficient Quantum Simulation Algorithms in the Path Integral Formulation}
\author{
    Serene Shum$^{1}$, Nathan Wiebe$^{2,3,4}$
}
\affiliation{$^{1}$University of Toronto, Department of Physics, Toronto Canada}
\affiliation{$^{2}$University of Toronto, Department of Computer Science, Toronto Canada}
\affiliation{$^{3}$Pacific Northwest National Laboratory, High Performance Computing Division, Richland USA}
\affiliation{$^{4}$Canadian Institute for Advanced Study, Toronto Canada}

\begin{abstract}
    We provide a new paradigm for quantum simulation that is based on path integration that allows quantum speedups to be observed for problems that are more naturally expressed using the path integral formalism rather than the conventional sparse Hamiltonian formalism. 
    We provide a quantum algorithm for Lagrangians of the form $\sum_\gamma\frac{m_\gamma}{2}\dot{x}^2_\gamma - V(x)$, as well as two novel quantum algorithms based on Hamiltonian versions of the path integral formulation.  The Lagrangian path integral algorithm is based on a new rigorous derivation of a discrete version of the Lagrangian path integral. We show that for the time evolution of a system with $\eta$ particles in $D+1$ dimensions for time $T$, our Lagrangian simulation algorithm requires a number of queries to an oracle that computes the discrete Lagrangian that scales in the continuum limit as $\widetilde{O}(\eta D T^2/\epsilon)$, if $V(x)$ is bounded and finite and the wave function obeys appropriate position and momentum cutoffs.  This shows that Lagrangian dynamics can be efficiently simulated on quantum computers and opens up the possibility for quantum field theories for which the Hamiltonian is unknown to be efficiently simulated on quantum computers. Our first Hamiltonian path integral method breaks up the paths into short time steps. It is efficient under appropriate sparsity assumptions and requires a number of queries to  oracles that give the eigenvalues and overlaps between the eigenvectors of the Hamiltonian terms that scales as $t^{o(1)}/\epsilon^{o(1)}$ for simulation time $t$ and error $\epsilon$. The second approach uses long-time path integrals for near-adiabatic systems and has query complexity that scales as $O(1/\sqrt{\epsilon})$ if the energy eigenvalue gaps and simulation time is sufficiently long. 
\end{abstract}

\maketitle
\section{Introduction}
\setlength{\parindent}{15pt}

The central approach behind traditional algorithms for quantum simulation involves finding novel ways to compile, for a given Hamiltonian $\hat{H} \in \mathbb{C}^{2^n\times 2^n}$; evolution time $t$; uncertainty $\epsilon$ and distance measure dist$(\cdot)$, a quantum channel $\Lambda$ such that for all input states $\rho \in \mathbb{C}^{2^n\times 2^n}$ 
\begin{equation}
{\rm dist}(e^{-i \hat{H} t} \rho e^{i \hat{H} t} - \Lambda(\rho)) \le \epsilon.
\end{equation}
Traditionally, the aim is then to tailor the construction of $\Lambda$ such that the number of gates, ancillae and queries are minimized subject to the required accuracy constraint.  Such simulation of quantum systems was in fact the initial motivation for the proposal of a quantum model of computation in the first place. The fact that nature itself is quantum led Y. Manin and R. Feynman to both independently propose a new type of computing for the simulation of physical processes -- a \textit{quantum} computer rather than a classical one. In his 1980 book \textit{Computable and Non-Computable} \cite{manin1980}, Manin suggested mapping biomolecular processes to unitary operators in a finite-dimensional Hilbert space. A year later, Feynman conjectured in a keynote speech \cite{feynman1981} that quantum computers could be programmed to simulate any local quantum system efficiently. Feynman's claim was verified by S. Lloyd fifteen years later \cite{lloyd1996}. Since Lloyd's groundbreaking work many other simulation algorithms have since been constructed, such as the row-sparse Hamiltonian method \cite{aharonov2003}, multi-product Trotter formulas \cite{cw2012,lkw2019}, linear combination of unitaries \cite{cw2012}, quantum walks \cite{bn2016} and qubitization \cite{lc2019}.

A significant challenge arises, however, when we try to apply quantum simulation algorithms to quantum field theories.  Often in such applications the Hamiltonian is not easy to find, although the Lagrangian density may be.  The lack of a Hamiltonian formalism is not a significant drawback in these settings because the path integral formalism is often used to describe quantum dynamics in that setting.  This begs the question of whether efficient quantum algorithms can be designed to perform quantum simulation in the path integral formalism.

The primary motivation behind our work is to bridge this gap by providing quantum algorithms designed to simulate quantum dynamics on quantum computers within the path integral formalism.  We consider two different types of path integral simulation algorithms that Lagrangian-based and Hamiltonian path-integral based. 

Our Lagrangian-based algorithm assumes that we are given an oracle that can compute a discretized approximation to the Lagrangian functional $\mathscr{L}(x(t),\dot{x}(t))$, which is a function of position and velocity.  In this latter case, it is worth noting that the Lagrangian is not an operator like the Hamiltonian is, so showing that we can efficiently simulate dynamics within this such an input is a substantial departure from conventional quantum simulation literature.  Interestingly, our approach requires that the discretization chosen for space and time for our paths needs to be taken to be proportional to each other for our approach.  This is because the rigorous discrete Lagrangian path integral proof that we provide needs to have this requirement to have the precise classical Lagrangian to show up in discrete systems.  This new derrivation is needed because discretization of continuum Lagrangian path integration methods can lead to highly non-unitary dynamics that can lead to impractical quantum algorithms even with techniques like oblivious amplitude amplification and block-encodings~\cite{berry2015hamiltonian}.  This discretization requirement means that the Lagrangian methods can only yield arbitrarily small error in the continuum limit.

Our first class of Hamiltonian simulation algorithms works within a Hamiltonian formalism wherein we assume that we can write the Hamiltonian as a sum of terms
$   \hat{H} = \sum_j \hat{H}_j,$
where the eigenvalues for each $\hat{H}_j$ can be efficiently computed and the magnitudes of the inner products between eigenvectors of different $\hat{H}_j$ terms can also be efficiently computed.  The current inability to perform such simulations using quantum computers renders certain quantum dynamical systems, such as systems that are nearly adiabatic, difficult to simulate using existing quantum simulation techniques. On the other hand, our approach can simulate such systems efficiently because near-adiabatic evolutions can be concisely described using Hamiltonian path integrals~\cite{cheung2011improved,mackenzie2006perturbative}.

Our paper is laid out as follows.  First, we provide a review of the path integral formalism and discuss how the Lagrangian emerges out of the Schr\"{o}dinger equation for continuous variable dynamical systems in Section~\ref{sec:lagrange}. In Section~\ref{sec:LagAlg} we provide a derivation of a discrete version of the Lagrangian path integral formalism and then show how to use to efficiently simulate these dynamics.  Next, we review the Hamiltonian version of the path integral in Section~\ref{sec:HamPathInt}.  We further provide in Section~\ref{sec:HamPathInt} the details of our quantum algorithm for simulating quantum dynamics within the Hamiltonian path integral formalism that works by dividing the paths that comprise the dynamics into a large number of short paths and demonstrate that we can perform the simulation using a near-optimal number of queries to our oracles through the linear combination of unitaries (LCU) method and robust oblivious amplitude amplification (ROAA).
Section~\ref{sec:LTHamPathInt} provides an alternative grouping of paths wherein we break up our dynamics into a sum of long, rather than short paths, for cases that are near-adiabatic and propose an LCU-based strategy for performing such simulation.  We then conclude and discuss future applications.

\section{Review of Lagrangian Path Integral Formalism}\label{sec:lagrange}

The path integral formulation of quantum mechanics was worked out by R. Feynman in his PhD thesis in 1942 \cite{feynman1942}, based on ideas previously developed by P. Dirac \cite{dirac1932}. The Feynman path integral presents an alternate way of viewing the time evolution of a quantum state. The initial motivation was to formulate a quantum analogue to the stationary action principle (also called the least action principle) in classical mechanics, which states that the trajectories $x(t)$ that solve the Lagrange equations of motion are the stationary points of the action functional 
\begin{equation}
S = \int \mathscr{L}(x(t), \dot{x}(t)) dt,
\end{equation}
where $\mathscr{L}(x(t), \dot{x}(t))$ is the Lagrangian of the system. 

The stationary action principle is useful for solving the dynamics of many classical systems, especially those for which the Hamiltonian is not known and the coordinates and their conjugate momenta cannot be defined. However, if we want to quantize such systems, the canonical quantization procedure cannot be applied here due to the inability to describe them in the Hamiltonian formulation. The path integral method overcomes this issue because it does not require a Hamiltonian or momentum operator. 

The transition amplitude of an initial state $\ket{x_a}$ at time $t_a$ to a final state $\ket{x_b}$ at time $t_b$ is given in the path integral formulation as
\begin{equation}
    (x_a, t_a | x_b, t_b) = \int \mathcal{D} x \;e^{\frac{i}{\hbar} S(x, \dot{x})}
\end{equation}
where the functional integral $\int \mathcal{D} x$ is an integral over all trajectories $x(t)$ such that $x(t_a) = x_a$ and $x(t_b) = x_b$. Unlike the classical case where a particle can only travel along a single trajectory, the quantum particle's ``motion'' can be viewed as an integral over \textit{all} possible trajectories, where each trajectory $x(t)$ picks up a corresponding phase $\frac{S}{\hbar}$ that interferes with the phases of other trajectories. The transition amplitude comes from the interference from integrating over all these paths with varying phases. As $\hbar \rightarrow 0$ (i.e., taking the classical limit), the integrand $e^{\frac{i}{\hbar} S (x, \dot{x})}$ oscillates extremely rapidly, that is, small perturbations in $x$ will generally produce very large changes in $e^{\frac{i}{\hbar} S (x, \dot{x})}$. Most regions of the integral will go to zero because of the phase cancellation from this wildly fluctuating phase, but the regions of the integral that will contribute the most will be those where the phase is stationary, since there will mostly be constructive interference in those regions. This gives back the classical trajectory in the small $\hbar$ limit. This intuition can be made more rigorous using the stationary phase approximation \cite{ah1977}.

While Feynman \cite{feynman1942} originally began with the stationary action principle and then showed that the Schr\"odinger equation is recovered in the case where the Lagrangian is $\mathscr{L} = \frac{1}{2} m \dot{x}^2 - V(x)$, it is also possible to begin with the Schr\"odinger formulation and derive the expression for the path integral from there. Quantum mechanics usually starts from the Schr\"odinger equation which also applies in discrete as well as in the continuum, so this is the approach that we will take. We will show this derivation in Appendix \ref{sec:appendix_a} because the quantum algorithms in this paper closely follow the ideas shown here.

The unitary yielded by an arbitrary quantum circuit can also be expressed as a sum over paths as well~\cite{aaronson2010bqp}. The amplitude of a basis state after a unitary circuit is applied can be found by summing over all the possible ``paths'' to arrive at that state. If we have a quantum circuit acting on $n$ qubits, composed of a sequence of unitaries $\hat{U}_1, \hat{U}_2, \ldots, \hat{U}_N$, the unitary transformation corresponding to the full circuit is
\begin{equation}
    \hat{U} = \hat{U}_N \hat{U}_{N-1} \ldots \hat{U}_1
\end{equation}

We can insert a resolution of identity between each unitary in the sequence, as in the Feynman path integral derivation. Let $\ket{j}$ states ($j \in \{0, \ldots, 2^n-1\}$) be $n$-qubit computational basis states, then we can write 
\begin{align}
    \hat{U} = \sum_{j_1=0}^{2^n-1} \ldots \sum_{j_{N-1}=0}^{2^n-1} \hat{U}_N \ket{j_{N-1}} \bra{j_{N-1}} \hat{U}_{N-1} \ldots \ket{j_1} \bra{j_1} \hat{U}_1.
\end{align}
The matrix element $\bra{k} \hat{U} \ket{\ell}$, where $\ket{k}$ and $\ket{\ell}$ are computational basis states, is
\begin{equation}
    \sum_{j_1=0}^{2^n-1} \ldots \sum_{j_{N-1}=0}^{2^n-1} \bra{k} \hat{U}_N \ket{j_{N-1}} \bra{j_{N-1}} \hat{U}_{N-1} \ldots \ket{j_1} \bra{j_1} \hat{U}_1 \ket{\ell}.
\end{equation}
If $\hat{U}_j$ are row-computable---meaning that there exists an efficient algorithm for computing the values of each of the non-zero matrix elements in every row---then we can compute each of the terms in the path integral expansion as a product of $\bra{j_i} \hat{U}_i \ket{j_{i-1}}$ matrix elements. This idea of decomposing quantum circuits into paths was used in \cite{hl2021} as a method of building a classical simulator for quantum circuits. While there are exponentially many paths, we can estimate the mean value over these paths by Monte Carlo simulation.  Specifically, we can randomly draw paths from the sum and then compute the mean over the paths so drawn and then estimate the sum by multiplying the mean by the total number of paths.  This allows quantum circuits to be simulated using path integrals in polynomial space, which in turn is instrumental to showing that $\BQP$ is contained in $\PSPACE$ \cite{bv1993}. 

Despite this, polynomial time simulation is generally not possible because of the variance over the values of the paths.  Specifically, the fact that many complex numbers with wildly varying phases are summed over leads to large variance in the estimate in the mean; in the worst case, the number of required samples scales exponentially in the length of the circuit \cite{bv1993}.  This issue, known as a sign problem, is a ubiquitous problem with path integral methods and a great hope of quantum computers is their ability to sidestep sign problems in simulations.

\section{Lagrangian Path Integral Simulations}\label{sec:LagAlg}
We propose a quantum algorithm here to simulate the dynamics of a system given oracle access to a discretized version of the infinitesimal action $\mathscr{L} dt = \left( \frac{m}{2} \dot{x}^2 - V(x) \right) dt$ as a phase.
The Lagrangian is invariant under Lorentz transformations, while the Hamiltonian changes in different reference frames. This makes the Lagrangian much easier to work with in relativistic theories such as quantum field theories. There are also cases where the Lagrangian is easier to derive than the Hamiltonian. In the case where we do not know a system's Hamiltonian, it would be useful to still be able to simulate it. While we begin with a Hamiltonian in order to derive the expression for the Lagrangian path integral, the derived expression can then be used assuming oracle access to the Lagrangian, without requiring knowledge about the Hamiltonian terms.

One of the major downsides of the Lagrangian formalism, especially in the context of quantum computing, is that the unitarity of the time evolution is not readily apparent. We know that the evolution should be unitary because it is derived from the solution to the Schr\"odinger equation, but this derivation was done in infinite, continuous space. One of the major challenges that we solve here is that we derive a unitary Lagrangian evolution in the discrete setting in a restricted domain.

The Hamiltonian for a system of $\eta$ particles in $D$ spatial dimensions is
\begin{equation}
    \hat{H} = \sum_{\gamma=0}^{\eta D-1} \frac{1}{2m_\gamma} \hat{P}^2_\gamma + \hat{V}(\vec{x})
\end{equation}
where $\hat{P}_\gamma$ is the momentum operator for the $\gamma^{\rm th}$ direction/particle, $m_\gamma$ is the mass of particle corresponding to the $\gamma^{\rm th}$ index, and $\hat{V}$ is a potential that depends on time $t$ and the particle position vector $\vec{x} \in \mathbb{R}^{\eta D}$. We will have to redefine these continuous operators so that they act on a finite, discrete space. 

Let $[0,x_{{\rm max}, \gamma})$ be the position domain for each of the $\eta D$ coordinates. We will discretize the space and represent the positions with the computational basis states of $ \eta D n$ qubits, where we divide the Hilbert space into $\eta D$ registers of $n$ qubits. We have $2^n$ positions that we can represent for each of the $\eta D$ coordinates, where the computational basis state $\ket{q_\gamma}_\gamma$ ($q_\gamma \in \{0,...,2^n-1\}$) in the $\gamma^{\rm th}$ $n$-qubit register represents the position of the $\gamma^{\rm th}$ coordinate
\begin{equation}
    x_\gamma = \frac{q_\gamma x_{{\rm max},\gamma}}{2^n}.
\end{equation}
For simplicity, let
\begin{align}
    \ket{q} := \bigotimes_{\gamma=0}^{\eta D-1} \ket{q_\gamma}_\gamma.
\end{align}
Let 
\begin{equation}
    \Delta_{x, \gamma} = \frac{x_{\max,\gamma}}{2^n}
\end{equation} be the spacing between the positions represented by the computational basis states.  It may seem sensible that we would choose the spacing of all the grids for each of the coordinates to be the same, however, we will see later that in order accommodate arbitrary masses it will be useful to vary the spacing of each of the coordinates separately.

Firstly, we need to rewrite the Hamiltonian $\sum_\gamma \frac{1}{2m_\gamma} \hat{P}_\gamma^2 + \hat{V}$ so that it acts on our qubit space. The potential operator $\hat{V}$ only depends on position, so it is diagonal in the position basis (the computational basis). The computational basis states represent positions from $x = 0$ to $x = x_{\max,\gamma} - \Delta_x$, so the position eigenvalue corresponding to each computational basis state $\ket{q}$ is $q \Delta_x$. Thus, we can define the discrete position operator as follows:

\begin{definition}[Discrete position operator]\label{def:X}
    Let $n$ be a positive integer, $\gamma \in \{0, \ldots \eta D-1 \}$ where $\eta$ and $D$ are the number of particles and dimensions respectively, and $x_{\max,\gamma} \in \mathbb{R}$.  We then define the discrete position operator acting on the $\gamma^{\rm th}$ coordinate to be
        \begin{equation}
            \hat{X}_\gamma := \sum_{q=0}^{2^n-1} q \Delta_x \ket{q} \bra{q}
        \end{equation}
    acting on the $\gamma^{\rm th}$ $n$-qubit register, where $\Delta_{x,\gamma} = x_{\max, \gamma} / 2^n$.
\end{definition}
We similarly define the discrete momentum operator to be given by the Fourier dual of the position operator. 

\begin{definition}[Discrete momentum operator]\label{def:discP}
    Let $n$ be a positive integer and let $\gamma \in \{0, \ldots \eta D -1 \}$ for integer $\gamma$ where $\eta$ and $D$ are the number of particles and dimensions, respectively. We then define the discrete momentum operator for the $\gamma^{\rm th}$ particle/direction to be
        \begin{equation}
            \hat{P}_\gamma := \frac{2 \pi}{x_{\max,\gamma} \Delta_{x,\gamma}} \QFT_\gamma \hat{X}_\gamma \QFT_\gamma ^{\dagger},
        \end{equation}
    where $\QFT_\gamma$ is the quantum Fourier transform acting on the $\gamma^{\rm th}$ $n$-qubit register and $\hat{X}_\gamma$ is the discrete position operator acting on the $\gamma^{\rm th}$ coordinate as given in Definition~\ref{def:X}.
\end{definition}

In order to justify the discrete momentum operator $\hat{P}$, we will demonstrate that the above discrete momentum is equal to the generator of position translations,
that is, we will define $\hat{P}_\gamma$ as a Hermitian operator such that
\begin{align}
    e^{-i\hat{P}_\gamma \Delta_{x,\gamma}} \ket{q}_\gamma = \ket{q+1}_\gamma \label{p_generator}
\end{align}
where the computational basis states are taken mod $2^n$, so $e^{-i \hat{P}_\gamma \Delta_{x,\gamma}} \ket{2^n-1}_\gamma = \ket{0}_\gamma$. We will not impose boundary periodic boundary conditions on the wave function; the assumption is that $x_{\max,\gamma}$ has been chosen to be large enough that the dynamics do not take the state close to the edges of the position domain. It turns out that by defining the momentum operator this way, the discrete momentum operator can be written (with some constants) as the quantum Fourier transform of the discrete position operator, similar to the continuous case.

Recall that the quantum Fourier transform of an $n$-qubit computational basis state $\ket{j}$ is defined as
    \begin{align}
       \QFT \ket{j} &= \frac{1}{\sqrt{2^n}} \sum_{k=0}^{2^n-1} e^{i \frac{2 \pi jk}{2^n}} \ket{k}
    \end{align}
Let $\QFT_\gamma$ be the QFT acting on the $\gamma^{\rm th}$ $n$-qubit register. It is then easy to see that $\QFT_\gamma \ket{q}_\gamma$ is an eigenvector of $e^{-i\hat{P}_\gamma \Delta_{x,\gamma}}$:
\begin{align}
    e^{-i\hat{P}_\gamma \Delta_{x,\gamma}} \QFT_\gamma \ket{q}_\gamma &= \frac{1}{\sqrt{2^n}} e^{-i \hat{P}_\gamma \Delta_x}\sum_{j=0}^{2^n-1} e^{i \frac{2 \pi qj}{2^n}} \ket{j}_\gamma \\
    &= \frac{1}{\sqrt{2^n}} \sum_{j=0}^{2^n-1} e^{i \frac{2 \pi qj}{2^n}} \ket{j+1}_\gamma \\
    &= \frac{1}{\sqrt{2^n}}e^{-i \frac{2 \pi q}{2^n}}\sum_{j=0}^{2^n-1} e^{i \frac{2 \pi q(j+1)}{2^n}} \ket{j+1}_\gamma \\
    &= e^{-i \frac{2 \pi q}{2^n}} \QFT_\gamma \ket{q}_\gamma
\end{align}
This means that QFT$\ket{q}_\gamma$ is also an eigenvector of $\Delta_{x,\gamma} \hat{P}_\gamma$ with eigenvalue $\frac{2 \pi q}{2^n}$, so it is an eigenvector of $\hat{P}_\gamma$ with eigenvalue $\frac{2 \pi q}{2^n \Delta_{x,\gamma}}$. Thus, the definition of $\hat{P}_\gamma$ given in Def. \ref{def:discP} matches the definition of the momentum operator as the generator of position translations.

We will follow the Feynman path integral derivation, but adapted for this discrete, truncated space. The unitary time evolution operator for time $T \in \mathbb{R}$ is
\begin{align}
    \hat{U}(T) =e^{-i \hat{H} T } =e^{-i T \left(  \sum_\gamma \frac{1}{2m_\gamma} \hat{P}^2_\gamma + \hat{V}(\vec{x}) \right) }
\end{align}
where now $\hat{P}_\gamma$ is the discrete momentum operator defined above, and $\hat{V}$ is a function of $\vec{x}$, the length-$\eta D$ vector of discrete position operators $\hat{X}_\gamma$. Given some evolution time $T$, the transition amplitude from some initial state $\ket{q_a} := \bigotimes_\gamma \ket{q_{a,\gamma}}_\gamma$ at time $0$ to state $\ket{q_b} := \bigotimes_\gamma \ket{q_{b,\gamma}}_\gamma$ at $T$ is 
\begin{align}
    \bra{q_b} e^{-i T \sum_{\gamma'} \left( \frac{1}{2m_{\gamma'}} \hat{P}^2_{\gamma'} +  \hat{V}(\vec{x}) \right) }  \ket{q_a}
\end{align}
Following the standard path integral derivation, we will divide the total evolution time $T$ into $r$ time steps $T/r$:
\begin{align}
    & \quad\bra{q_b}  e^{-i T \sum_{\gamma'} \left( \frac{1}{2m_{\gamma'}} \hat{P}^2_{\gamma'} +  \hat{V}(\vec{x}) \right) } \ket{q_a} \\
    & = \bra{q_b}  \left( \prod_{k=0}^{r-1} e^{-i \frac{T}{r} \sum_\gamma\left( \frac{1}{2m_\gamma} \hat{P}^2_\gamma +  \hat{V}(\vec{x})  \right) } \right) \ket{q_a}
\end{align}
We will see that in order to preserve the unitarity of the evolution and still get back the Lagrangian in the phase in the continuum limit, we need to specifically choose the duration of the time steps such that
\begin{equation}
    r = \frac{2 \pi T}{m_\gamma x_{\max,\gamma} \Delta_{x,\gamma}}~\forall~\gamma
\end{equation}
It may seem a bit strange that such a specific value of $r$ needs to be chosen, since in the continuous case we just take any large value of $r$ and take the limit as $r \rightarrow \infty$. However, in the case of discrete position and momentum, as we will see, it turns out that we are also need to choose $r$ to take only specific values, in order to keep the evolution unitary while still ending up with something in the phase that resembles a Lagrangian. Since this is a quantum algorithm, keeping the evolution exactly unitary is advantageous because it keeps the algorithm straightforward and easy to analyze; we will not have to use things like linear combinations of unitaries and amplitude amplification to implement non-unitary transitions. 

The central reason for this choice of spacing arises from the following result for generalized Gauss sums.  This result allows us to, under these assumptions, perform a mathematical trick similar to the integration of the Gaussian integral used to remove the momentum from the continuous path integral derivation.
\begin{theorem}[Reciprocity theorem for generalized Gauss sums~\cite{berndt1998}] For integers $a, b, c$ such that $ac\neq0$ and $ac+b$ is even, the following holds:
    \begin{align}
        \sum_{n=0}^{|c|-1} e^{i \frac{\pi}{c} (an^2 + bn)} = \left| \frac{c}{a} \right|^{\frac{1}{2}} e^{i \frac{\pi}{4} \left( \text{\normalfont{sgn}}(ac) - \frac{b^2}{ac} \right)} \sum_{n=0}^{|a|-1} e^{-i \frac{\pi}{a} (cn^2 + bn)} 
    \end{align}
    \label{gauss_sum}
\end{theorem}

With this result in hand we can prove the following result, which gives a new derivation of a discrete analogue of the Feynman path integral.  As a further benefit, the derivation also is fully rigorous as all the operators used are bounded operators on finite-dimensional Hilbert spaces and thereby avoid the technical challenges posed in the continuous case.
\begin{lemma}\label{lem:lagrangian}
Let $\hat{P}_\gamma$ be the discrete momentum operator given in Definition~\ref{def:discP} and let $\hat{V}$ be a function of the discrete position operators $\hat{X}_\gamma$ given in Def. \ref{def:X}.  Further, let the duration of each of a short time step, $\tau: = T/r$, satisfy 
$\tau = m_\gamma x_{\max,\gamma} \Delta_{x,\gamma} /(2\pi)$ for all $\gamma$, let $\vec{x}_k\in \mathbb{R}^{\eta d}$ be a vector with entries $x_{k,\gamma} = q_{k,\gamma} \Delta_{x,\gamma}$ for $q_{k,\gamma} \in \{ 0, \ldots, 2^n-1\}$. Let the discrete Lagrangian be for any positions $\vec{x}_{k+1},\vec{x}_k \in \mathbb{R}^{\eta D}$
\begin{equation}
\mathscr{L}(\vec{x}_{k+1}, \vec{x}_k):=\sum_\gamma\frac{m_\gamma}{2} \left(\frac{x_{k+1,\gamma} - x_{k,\gamma}}{\tau} \right)^2 - V(x_{k,\gamma}).  
\end{equation}   Then, the time evolution operator can be expressed as the sum over all paths for each of the $\eta D$ coordinates:
\begin{align}
e^{-i(\sum_\gamma\frac{1}{2m_\gamma} \hat{P}^2_\gamma +\hat{V})\tau} &= \left( \frac{e^{-i \frac{\pi}{4}}}{\sqrt{2^n}} \right)^r \sum_{q_{a,\gamma}=0}^{2^n-1} \sum_{q_{1,\gamma}=0}^{2^n-1}  \cdots \sum_{q_{b,\gamma}=0}^{2^n-1}  e^{i \sum_{k=0}^{r-1} \mathscr{L}(\vec{x}_{k+1} ,\vec{x}_k ) \tau} \ket{q_b}\bra{q_a} \nonumber\\
&\qquad+\mathcal{O}(\eta D \max_{\vec{x}} |V(\vec{x})|T^2/(\min_\gamma \Delta_{x,\gamma}^2 \min_\gamma m_\gamma r)).
\end{align}
Here we take for notational convenience $q_0:=q_a$ and $q_r:=q_b$.
\end{lemma}

\begin{proof}
Our first step involves inserting resolutions of the identity $\openone = \bigotimes_{\gamma} \sum_{q_{k,\gamma}} \ket{q_{k,\gamma}} \bra{q_{k,\gamma}}:= \sum_{q_k} \ketbra{q_k}{q_k} $ between the time evolution operators in each time step of duration $\tau := T/r$:
\begin{align}
    & \quad \bra{q_b} \left( \prod_{k=0}^{r-1} e^{-i \tau \sum_{\gamma}\left( \frac{1}{2m_{\gamma}} \hat{P}^2_{\gamma} +  \hat{V}(\vec{x}_k)  \right) } \right) \ket{q_a}  \\
    &= \prod_{k=0}^{r-1} \left( \sum_{q_k} \bra{q_{k+1}} e^{-i \tau \sum_{\gamma} \left( \frac{1}{2m_{\gamma}} \hat{P}^2_{\gamma} +  \hat{V}(\vec{x}_k)  \right) } \ket{q_{k}} \right)
\end{align}
Consider one time slice $\bra{q_{k+1}} e^{-i \tau \hat{H}} \ket{q_k}$:
\begin{align}
    & \quad \bra{q_{k+1}} e^{-i \tau \sum_{\gamma'} \left( \frac{1}{2m_{\gamma'}} \hat{P}^2_{\gamma'} + \hat{V} (\vec{x}) \right)} \ket{q_k} \\
    &= \bra{q_{k+1}} e^{-i \tau \sum_{\gamma} \frac{1}{2m_{\gamma}} \hat{P}^2_{\gamma} } e^{-i \tau \hat{V}(\vec{x})} \ket{q_k} + \epsilon_{Trot_1} \\
    &=\bra{q_{k+1}} \left( e^{-i \tau  \sum_\gamma \frac{1}{2m_\gamma} \hat{P}^2_\gamma } \bigotimes_{\gamma'=0}^{\eta D-1} \sum_{j_{\gamma'}=0}^{2^n-1} \QFT \ket{j_\gamma} \bra{j_\gamma} \QFT^\dagger\right) e^{-i \tau V(\vec{x}_k)} \ket{q_k} + \epsilon_{Trot_1}\label{derivation1} \\
    &= e^{-i \tau V(\vec{x}_k)} \bigotimes_{\gamma=0}^{\eta D-1} \sum_{j_\gamma=0}^{2^n-1}\bra{q_{k+1, \gamma}} \left( e^{-i \tau \frac{1}{2m_\gamma} \left(\frac{2 \pi j_\gamma}{x_{\max,\gamma}} \right)^2 } \QFT \ket{j_\gamma} \bra{j_\gamma} \QFT^\dagger\right)  \ket{q_{k, \gamma}} + \epsilon_{Trot_1}\\
    &= \frac{1}{2^{n\eta D}} e^{-i \tau V(\vec{x}_k)}\prod_{\gamma=0}^{\eta D-1}\sum_{j_\gamma=0}^{2^n-1} e^{-i \tau  \frac{2 \pi^2}{m_\gamma x_{\max,\gamma}^2} j_\gamma^2 } e^{i \frac{2 \pi (q_{k+1,\gamma}-q_{k,\gamma})}{2^n}j_\gamma }+ \epsilon_{Trot_1}\\
    &= \frac{1}{2^{n\eta D}} e^{-i \tau V(\vec{x}_k)}\prod_{\gamma=0}^{\eta D-1} \sum_{j=0}^{2^n-1} e^{i \frac{\pi}{2^n} \left( -\frac{2 \pi \tau}{m_\gamma x_{\max,\gamma} \Delta_{x,\gamma} } j_\gamma^2 + 2(q_{k+1,\gamma}-q_{k,\gamma})j_\gamma \right)}+ \epsilon_{Trot_1}
\end{align}
$\epsilon_{Trot_1}$ is the Trotter error from one time step.  This can be bounded by the sum of the norms of the commutators multiplied by $T^2/r$~\cite{childs2021theory}.  Upper bounding all commutators by the worst commutator norm, we find
\begin{equation}
    |\epsilon_{Trot_1}| \in \mathcal{O}\left(\max_x |V(x)| \eta D \|P_\gamma\|^2 t^2/(\min_\gamma m_\gamma r) \right)=\mathcal{O}(\eta D \max_x |V(x)|t^2/(\min_\gamma \Delta_{x,\gamma}^2 \min_\gamma m_\gamma r))
\end{equation}

Since $r = \frac{2 \pi t}{m_\gamma x_{\max,\gamma} \Delta_{x,\gamma}}$ for all $\gamma$, for each $\gamma \in \{0, \ldots, \eta D-1\}$ we get a factor of
\begin{equation}
    \frac{1}{2^n}\sum_{j_\gamma=0}^{2^n-1} e^{i \frac{\pi}{2^n} \left( -\frac{2 \pi \tau}{m_\gamma x_{\max,\gamma} \Delta_{x,\gamma} } j_\gamma^2 + 2(q_{k+1,\gamma}-q_{k,\gamma})j \right)} = \frac{1}{2^n} \sum_{j=0}^{2^n-1} e^{i \frac{\pi}{2^n} \left( -j^2 + 2 (q_{k+1,\gamma} - q_{k,\gamma}) j_\gamma \right)} \label{derivation2}
\end{equation}
The above is a generalized quadratic Gauss sum, and we can use  Theorem~\ref{gauss_sum} to evaluate the summation.  Specifically,
In \eqref{derivation2}, we have $a=-1$, $b = 2(q_{k+1,\gamma} - q_{k,\gamma})$ and $c=2^n$. Since $-ac = 2^n \neq 0$ and $ac+b = 2^n + 2(q_{k+1,\gamma}-q_{k,\gamma})$ is even, we can use Theorem \eqref{gauss_sum}:
\begin{align}
    \frac{1}{2^n} \sum_{j_\gamma=0}^{2^n-1} e^{i \frac{\pi}{2^n} \left( - j_\gamma^2 + 2(q_{k+1,\gamma}-q_{k,\gamma})j_\gamma \right)} &= \frac{1}{2^n}  \sqrt{2^n} e^{i \frac{\pi}{4} \left( -1 + \frac{4(q_{k+1,\gamma}-q_{k,\gamma})^2}{2^n} \right) } \\
    &= \frac{e^{-i \frac{\pi}{4}}}{\sqrt{2^n}}   e^{i\frac{\pi}{4} \frac{4 m_\gamma \Delta_{x,\gamma}^2 (q_{k+1,\gamma}-q_{k,\gamma})^2}{2 \pi \tau} } \\
    &=  \frac{e^{-i \frac{\pi}{4}}}{\sqrt{2^n}} e^{i \tau \left( \frac{m_\gamma}{2} \left(\frac{x_{k+1,\gamma} - x_{k,\gamma}}{\tau} \right)^2\right)}
\end{align}
We then have that
\begin{align}
    \frac{1}{2^{n\eta D}} e^{-i \tau V(\vec{x}_k)}\prod_{\gamma=1}^{\eta D-1}\sum_{j=0}^{2^n-1} e^{i \frac{\pi}{2^n} \left( -\frac{2 \pi \tau}{m_\gamma x_{\max,\gamma} \Delta_{x,\gamma} } j_\gamma^2 + 2(q_{k+1,\gamma}-q_{k,\gamma})j_\gamma \right)} &= \left(\frac{e^{-i \frac{\pi}{4}}}{\sqrt{2^n}}\right)^{\eta D} e^{-iV(\vec{x}_k) \tau}\prod_{\gamma=0}^{\eta D-1} e^{i \tau \left( \frac{m_\gamma}{2} \left(\frac{x_{k+1,\gamma} - x_{k,\gamma}}{\tau} \right)^2\right)}\nonumber\\
    &= \left(\frac{e^{-i \frac{\pi}{4}}}{\sqrt{2^n}}\right)^{\eta D} e^{i
    \tau \left(\sum_{\gamma} \frac{m_\gamma}{2} \left(\frac{x_{k+1,\gamma} - x_{k,\gamma}}{\tau} \right)^2-V(\vec{x}_k)\right)} 
\end{align}
The phase 
\begin{equation}
\sum_\gamma\frac{m_\gamma}{2} \left(\frac{x_{k+1,\gamma} - x_{k,\gamma}}{\tau} \right)^2 - V(\vec{x}_k) = \mathscr{L}(\vec{x}_{k+1}, \vec{x}_k)
\end{equation}is the discrete Lagrangian evaluated at position, which is exactly the same as the result we get in the continuous, infinite case before we take the $r \rightarrow \infty$ continuum limit. Thus, by picking a specific value for the number of time steps, we arrive at a very similar expression to the standard path integral, even in our discrete, truncated space, with no approximations made other than the Trotter expansion. This expression is exactly equal to \eqref{derivation1}, so the total evolution is still unitary. The full expression for $\bra{q_b} e^{-i T \sum_\gamma \frac{\hat{P}^2_\gamma}{2m_\gamma} } e^{-i T \hat{V} } \ket{q_a}$ is
\begin{align}
    \bra{q_b}_\gamma e^{-i T \sum_\gamma \frac{\hat{P}^2_\gamma}{2m_\gamma} } e^{-i T \hat{V} } \ket{q_a} &= \left( \frac{e^{-i \frac{\pi}{4}}}{\sqrt{2^n}} \right)^{\eta D r} \sum_{q_1} \cdots \sum_{q_{r-1}}  e^{i \tau \sum_{k=0}^{r-1}\sum_{\gamma=0}^{\eta D-1} \left( \frac{m_\gamma}{2} \left(\frac{x_{k+1,\gamma} - x_{k,\gamma}}{\tau} \right)^2 - V(\vec{x}_k)\right)}
\end{align}
with $q_0 := q_a$ and $q_r := q_b$ (recall above that $x_{k,\gamma} = q_{k,\gamma}\Delta_{x,\gamma}$, and $q_k$ is the vector of $q_{k,\gamma}$ values). This can be seen as a sum over all the possible paths from $\ket{q_a}$ to $\ket{q_b}$, with a phase associated with each one. Thus, we can write the time evolution operator as
\begin{align}
    \hat{U}_{\mathscr{L}} = \left( \frac{e^{-i \frac{\pi}{4}}}{\sqrt{2^n}} \right)^{r\eta D} \sum_{q_{a}} \sum_{q_{b}} \sum_{q_1} \cdots \sum_{q_{r-1}}  e^{i \tau \sum_{k=0}^{r-1}\sum_{\gamma=0}^{\eta D-1} \left( \frac{m_\gamma}{2} \left(\frac{x_{k+1,\gamma} - x_{k,\gamma}}{\tau} \right)^2 - V(\vec{x}_k)\right)} \ketbra{q_{b}}{q_{a}} \label{U_lagrange}
\end{align}
\end{proof}
This result shows a discrete analogue of the Lagrangian path integral.
\subsection{Algorithm}
The Lagrangian for the $k^{\rm th}$ time slice is a function of positions $\vec{x}_k$ and $\vec{x}_{k+1}$. We will define the action oracle $O_S$, which gives the action for one time slice in the phase:
\begin{definition}[Action oracle]
Let $O_S \in \mathcal{H}_{2^{n \eta D}} \otimes \mathcal{H}_{2^{n \eta D}})$ be a unitary operator such that its action on an arbitrary computational basis state is
$$
    O_S:\ket{q_{k}} \ket{q_{k+1}} \mapsto e^{i  \left( \sum_{\gamma} \frac{m_\gamma}{2 \tau} (x_{k+1,\gamma}-x_{k,\gamma})^2 - \tau V(\vec{x}_{k}) \right)} \ket{q_{k}} \ket{q_{k+1}} $$\label{def:O_S}
\end{definition}

In order to implement the Lagrangian version of the path integral, we need more than just the infinitesimal action.  We also need to perform the transition from position $q_k$ to position $q_{k+1}$.  The following theorem is our main theorem that provides us with a construction for performing the required transitions weighted by the discretized action integrals for the paths.

\begin{theorem}\label{thm:mainLagrangian}
    Let $\frac{2 \pi T}{m_\gamma x_{\max,\gamma} \Delta_{x,\gamma}}$ be constant for all $\gamma$ and let $r=\frac{2 \pi T}{m_\gamma x_{\max,\gamma} \Delta_{x,\gamma}}$, where $\Delta_{x,\gamma}$ is the spacing size in the space discretization for the $\gamma^{\rm th}$ coordinate. Let $\vec{q}_k \in \{0, \ldots, 2^n-1\}^{\eta D}$ and let $\vec{x}_k$ have entries $x_{k, \gamma} = \Delta_{x, \gamma} q_{k,\gamma}$. There exists a quantum algorithm that exactly implements the unitary operator $\hat{W} \in \mathcal{L}(\mathcal{H}(2^{n\eta D}))$
    \begin{align}
        \hat{W}=\left( \frac{e^{-i \frac{\pi}{4}}}{\sqrt{2^n}} \right)^{r\eta D} \sum_{\vec{q}_{a}}\sum_{q_{b}}\sum_{q_1} \cdots \sum_{\vec{q}_{r-1}}  e^{i \tau \sum_{k=0}^{r-1}\sum_{\gamma=0}^{\eta D-1} \left( \frac{m_\gamma}{2} \left(\frac{x_{k+1,\gamma} - x_{k,\gamma}}{\tau} \right)^2 - V(\vec{x}_k)\right)} \ketbra{\vec{q}_{b}}{\vec{q}_{a}},
    \end{align}
    up to a global phase with $O(r)$ queries to the action oracle defined in Def. \ref{def:O_S} and $O(r\eta D n^2)$ additional two-qubit gates.
\end{theorem}
\begin{proof}

We can implement the transition at each time step by:
\begin{align}
    \ket{q_k} \ket{0} \xrightarrow{O_S} & e^{i \left( \sum_{\gamma=0}^{\eta D-1}\frac{m_\gamma}{2 \tau} (0-x_{k,\gamma})^2 - \tau V(\vec{x}_k) \right)} \ket{q_k} |0\rangle\\
    & \quad =  e^{i \sum_{\gamma=0}^{\eta D-1}\frac{2 \pi q_{k,\gamma}^2}{2^{n+1}} - \tau V(\vec{x}_k) } \ket{q_k} \ket{0} \\
\xrightarrow{(\QFT^\dagger)^{\otimes \eta D} \otimes \openone} & e^{i \frac{2 \pi \sum_\gamma q_{k,\gamma}^2}{2^{n+1}} - \tau V(\vec{x}_k) } \left( \frac{1}{\sqrt{2^n}} \right)^{\eta D} \bigotimes_{\gamma=0}^{\eta D-1}\left(\sum_{q_{k+1,\gamma}=0}^{2^n-1} e^{-i \frac{2 \pi q_{k,\gamma} q_{k+1,\gamma}}{2^n}} \ket{q_{k+1,\gamma}}_\gamma \ket{0}\right)\\
    \xrightarrow{O_S} & e^{i \frac{2 \pi \sum_\gamma q_{k,\gamma}^2}{2^{n+1}} - \tau V(\vec{x}_k) } \left( \frac{1}{\sqrt{2^n}} \right)^{\eta D} e^{-i \frac{2 \pi \sum_{\gamma=0}^{\eta D-1} q_{k,\gamma} q_{k+1,\gamma}}{2^n}} e^{i \left( \frac{2 \pi \sum_{\gamma=0}^{\eta D-1} q_{k+1,\gamma}^2}{2^{n+1}} - \tau V(\vec{0}) \right)} \ket{q_{k+1}} \ket{0}
\end{align}
The process above maps each $\ket{q_k}$ state to a superposition of all $\ket{q_{k+1}}$ states. The associated phase on each $\ket{q_{k+1}}$ state is
\begin{align}
    & \frac{2 \pi}{2^n} \sum_\gamma\left( \frac{q_{k,\gamma}^2}{2} - q_{k,\gamma} q_{k+1,\gamma} + \frac{q_{k+1,\gamma}^2}{2} \right) - \tau V(\vec{x}_k) - \tau V(\vec{0}) \\
    & \quad = \frac{2 \pi}{2^n} \frac{1}{2} \sum_\gamma\left( q_{k+1,\gamma} - q_{k,\gamma} \right)^2 - \tau V(\vec{x}_k) - V(\vec{0}) \label{phase}
\end{align}
Since $\tau = \frac{m_\gamma x_{\max,\gamma} \Delta_{x,\gamma}}{2 \pi} = \frac{m_\gamma 2^n (\Delta_{x,\gamma})^2}{2 \pi}$, \eqref{phase} (for all $\gamma$) the total phase is  equal to
\begin{align}
    & \sum_\gamma\frac{m_\gamma (\Delta_{x,\gamma})^2}{2 \tau}  \left( q_{k+1,\gamma} - q_{k,\gamma} \right)^2 - \tau V(\vec{x}_k) - \tau V(\vec{0})\\
    & \quad = \sum_\gamma\frac{m_\gamma}{2 \tau}  \left( x_{k+1,\gamma} - x_{k,\gamma} \right)^2 - \tau V(\vec{x}_k) - \tau V(\vec{0})
\end{align}
which is exactly what we want, except for the extra $-\tau V(\vec{0})$ term. However, this term will appear on every transition, so it will contribute a global phase that can be factored out and ignored at the end.

The above implements the transition for one time step of length $\tau = \frac{T}{r}$, and then in order to implement the entire evolution we simply apply it $r$ times.
    Each time step of length $\tau = \frac{T}{r}$ takes two oracle queries and $\eta D$ iterations of the $n$-qubit quantum Fourier transform to implement, which gives us the required complexity scaling for implementing the path integral using the exact quantum Fourier transform which consists of $n^2$ two-qubit gates.  As this is repeated $r$ times, the claimed complexity follows.
\end{proof}

\subsection{Time and Query Complexity of Lagrangian Path Integral Simulation}

The algorithm of Theorem~\ref{thm:mainLagrangian} exactly implements $\hat{U}_{\mathscr{L}}$ defined in Eq. \eqref{U_lagrange}. If we consider the Lagrangian path integral expression to be the exact time evolution that we want to implement, then there is no error to consider, other than any discretization error if the path integral is being used to approximate a continuous system. If we would like to compare the path integral-evolved state to the exact Hamiltonian time evolution, then we do have to consider the Trotter error. This error comes from applying the Trotter decomposition in \eqref{derivation1}. For example, in the one-dimensional case the approximation replaces the exact expression
\begin{align}
    U(\tau):=e^{-i \tau \left( \sum_\gamma\frac{\hat{P_\gamma}^2}{2m_\gamma} + \hat{V} \right)}
\end{align} with the first-order Trotter approximation
\begin{align}
    \tilde{U}(\tau):=e^{-i \tau \frac{\hat{P}^2}{2m}} e^{-i \tau \hat{V}}
\end{align}
This approximation incurs an additive error of 
\begin{equation}
    \| U(\tau) - \tilde{U}(\tau) \| \le \|[ \hat{P}^2/2m, \hat{V} ]\| \tau^2
\end{equation} per time step \cite{cstwz2021}. The error scales with the commutator between the two terms of the Hamiltonian. Using the naive upper bound by simply taking the norms of the operators does not give ideal error scaling, since the norm of $\hat{P}$ grows with $2^n$. The issue of the Trotter error being unbounded arises in the standard continuous Feynman path integral derivation as well. Many commonly used potentials, such as the Coulomb potential, are unbounded, and the Trotter decomposition does not hold in those cases. We will need to systematically enforce specific cutoffs to ensure that the analysis is well posed. This is a valid assumption for many physically realistic systems. The overall error will end up scaling with the maximum momentum $P_{max}$ and the maximum potential $V_{\max}$.


\begin{definition} \label{def:S_n} 
    Let $\hat{K}:= \sum_{\gamma=0}^{\eta D-1} \frac{1}{2m_\gamma} \hat{P}_\gamma^2 \in \mathcal{L}(\mathcal{H}_{ 2^{n}}^{\otimes \eta D})$ with $\hat{P}$ defined as in Def \ref{def:discP}, let $\hat{V}\in \mathcal{L}(\mathcal{H}_{ 2^{n}}^{\otimes \eta D})$ be diagonal in the computational basis, and let $\hat{H} = \hat{K} + \hat{V}$. Let the ``feasible subspace'' $S_n$ be the set of $\eta D n$-qubit states such that
 for all $\ket{\psi}\in S_n$ ($|\psi \rangle = \sum_{k_0}^{2^n-1} \cdots \sum_{k_{\eta D-1}}^{2^n-1} a_{k_0, \ldots, k_{\eta D-1}} \bigotimes_{\gamma=0}^{\eta D-1} \ket{p_{k_\gamma}}$ where $\ket{p_{k_\gamma}}$ are the eigenstates of $\hat{P}_\gamma$ with eigenvalues $p_{k_\gamma} = \frac{2 \pi k}{x_{\max}}$), the following are satisfied:
\begin{enumerate} 
\item There exists $P_{\max}$ such that $a_k =0$ when $\max_\gamma |p_{k_\gamma}| > P_{\max}$ 
\item For all $|p_{k_\gamma}| \le P_{\max}$, $\| \hat{K} e^{-it \hat{V}} |p_k\rangle \|\le \sum_\gamma \frac{P_{\max}^2}{2m_\gamma}:= K_{\max}$   for all $t \in [0,T]$
\item For all $t \in [0,T]$, $e^{-i \hat{H} t} |\psi\rangle:= \sum_{k_0}^{2^n-1} \cdots \sum_{k_{\eta D-1}}^{2^n-1} a_{k_0, \ldots, k_{\eta D-1}}(t) \bigotimes_{\gamma=0}^{\eta D-1} \ket{p_{k_\gamma}}$ where $a_{k_0, \ldots, k_{\eta D-1}}(t) =0$ if $|p_{k_\gamma}| > P_{\max}$ for any $k_\gamma \in \{k_0, \ldots, k_{\eta D-1} \}$
\item There exists $V_{\max}\ge 0$ such that for all $n\ge 1$, $|\bra{\psi} \hat{V} \ket{\psi}| \le V_{\max}$ for all $\ket{\psi} \in S_n$.
\end{enumerate}
\end{definition}
With these definitions in place, we can bound the Trotter error within the truncated space $S_n$. We formally describe the Trotter error by describing the maximum error as the error in measuring the discrepancies in probabilities of the outcomes of any POVM; however, at an intuitive level one can simply think of this result as bounding the maximum observable error seen for any input state conforming to Definition~\ref{def:S_n}.
\begin{lemma}[Lagrangian Trotter Error]\label{lem:Trotter_Lagrange}
Let $\hat{H}\in \mathcal{L}(\mathcal{H}_{2^n}^{\otimes \eta D})$ be a Hamiltonian operator conforming to Definition~\ref{def:S_n} that acts on a feasible subspace $S_n$.  Let ${M} = \{M_k: k =1,\ldots N \}$ be a POVM such that $M_k$ are POVM elements that maps the subspace $S_n$ onto itself.  We then define
$$    \hat{U} = e^{-\frac{iT}{r} \left({\hat{K}} + \hat{V} \right) },\qquad     \widetilde{U} = e^{-\frac{iT}{r} \hat{K} } e^{-\frac{iT}{r}\hat{V}}, $$
then the Trotter-Suzuki error in the expectation value within this space can be be bounded above by
$$
D(\hat{U}^r,\tilde{U}^r):=\max_{k}\sup_{\ket{\psi}\in S_n}\left| \langle \psi | (\hat{U}^\dagger)^r \hat{M}_k (\hat{U})^r |\psi\rangle - \langle \psi | (\widetilde{U}^\dagger)^r \hat{M}_k (\widetilde{U})^r |\psi\rangle \right| \le     \frac{2 T^2 V_{\max} K_{\max}}{r} 
$$
where $K_{\max} := \sum_\gamma P_{\max}^2/2 m_\gamma$.
\end{lemma}

Proof of this result is provided in Appendix \ref{sec:t_error_proof}.  This result will be used to compare the error in the propagator from the discretized Lagrangian path integral and the error from the idealized Hamitonian evolution below.

\begin{theorem}
For any positive integer $r$, let $\hat{K}_r, \hat{V}_r \in \mathcal{L}(\mathcal{H}_{ 2^{n}}^{\otimes \eta D})$ be elements of sequences of Hermitian operators acting on input states $\ket{\psi} \in S_n$ where $S_n$ conforms to Def.~\ref{def:S_n}, and further assume that $x_{\max, \gamma}$ and $\Delta_{x, \gamma}$ are chosen such that for evolution time $T\in \mathbb{R}$, position cutoff $x_{\max}$
    $$
r=\frac{2 \pi T}{m x_{\max, \gamma} \Delta_{x, \gamma}}, x_{\max, \gamma} = \Delta_x 2^n.
    $$
Then, for any $\epsilon>0$, there exists positive integer $r^*$ such that for all $r\ge r^*$, there exists an algorithm that can implement an evolution operator $\hat{W}$ such that the distance measure described in Lemma~\ref{lem:Trotter_Lagrange} obeys
$${\rm D}\left(\hat{W},e^{-i T \left( \hat{K}_r + \hat{V}_r \right)}\right)\le \epsilon$$
using a number of queries to the action oracle defined in Def. \ref{def:O_S} that scale as
    \begin{equation}
    N_{\rm queries} \in O(r^{*})\subseteq O\left( \frac{\eta D T^2 P_{\max}^2  V_{\max}}{\min_\gamma m_\gamma\epsilon}  \right)
\end{equation}
and the number of additional two-qubit gate operations required is
\begin{align}
    N_{\rm gates} \in     O\left( \frac{\eta^2 D^2 P_{\max}^2 T^2}{\min_\gamma m_\gamma \epsilon}  V_{\max}  \left( \log \left( \frac{T ( \max_\gamma x_{\max,\gamma}^{2} P_{\max}^2 \eta D V_{\max}}{\min_\gamma m_\gamma\epsilon} \right) \right)^2 \right).
\end{align}
\end{theorem}
\begin{proof}
From Theorem~\ref{thm:mainLagrangian}, we have that we can implement a unitary operation $\hat{W}$ that is of the form 
\begin{equation}
     \hat{W}=\left( \frac{e^{-i \frac{\pi}{4}}}{\sqrt{2^n}} \right)^{r\eta D} \sum_{\vec{q}_{a}}\sum_{\vec{q}_{b}}\sum_{q_1} \cdots \sum_{\vec{q}_{r-1}}  e^{i \tau \sum_{k=0}^{r-1}\sum_{\gamma=0}^{\eta D-1} \left( \frac{m_\gamma}{2} \left(\frac{x_{k+1,\gamma} - x_{k,\gamma}}{\tau} \right)^2 - V(\vec{x}_k)\right)} \ketbra{\vec{q}_{b}}{\vec{q}_{a}},
\end{equation}
using a number of queries to the action oracle given in Definition~\ref{def:O_S} that scales as $O(r)$.  Using Lemma~\ref{lem:lagrangian}, we see that, up to an irrelevant global phase,
\begin{equation}
    \hat{W} = (e^{-i \hat{K} T/r} e^{-i \hat{V} T/r})^r.
\end{equation}
and in turn Lemma~\ref{lem:Trotter_Lagrange} implies that the value of the distance $D$ restricted to the feasible space $S_n$ is
\begin{equation}
    {\rm D}(\hat{W},e^{-iT (\hat{K} + \hat{V})}) \le \frac{2 \eta D T^2 V_{\max} K_{\max}}{r} 
\end{equation}
As we wish this error measure to be at most $\epsilon$, we can use the scaling of $K_{\max}$ given in Lemma~\ref{lem:Trotter_Lagrange} to see that it suffices to take
\begin{equation}
   r > r^* \in  O\left( \frac{\eta D T^2 K_{\max} V_{\max} }{\epsilon} \right) = O\left( \frac{\eta D T^2 P_{\max}^2 V_{\max} }{\epsilon} \right) .
\end{equation} 

From Theorem~\ref{thm:mainLagrangian}, implementing $\hat{W} = \widetilde{U}^r$ requires $O(r)$ queries. Further, each of the $r$ time steps also requires two iterations of the quantum Fourier transform for each of the $\eta D$ registers, which takes $O(n^2)$ two-qubit gates~\cite{nielsenchuang} for $n$ qubits per time step. 
This needs to be applied to each of the $\eta D$ position registers, which leads to an  additional $O(r \eta D n^2)$ two-qubit gates required for the whole evolution of $r$ time steps. 

The assumptions of Theorem~\ref{thm:mainLagrangian} imply that
\begin{equation}
 r= \frac{2 \pi T}{m_\gamma x_{\max,\gamma} \Delta_{x,\gamma}} = \frac{2 \pi T 2^n}{m_\gamma x_{\max,\gamma}^2} \; \forall \gamma.
\end{equation} 
Solving the above equation for $n$ then yields
\begin{equation}
n = \log_2 \left( \frac{m_\gamma r x_{\max,\gamma}^2}{2 \pi T} \right) \in O\left( \log \left( \frac{T x_{\max,\gamma}^{2} K_{\max} V_{\max}}{\epsilon}\right)\right).
\end{equation}
Thus, an additional
\begin{align}
   O(r \eta D n^2) \subseteq O\left( \frac{\eta D T^2}{ \epsilon} K_{\max} V_{\max}  \left( \log \left( \frac{T x_{\max,\gamma}^{2} K_{\max} V_{\max}}{\epsilon} \right) \right)^2 \right)
\end{align}
two-qubit gates required to implement the quantum Fourier transforms~\cite{nielsenchuang}. 
The final scalings then follow from the observation that $K_{\max} \in O( \eta D P_{\max}^2/ \min_\gamma m_\gamma)$.
\end{proof}

\section{Hamiltonian path integral}
\label{sec:HamPathInt}
In this approach we assume that we are provided a Hamiltonian $\hat{H}$ that is promised to be expressible as a sum of efficiently diagonalizable Hamiltonians. This holds without loss of generality because all finite-dimensional Hamiltonians can be expressed as a sum of elements of the Pauli group, although we can also choose another decomposition if we would like. Specifically, we assume that we have a decomposition that is given by the following definition.
\begin{definition}[Hamiltonian]\label{def:Ham}
    Let $\hat{H}\in \mathcal{L}(\mathcal{H}_{2^n})$ (where $\mathcal{L}(\cdot)$ denotes the set of linear operators acting on a space) be a Hermitian operator that is decomposed in the form $\hat{H} = \sum_{\ell=0}^{L-1} \hat{H}_\ell$ where $\hat{H}_\ell \in \mathcal{L}(\mathcal{H}_{2^n})$ are also Hermitian operators such that $\lambda^{(\ell)}_j$ is the $j^{th}$ eigenvalue in a sorted list of eigenvalues of $\hat{H}_{\ell}$ with corresponding eigenvector $\ket{\chi_j^{(\ell)}}$. Without loss of generality, assume that the first Hamiltonian term $\hat{H}_{0}$ is diagonal in the computational basis; we can choose our computational basis to make this true, or simply let $\hat{H}_0$ be the zero operator.
\end{definition}
In order to make sense of the path integral formulation, we need to consider Trotter-Suzuki decompositions.  We review the definitions of high-order product formulas below.
\begin{definition}[Trotter-Suzuki decompositions] \label{def:trotterdecomp}
    Let $\tau\in \mathbb{R}$ and $\hat{H}$ be defined as above.  The ``first-order" Trotter formula is defined to be
    $$
    \hat{U}_1(\tau) = \prod_{\ell =0}^{L-1} e^{- i \hat{H}_{\ell} \tau},
    $$
    the second-order Trotter-formula is
    $$
    \hat{U}_2(\tau) =  \prod_{\ell = L-1}^{0} e^{-i\hat{H}_\ell \frac{\tau}{2}} \prod_{\ell=0}^{L-1}e^{-i \hat{H}_{\ell} \frac{\tau}{2}},
    $$
    and for any integer $k>1$, $s_k := \frac{1}{4-4^{\frac{1}{k+1}}}$,
    $$
    \hat{U}_{2k}(\tau) = \hat{U}_{2k-2}^2(s_k \tau)\hat{U}_{2k-2}([1-4s_k] \tau)\hat{U}_{2k-2}^2(s_k \tau).
    $$
\end{definition}
Because the first Hamiltonian term $\hat{H}_0$ is defined to diagonal in the computational basis, the last Hamiltonian in the second-order Trotter product will be diagonal as well, since the product formula is symmetric and the first Hamiltonian evolution that is applied is the same as the last one. This will be significant because it implies that the eigenbasis of the last Hamiltonian that is applied is also the computational basis. As the higher-order formulas are found be a symmetric combination of symmetric formulas, it is easy to see that the last Hamiltonian in the $2k^{\rm th}$-order Trotter-Suzuki formula is also diagonal in the computational basis.  Consequently, we do not actually need to know within this formalism what the intermediate bases actually are, just the inner products between the eigenvectors and also the eigenvalues.  These higher-order formulas will be needed in order to achieve near-linear simulation time with this method, as we will see below. 

\begin{lemma}[Trotter error bounds] \label{lemma:trottererror}
Let $r \in \mathbb{Z}^+$, then
$$
\left\|e^{-i\hat{H} t} - \hat{U}_1^r \left( \frac{t}{r} \right) \right\|_\infty \leq  \frac{t^2}{2r} \sum_{\ell=0}^{L-1} \sum_{k=\ell+1}^{L-1} \left\| [\hat{H_k}, \hat{H_\ell} ] \right\|_\infty,
$$
and for higher-order Trotter formulas
\begin{align}
    \left\| e^{-i\hat{H} t} - \hat{U}_{2k}^r \left( \frac{t}{r} \right) \right\|_\infty \le \frac{\alpha_{\rm comm} t^{2k+1}}{r^{2k}}
\end{align}
where
\begin{equation}
    \alpha_{\rm comm} := \sum_{j_0,\ldots,j_{2k}} \left\|[ \hat{H}_{j_0}, [\hat{H}_{j_1},\cdots[ \hat{H}_{j_{2k-1}},\hat{H}_{j_{2k}}]\cdots]] \right\|_\infty.
\end{equation}
\end{lemma}
See \cite{cstwz2021} for proof of Trotter error bounds. From this, we can then proceed to define the Hamiltonian form of the path integral.  Previous work has taken this alternative form of the path integral. However, it has typically only been considered in the case where the lowest order formula is used and in the limit where an infinite number of time slices is considered~\cite{farhi1992functional,cheung2011improved,mackenzie2006perturbative}.  The discrete analogue of this result for the propagator as a finite sum is given below.

\begin{lemma}[Discrete Hamiltonian path integral] \label{lemma:HamPathInt}
Let $\hat{H}$ be a Hamiltonian in accordance with Definition~\ref{def:Ham}, $k, r \in \mathbb{Z}^+$, $M = 2L 5^{k-1}r$ (recall that $L$ is the number of terms in the decomposition of $\hat{H}$), $t\in \mathbb{R}$, and let $\{ \ell_m\}_m$ and $\{\tau_m\}_m$ be sequences (that depend on $r$ and $k$) such that  $\left(\hat{U}_{2k}\left(\frac{t}{r} \right)\right)^r = \prod_{m=0}^{M-1} e^{-i \hat{H}_{\ell_m} \tau_m \frac{t}{r}}$. Then,
$$\left(\hat{U}_{2k}  \left( \frac{t}{r} \right) \right)^r = \sum_{j_0=0}^{2^n-1} \ldots \sum_{j_{M-1}=0}^{2^n-1} \left( 
e^{-i \left( \sum_{m=0}^{M-1} \lambda_{j_m}^{(\ell_m)}\tau_m \right) \frac{t}{r}} \prod_{u=0}^{M-2} \braket{\chi_{j_{u+1}}^{(\ell_{u+1})}}{\chi_{j_u}^{(\ell_u)}} \ketbra{\chi_{j_{u+1}}^{(\ell_{u+1})}}{\chi_{j_u}^{(\ell_u)}} \right) .$$
\end{lemma}
\begin{proof}
Since the sets $\{ \ket{\chi^{(\ell)}_j} \}_j$ each form an orthonormal basis for $\mathcal{H}_{2^n}$ for every $\ell \in \{ 0, \ldots, L-1 \}$, we have
\begin{align}
    \sum_{j=0}^{2^n-1} \ketbra{\chi^{(\ell)}_j}{\chi^{(\ell)}_j} = \openone. \label{1}
\end{align}
For any integer $k$ and real $r$ and $t$, we write the $2k^{\rm th}$ Trotter decomposition in Def. \ref{def:trotterdecomp} as a product of $e^{-i \hat{H}_{\ell_m} \frac{t}{r}}$ terms,
\begin{equation}
    \left(\hat{U}_{2k} \left(\frac{t}{r} \right) \right)^r := \prod_{m=0}^{M-1} e^{-i \hat{H}_{\ell_m} \tau_{m} \frac{t}{r}}, \label{trotter}
\end{equation}
for sequences of Hamiltonian term indices $\{\ell_m\}_m$ and times $\{\tau_m\}_m$ that depend on the Trotter order $k$ and number of time steps $r$ (we will not explicitly indicate the $k$ and $r$ dependence, in order to avoid the notation becoming too messy), where
\begin{equation}
    M = 2L5^{k-1}r
\end{equation}
is the number of unitary factors that the $2k^{\rm th}$-order Trotter formula breaks the original unitary into.

Applying the $2k^{\rm th}$-order Trotter product formula to $e^{-i \hat{H} t}$, and making use of \eqref{1}, we get:
\begin{align}
\left( \hat{U}_{2k} \left( \frac{t}{r} \right) \right)^r
    &= \prod_{m=0}^{M-1} e^{-i \hat{H}_{\ell_m} \tau_{m} \frac{t}{r}} \\
    &=  e^{-i \hat{H}_{\ell_{M-1}} \tau_{M-1} \frac{t}{r}} \prod_{m=0}^{M-2}  \left( \sum_{j_{m+1}=0}^{2^n-1} \ketbra{\chi_{j_{m+1}}^{(\ell_{m+1})}}{\chi_{j_{m+1}}^{(\ell_{m+1})}} e^{-i \hat{H}_{\ell_m} \tau_{m} \frac{t}{r}}\sum_{j_m=0}^{2^n-1} \ket{\chi_{j_m}^{(\ell_m)}}\! \bra{\chi_{j_m}^{(\ell_m)}} \right)   \\
    &=  \sum_{j_{0},\ldots, j_{M-1}} e^{-i \lambda_{j_{M-1}}^{(\ell_{M-1})} \tau_{M-1} \frac{t}{r}}  \prod_{m=0}^{M-2} \left(  \ketbra{\chi_{j_{m+1}}^{(\ell_{m+1})}}{\chi_{j_{m+1}}^{(\ell_{m+1})}}  e^{-i \lambda^{(\ell_m)}_{j_m} \tau_{m} \frac{t}{r}} \ket{\chi_{j_m}^{(\ell_m)}}\! \bra{\chi_{j_m}^{(\ell_m)}} \right)   \\
    &= \sum_{j_{0},\ldots, j_{M-1}} e^{-i \lambda_{j_{M-1}}^{(\ell_{M-1})} \tau_{M-1} \frac{t}{r}}  \prod_{m=0}^{M-2}  \left( e^{-i \lambda^{(\ell_m)}_{j_m} \tau_{m} \frac{t}{r}} \braket{\chi_{j_{m+1}}^{(\ell_{m+1})}}{\chi_{j_m}^{(\ell_m)}} \ketbra{\chi_{j_{m+1}}^{(\ell_{m+1})}}{\chi_{j_m}^{(\ell_m)}} \right) \\
    &=  \sum_{j_{0},\ldots, j_{M-1}}e^{-i \left( \sum_{m=0}^{M-1} \lambda_{j_m}^{(\ell_m)}\tau_m \right) \frac{t}{r}} \prod_{u=0}^{M-2} \braket{\chi_{j_{u+1}}^{(\ell_{u+1})}}{\chi_{j_u}^{(\ell_u)}} \ketbra{\chi_{j_{u+1}}^{(\ell_{u+1})}}{\chi_{j_u}^{(\ell_u)}} 
\end{align}
\end{proof}
\begin{definition}[Path]
    Let $\hat{H}$ be a Hamiltonian in accordance with Definition~\ref{def:Ham}, $k, r \in \mathbb{Z}^+$, $M = 2L 5^{k-1}r$, $t \in \mathbb{R}$, and let $\{\ell_m \}_m$ and $\{\tau_m\}_m$ be sequences such that  $\left(\hat{U}_{2k} \left( \frac{t}{r} \right) \right)^r = \prod_{m=0}^{M-1} e^{-i \hat{H}_{\ell_m} \tau_m \frac{t}{r}}$. A path is a sequence $\left( \ket{\chi_{j_0}^{(\ell_0)}}, \ket{\chi_{j_1}^{(\ell_1)}}, \ldots , \ket{\chi_{j_{M-1}}^{(\ell_{M-1})}} \right)$ where each $j_m \in \{0, \ldots, 2^n-1 \}$.
    We say such a path is computable if there exists a polynomial time algorithm on a quantum Turing machine that can compute the eigenvalues $\lambda_{j}^{(\ell)}$ and the overlaps between the eigenvectors $\braket{\chi_{j_{m+1}}^{(\ell_{m+1})}}{\chi_{j_{m}}^{(\ell_m)}}$.
\end{definition}
With this definition of a path as a sequence of states, the Hamiltonian path integral can be seen as a sum over paths from initial states $\ket{\chi_{j_0}}$ to final states $\ket{\chi_{j_M}}$, where the complex amplitude of each path is $e^{-i \sum_{m=0}^{M-1} \lambda_{j_m}^{(\ell_m)}\tau_m \frac{t}{r}} \prod_{u=0}^{M -2} \braket{\chi_{j_{u+1}}^{(\ell_{u+1})}}{\chi_{j_{u}}^{(\ell_{u})}}$. The total evolution is the sum over all such paths. If we want to find the amplitude of the transition from one particular $j_0$ to one particular $j_{M-1}$, we fix $j_0$ and $j_{M-1}$ and sum over all possible paths from $j_0$ to $j_{M-1}$. 

Lemma~\ref{lemma:HamPathInt} shows that the Hamiltonian evolution can be written as a sum over paths between eigenstates of the terms in the Hamiltonian decomposition. Note that in the path integral representation we only need to know the eigenvalues and overlaps between adjacent eigenvectors of the terms within the Hamiltonian.

\begin{definition}[$d$-sparsity]\label{def:sparse}
Let $\hat{H}$ be a Hamiltonian in accordance with Def. \ref{def:Ham}. We refer to a Hamiltonian decomposition $\hat{H} = \sum_{\ell=0}^{L-1} \hat{H}_\ell$ as $d$-sparse if, for each eigenvector $\ket{\chi_j^{(\ell)}}$ where $0 \leq \ell < L-1$, there are at most $d$ values of $m$ such that $\braket{\chi_m^{(\ell+1)}}{\chi_j^{(\ell)}} \neq 0$, and for each eigenvector $\ket{\chi_j^{(\ell)}}$ where $0 < \ell \leq L-1$, there are at least $d$ values of $m$ such that $\braket{\chi_j^{(\ell)}}{\chi_m^{(\ell-1)}} \neq 0$. \label{def:dsparse}
\end{definition}
In other words, each Hamiltonian term's eigenvectors have at most $d$ non-zero inner products with eigenvectors of adjacent Hamiltonian terms. Any Hamiltonian can always be represented via a $d$-sparse Hamiltonian decomposition for some value of $d \in \{1, \ldots, 2^n\}$, since each term in the Hamiltonian decomposition has $2^n$ eigenvectors, so each term can only have at most $2^n$ non-zero inner products.
Note that here the definition of sparsity that we take for the path integral setting differs from that customarily used for Hamiltonian simulation~\cite{berry2007efficient}.  Here we do not require that the underlying matrices are sparse, but we do require that the matrices can be diagonalized and that the eigenvectors of the underlying matrices have a sparse pattern of overlaps with the two adjacent terms of the Hamiltonian decomposition.

\subsection{Short Time Simulation Algorithm and One-Sparse Simulation} \label{sec:overview}
The algorithm that we propose for the case where we divide our path integral up into a series of short higher-order Trotter transition steps and implement these on a quantum computer proceeds as follows.
\begin{enumerate}
    \item For each operator exponential in $\hat{U}_{2k} \left( \frac{t}{r} \right)$, break the transitions from one eigenbasis given by the path integral sum in Lemma~\ref{lemma:HamPathInt} to the next into a sum of $O(Ld^2)$ weighted signature matrices using a distributed graph colouring algorithm.
    \item Use the linear combinations of unitaries (LCU) method to implement the path integral over the short time step via queries to the oracles to obtain the inner products and phases between the eigenstates.
    \item Use alternating sign trick with the LCU method and robust oblivious amplitude amplification (ROAA) to implement transitions with the correct magnitude.
    \item Repeat for all $M$ operator exponentials in the $2k^{\rm th}$- order Trotter decomposition.
\end{enumerate}

Instead of actually using the $\hat{H}_\ell$ eigenstates in the computation, we will use computational basis states to represent the eigenstates. This way, we can easily define oracles and other operations that can be controlled on the states. We are implementing the transitions one time step at a time. At the $m^{\rm th}$ step, we transition from the $\hat{H}_{\ell_m}$ basis to the $\hat{H}_{\ell_{m+1}}$ basis. We will need $\log M$ qubits to store the step number $M$, and one qubit to indicate whether the state is in the first or second basis of that step. Each eigenbasis set is enumerated from 0 to $2^n-1$, which can be stored in $n$ qubits. At the $m^{\rm th}$ step, we represent the eigenstate $\ket{\chi^{(\ell_m)}_j}$ by the computational basis state $\ket{m} \ket{0} \ket{j}$ and $\ket{\chi^{(\ell_{m+1})}_q}$ by $\ket{m} \ket{1} \ket{q}$.

\subsection{Definitions and Oracles}

For the order $2k$ Trotter product formula, we have a sequence $\{\ell_m\}_m$ which specifies the Hamiltonian term number that is being implemented at the $m^{\rm th}$ step. This sequences are efficiently computable, so we will define an oracle that will give us the proper value of $\ell_m$ at the $m^{\rm th}$ step. 

Since there are $d$ non-zero inner products between the $\hat{H}_{\ell_m}$ and $\hat{H}_{\ell_{m+1}}$ eigenstates, we would like to define a way of indexing them. Let us label the eigenstates of $\hat{H}_{\ell_m}$ and $\hat{H}_{\ell_{m+1}}$with indices $j$ and $q$, respectively, where $j, q \in \{0, \ldots, 2^n-1 \}$. For each eigenstate $\ket{\chi^{\ell_m}_j}$ of $\hat{H}_{\ell_m}$, we enumerate the non-zero inner products $\braket{\chi^{\ell_{m+1}}_q}{\chi^{\ell_m}_j}$ by increasing $q$ (with enumeration starting at 0), similarly for the $\ket{\chi^{\ell_{m+1}}_q}$ states. If there are fewer than $d$ non-zero inner products, we will include some states with zero inner product in the enumeration so that there are exactly $d$ $\ket{\chi_q^{(\ell_{m+1})}}$ states for each $\ket{\chi_j^{(\ell_m)}}$ state, and vice versa. We will set the amplitude of these transitions to zero later, but we want there to be exactly $d$ transitions per state at each time step so that we have a consistent definition of the index function. The following index function takes $m$ (the time step number), $b$ (indicates whether the input state is the first or second basis in the transition), $j$ (the index of the eigenstate), and an index $p$, and returns $\hat{H}_{\ell_m}$ or $\hat{H}_{\ell_{m+1}}$ eigenstate that gives the $p^{\rm th}$ non-zero inner product.

\begin{definition}[Index function] \label{def:f_ind}
    Let $f_{ind} : \{ 0, \ldots, M-1\} \times \{0, 1\} \times \{0, \ldots, 2^n-1\} \times \{0,\ldots, d-1 \} \rightarrow  \{0,\ldots, 2^n-1\}$, $f_{ind}(m, b, j, p)$ is defined to be the index $q$ such that $\braket{\chi^{(\ell_{m + b})}_q}{\chi^{(\ell_{m + (b \oplus 1)})}_j}$ is the $p^{\rm th}$ non-zero inner product for $\ket{\chi^{(\ell_{m + b})}_j}$ (with indexing starting at 0), based on the enumeration system described above.
\end{definition}
As a simple example of the index function, consider a two-qubit Hamiltonian decomposition with sparsity $d=2$ where $\hat{H}_{\ell_m} = \hat{Z} \otimes \hat{Z}$ with ordered eigenstates $\{ \ket{0} \ket{0}, \ket{0} \ket{1}, \ket{1} \ket{0}, \ket{1} \ket{1} \}$, and $\hat{H}_{\ell_{m+1}} = \hat{Z} \otimes \hat{X}$ with ordered eigenstates $\{ \ket{0} \ket{+}, \ket{0} \ket{-}, \ket{1} \ket{+}, \ket{1}, \ket{-} \}$. $\ket{\chi_2^{(\ell_m)}} = \ket{1}\ket{0}$ has non-zero inner products with $\ket{\chi_2^{(\ell_{m+1})}} = \ket{1} \ket{+}$ and $\ket{\chi_3^{(\ell_{m+1})}} = \ket{1} \ket{-}$. Then, $f_{ind}(m, 0, 2, 0) = 2$ (since $\braket{\chi_2^{(\ell_{m+1})}}{\chi_2^{(\ell_m)}}$ is the 0$^{\rm th}$ non-zero inner product for $\ket{\chi_2^{(\ell_m)}}$), and $f_{ind}(m, 0, 2, 1) = 3$ (since $\braket{\chi_3^{(\ell_{m+1})}}{\chi_2^{(\ell_m)}}$ is the 1$^{\rm st}$ non-zero inner product for $\ket{\chi_2^{(\ell_m)}}$).

Now that we have defined the index function, we can define the index oracle $O_{ind}$, which takes $m$ (the time step number), $b$ (indicates whether the input state is the first or second basis in the transition), $j$ (the index of the eigenstate), and $p$, and outputs the value of $f_{ind}(m, b, j, p)$ in another register. 
\begin{definition}[Index oracle] \label{def:O_ind} Let $O_{ind} \in \mathcal{L} (\mathcal{H}_{M} \otimes \mathcal{H}_{2} \otimes \mathcal{H}_{2^n} \otimes \mathcal{H}_{d} \otimes \mathcal{H}_{2^n})$ be a unitary operator, such that its action on an arbitrary computational basis state is
    \begin{align}
        O_{ind} : \ket{m} \ket{b} \ket{j} \ket{p} \ket{c} \mapsto \ket{m} \ket{b} \ket{j} \ket{p}  \ket{c \oplus f_{ind}(m, b, j, p)},
    \end{align}
\end{definition}

Now we will define the inner product oracles that give the magnitude and the phase of the inner product between two states $\braket{\chi^{(\ell_{m+1})}_q}{\chi^{(\ell_m)}_j}$. The magnitude of the inner product is between 0 and 1, so we can multiply it by some constant $2^B$ to make use of $B$ qubits to store the value of the inner product as a number from 0 and $2^B-1$, up to $\frac{1}{2^B}$ precision. The inner product magnitude oracle takes as input three indices $m$, $j$, $q$, corresponding to two states $\ket{\chi^{(\ell_m)}_j}$ and $\ket{\chi^{(\ell_{m+1})}_m}$, and outputs the magnitude of the inner product (multiplied by a constant) as a bitstring in a fourth register. The inner product phase oracle takes the same three input indices as the magnitude oracle, and adds the corresponding phase of the inner product to the state.
\begin{definition}[Inner product magnitude oracle]\label{def:ipmag}
Let $B \in \mathbb{N}$. Let $O_{IM} \in \mathcal{L}(\mathcal{H}_{ M} \otimes \mathcal{H}_{2^n} \otimes \mathcal{H}_{2^n} \otimes \mathcal{H}_{2^B} )$ be a unitary operator, such that its action on an arbitrary computational basis state is
    \begin{align}
    O_{IM}: & \ket{m} \ket{j} \ket{q} \ket{c} \mapsto \ket{m} \ket{j} \ket{q} \ket{c \oplus \left[2^B \left| \braket{ \chi_{q}^{(\ell_{m+1)}} }{ \chi_j^{(\ell_m)}} \right| \right]_B},
    \end{align}
    where $[ \;\cdot\; ]_B$ denotes the bitstring representation in $B$ bits. \label{def:O_IM}
\end{definition}
\begin{definition}[Inner product phase oracle] \label{def:ipphase}
Let $O_{IP} \in \mathcal{L}(\mathcal{H}_{M} \otimes \mathcal{H}_{2^n} \otimes \mathcal{H}_{2^n})$ be a unitary operator, such that its action on an arbitrary computational basis state is
    \begin{align}
        O_{IP}: & \ket{m} \ket{j} \ket{q} \mapsto e^{i \arg \left( \braket{ \chi_{q}^{(\ell_{m+1)}} }{ \chi_j^{(\ell_m)}}  \right)} \ket{m} \ket{j} \ket{q}
    \end{align}
\end{definition}
Finally, we will define an oracle that implements the eigenvalue phase associated with each time step. The phase oracle takes a time step number $m$ and an index $j \in \{0,\ldots,2^n-1\}$, and adds a phase with the eigenvalue of the $j^{\rm th}$ eigenstate of $\hat{H}_{\ell_m}$, multiplied by the $\tau_m \frac{t}{r}$.
\begin{definition}[Eigenvalue phase oracle]
Let $O_{EP} \in \mathcal{L}(\mathcal{H}_{\log M} \otimes \mathcal{H}_{n} )$ be a unitary operator, such that its action on an arbitrary computational basis state is
    \begin{align}
        O_{EP} : \ket{m} \ket{j} \mapsto e^{-i \lambda^{(\ell_m)}_j \tau_m \frac{t}{r} } \ket{m} \ket{j}
    \end{align}
\end{definition}

\subsection{Graph Colouring} \label{sec:graphcolouring}
Recall that the Trotterized time evolution operator that we are trying to implement is
\begin{align}
\sum_{j_0=0}^{2^n-1} \ldots \sum_{j_{M-1}=0}^{2^n-1} \left( 
e^{-i \left( \sum_{m=0}^{M-1} \lambda_{j_m}^{(\ell_m)}\tau_m \right) \frac{t}{r}} \prod_{u=0}^{M-2} \braket{\chi_{j_{u+1}}^{(\ell_{u+1})}}{\chi_{j_u}^{(\ell_u)}} \ketbra{\chi_{j_{u+1}}^{(\ell_{u+1})}}{\chi_{j_u}^{(\ell_u)}} \right).
\end{align}
The Trotter product formula splits our unitary time evolution operator into $M$ transitions; the unitary that implements $m^{\rm th}$ transition is
\begin{align}
    \sum_{j=0}^{2^n-1} \sum_{q=0}^{2^n-1} e^{-i \lambda^{(\ell_m)}_{j} \tau_{m} \frac{t}{r}} \braket{\chi_{q}^{(\ell_{m+1})}}{\chi_{j}^{(\ell_m)}} \ketbra{\chi_{q}^{(\ell_{m+1})}}{\chi_{j}^{(\ell_m)}}. \label{timestep}
\end{align} 
If we write this in terms of the actual transitions that we are implementing on the quantum computer, in the computational basis, then, as the idealized evolution operator, we have
\begin{align}
    \hat{A}_m&:= \sum_{j=0}^{2^n-1} \sum_{q=0}^{2^n-1} \left|\braket{\chi^{(\ell_{m+1})}_q}{\chi^{(\ell_m)}_j} \right| e^{i \arg \left( \braket{ \chi_{q}^{(\ell_{m+1)}} }{ \chi_j^{(\ell_m)}}  \right) - i \lambda^{(\ell_m)}_j \tau_m \frac{t}{r}} \ket{q}\bra{j} \label{A_m}
\end{align}
The above transition is unitary, since it is just the original transition with changes of basis to the computational basis. The individual transitions $\ket{q} \bra{j}$  are not themselves unitary, but we can combine them in ways that form unitaries that we can implement using our oracles. In particular, permutation matrices are unitary, so we can form unitaries by grouping the operators so that within each unitary, each $\ket{\chi_j^{(\ell_m)}}$ state is uniquely mapped to some $\ket{\chi_q^{(\ell_{m+1})}}$ state, and vice versa. We do not need to consider the transition elements $\ketbra{q}{j}$ where $\braket{\chi^{(\ell_{m+1})}_q}{\chi^{(\ell_m)}_j} = 0$, since the amplitude of those transitions is zero. Thus, we will have up to $2^n d$ total transitions that we need to group together to form permutation matrices, and then we will deal with the relative amplitudes after that.

The grouping of the transitions can be viewed as a graph edge-colouring problem. We can think of the states $\{\ket{\chi^{(\ell_m)}_j}\}_j, \{ \ket{\chi^{(\ell_{m+1})}_q}\}_q$ as vertices, with edges between two states if the transition amplitude $\braket{\chi^{(\ell_{m+1})}_q}{\chi^{(\ell_m)}_j}$ is non-zero. In order to sort these transitions into different unitaries, we do an edge colouring so that no two edges adjacent to the same vertex have the same colour. This is equivalent to saying that within a single colour, each $\ket{\chi^{(\ell_m)}_j}$ to get mapped to a unique $\ket{\chi^{(\ell_{m+1})}_q}$, and each $\ket{\chi^{(\ell_{m+1})}_q}$ to only has one $\ket{\chi^{(\ell_m)}_j}$ that maps to it. All the edges of a given colour, taken together, will then correspond to a unitary permutation matrix. Since there are at most $d$ non-zero inner products $\braket{\chi_j^{(\ell_m)}}{\chi^{(\ell_{m+1})}_q}$ for each $j$ and for each $m$, the degree of the graph is at most $d$. The graph will be bipartite since edges only exist between $\ket{\chi^{(\ell_m)}_j}$ states and $\ket{\chi^{(\ell_{m+1})}_q}$ states, but not between two $\ket{\chi^{(\ell_m)}_j}$ states or two $\ket{\chi^{(\ell_{m+1})}_q}$ states. 
\begin{definition} \label{def:G_m}
Let $G_m$ be a bipartite graph, defined as follows:
The vertex set of $G_m$ is $V(G_m) := \{0, 1\} \times \{0,...,2^n-1\}$. The vertex sets of the two parts of the bipartite graph are $\Gamma_{m,0} := \{0\} \times \{0,...,2^n-1\}$ and $\Gamma_{m,1} := \{1\} \times \{0, ..., 2^n-1\}$. The edge set of $G_m$ is defined as $E(G_m) = \{ \big( (0, j), (1, q) \big) \in \Gamma_{m,0} \times \Gamma_{m,1} : \braket{\chi^{(\ell_{m+1})}_q}{\chi^{(\ell_m)}_j} \neq 0 \}$. 
\end{definition}
\begin{figure}
    \includegraphics[width=0.6\textwidth]{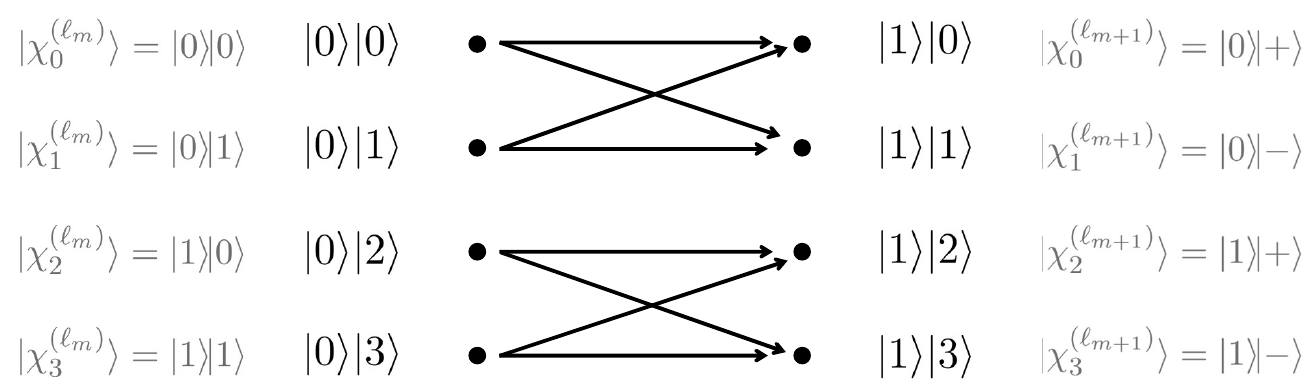}
    \caption{
    Bipartite graph $G_m$ for the example system described previously ($\hat{H}_{\ell_m} = \hat{Z} \otimes \hat{Z}$ with ordered eigenstates $\{ \ket{0} \ket{0}, \ket{0} \ket{1}, \ket{1} \ket{0}, \ket{1} \ket{1} \}$ and $\hat{H}_{\ell_{m+1}} = \hat{Z} \otimes \hat{X}$ with ordered eigenstates $\{ \ket{0} \ket{+}, \ket{0} \ket{-}, \ket{1} \ket{+}, \ket{1}, \ket{-} \}$). The four vertices on the left correspond to the basis states of $\hat{H}_{\ell_m}$ and the four vertices on the right correspond to the basis states of $\hat{H}_{\ell_{m+1}}$. The basis states are labelled in grey, while the labels written in black are how they are stored as computational basis states on the quantum computer. There is an edge between two vertices $\ket{0} \! \ket{j}$ and $\ket{1} \! \ket{q}$ if the inner product $\braket{\chi_q^{(\ell_{m+1})}}{\chi_j^{(\ell_m)}}$ is non-zero. Here we have added arrows to the edges to indicate the transition we are implementing goes from the states on the left to the states on the right.} \label{fig:graph1}
\end{figure}
We associate the vertices in $\Gamma_{m,0}$ with the basis vectors of $\hat{H}_{\ell_m}$ and the vertices in $\Gamma_{m,1}$ with the basis vectors of $\hat{H}_{\ell_{m+1}}$; that is, each vertex $(0, j) \in \Gamma_m$ is represents the $\hat{H}_{\ell_m}$ eigenstate $\ket{\chi^{(\ell_m)}_j}$, and each vertex $(1, q) \in \Gamma_{m,1}$ represents the $\hat{H}_{\ell_{m+1}}$ eigenstate $\ket{\chi^{(\ell_{m+1})}_q}$.
There is an edge between a vertex $(0, j) \in \Gamma_{m,0}$ and a vertex $(1, q) \in \Gamma_{m,1}$ if and only if the inner product $\braket{\chi^{(\ell_{m+1})}_q}{\chi^{(\ell_m)}_j}$ is non-zero.

Since there are at most $d$ non-zero inner products for each basis vector of $\hat{H}_{\ell_m}$ or $H_{\ell_{m+1}}$ (the $d$-sparse definition in Def. \ref{def:dsparse}), the degree of $G_m$ is at most $d$. The optimal number of colours in a bipartite graph edge-colouring is $d$ for a graph of degree $d$; however, the best known algorithm \cite{cos2001} to actually find this colouring takes $O(|E| \log d)$ time, where $|E|$ is the cardinality of the edge set. Since we have up to $2^n d$ edges, this scaling is not ideal. Instead, we will use a colouring that takes $d^2$ colours, but only takes a constant number of oracle queries.
\begin{lemma}[Graph colouring scheme] \label{def:graphcolouring}
    Let $G_m$ be the graph defined in Definition \ref{def:G_m}. The colour associated with an edge $\big( (0, j), (1, q) \big)$ is a label $(c_1, c_2) \in \{0, \ldots, d-1\}^2$, defined as follows: $c_1\left( m, \big( (0, j), (1, q) \big) \right) \in \{0,...,d-1\}$ is the index such that $f_{ind}(m, 0, j, c_1) = q$ and $c_2\left( m, \big( (0, j), (1, q) \big) \right) \in \{0,...,d-1\}$ is the index such that $f_{ind}(m, 1, q, c_2) = j$ (recall that $f_{ind}$ is the index function, as defined in Definition \ref{def:f_ind}). This is a valid edge colouring of $G_m$.
\end{lemma}
\begin{proof}
    In order for this to be a valid edge colouring, each edge in $G_m$ must be assigned exactly one colour, and no two edges adjacent to the same vertex have the same colour.
    
    Take any edge $\big( (0,j),(1,q) \big) \in E(G_m)$. The edges adjacent to $(0,j)$ represent non-zero inner products between the state $\ket{\chi^{(\ell_m)}_j}$ and $\hat{H}_{\ell_{m+1}}$ eigenstates. Each non-zero inner product for $\ket{\chi^{(\ell_m)}_j}$ is enumerated a different number from $0$ to $d-1$, so there is exactly one value $c_1 \in \{0, \ldots, d-1\}$ for which $\braket{\chi^{(\ell_{m+1})}_q}{\chi^{(\ell_m)}_j}$ is the $c_1^{\rm th}$ non-zero inner product for $\ket{\chi^{(\ell_m)}_j}$. All the non-zero inner products for ${\chi^{(\ell_{m+1})}_q}$ are also enumerated similarly, so for every $\big( (0,j),(1,q) \big) \in E(G_m)$ there is exactly one $c_2 \in \{0, \ldots, d-1\}$ such that $\braket{\chi^{(\ell_{m+1})}_q}{\chi^{(\ell_m)}_j}$ is the $c_2^{\rm th}$ non-zero product for ${\chi^{(\ell_{m+1})}_q}$. Thus, every edge in $E(G_m)$ has one ``colour'' $(c_1, c_2) \in \{0, \ldots, d-1\}^2$ assigned to it.

    Assume, for the sake of contradiction, that we have two different edges $\big( (0,j_1), (1, q_1) \big)$ and $\big( (0,j_2), (1, q_2) \big)$ with the same colour $(c_1, c_2)$, adjacent to the same vertex, i.e., $(j_1 = j_2) \oplus (q_1 = q_2) = 1$. If $j_1 = j_2$, then $f_{ind}(m,0,j_1,c_1) = f_{ind}(m,0,j_2,c_1)$, but that would mean $q_1 = q_2$, which is a contradiction. Similarly, if $q_1 = q_2$, then $f_{ind}(m,1,q_1,c_2) = f_{ind}(m,1,q_2,c_2)$, which would mean $j_1 = j_2$. Thus, two different edges cannot have the same colour if they share one common vertex.
\end{proof}
\begin{figure}
    \includegraphics[width=0.6\textwidth]{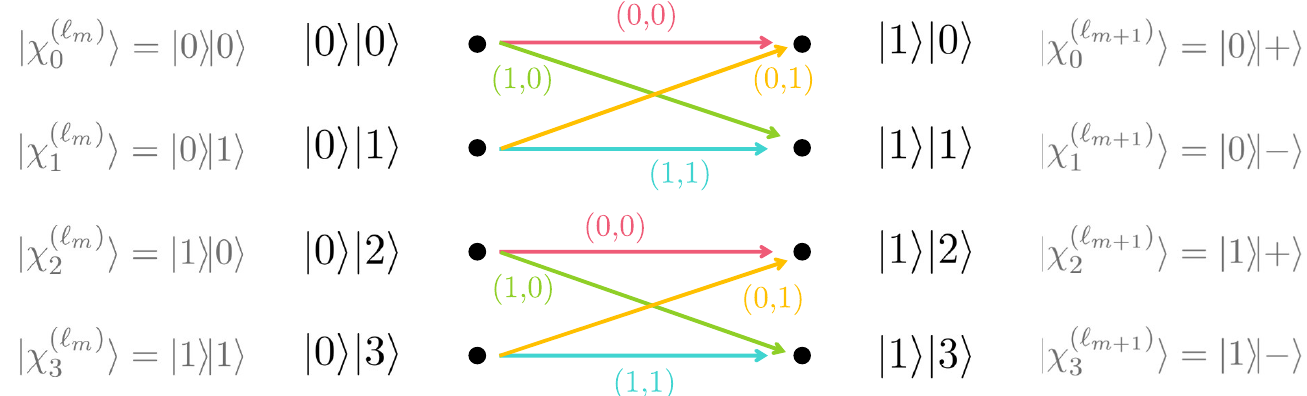}
    \caption{Graph colouring for the example graph in Fig. \ref{fig:graph1}. The labels on each edge indicate the colour that the graph colouring scheme defined in \ref{def:graphcolouring} assigns to that edge.} \label{fig:graph2}
\end{figure}
Since each pair $(c_1, c_2) \in \{0, \ldots, d-1\}^2$ corresponds to a unique colour, we have $d^2$ colours in total.
\begin{definition}[Colour subgraph]
    Let $G_m^{(c_1,c_2)}$ be the subgraph of $G_m$ where $E\left(G_m^{(c_1,c_2)}\right)$ is the set of edges of colour $(c_1,c_2)$ (following the graph colouring scheme described in Lemma \ref{def:graphcolouring}) and $V\left(G_m^{(c_1,c_2)}\right)$ is the set of vertices that the edges in $E\left(G_m^{(c_1,c_2)}\right)$ are incident upon.
\end{definition}
\begin{figure}
    \subfloat{\includegraphics[width=0.3\textwidth]{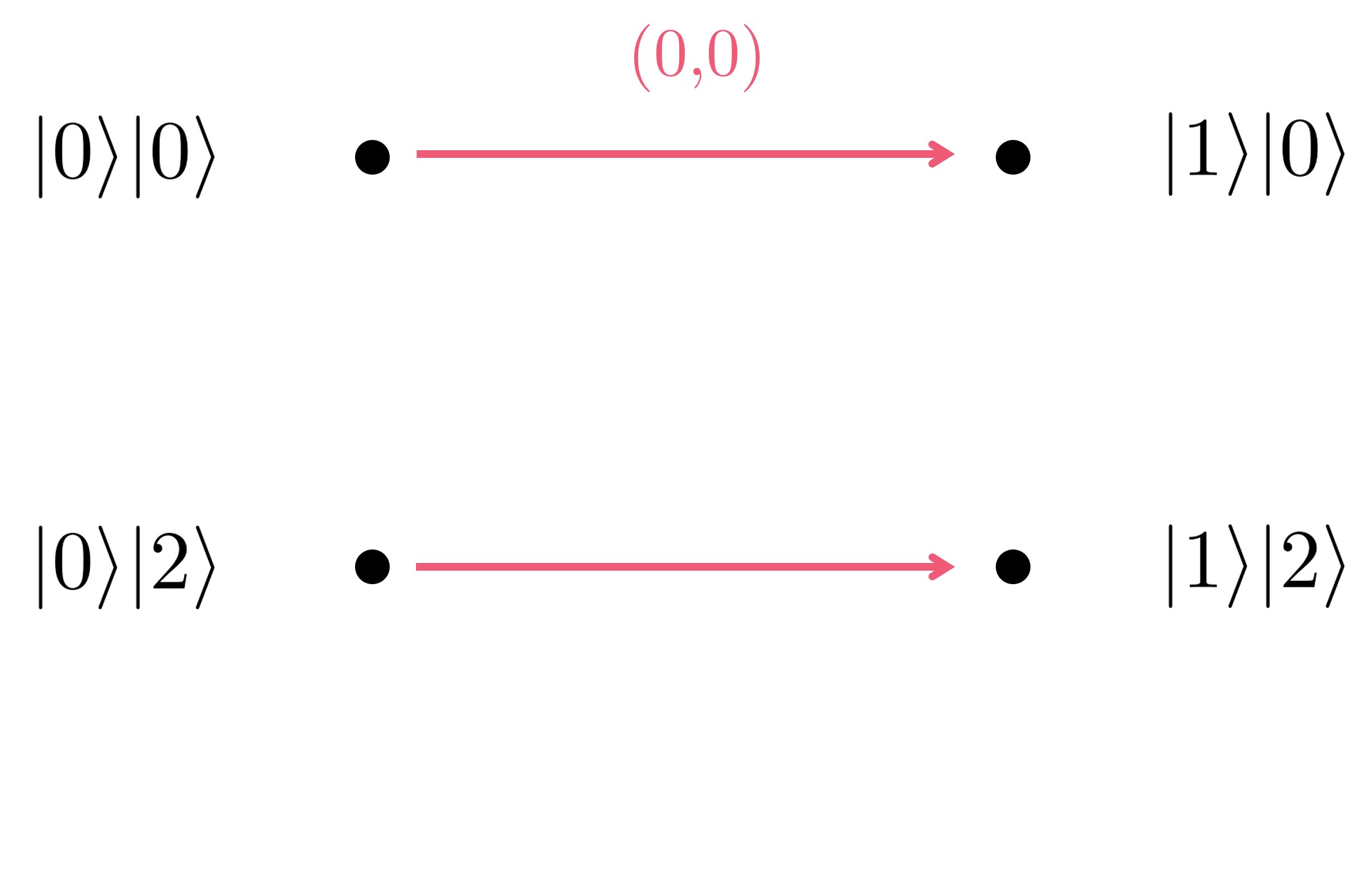}} \hspace{3 em}
    \subfloat{\includegraphics[width=0.3\textwidth]{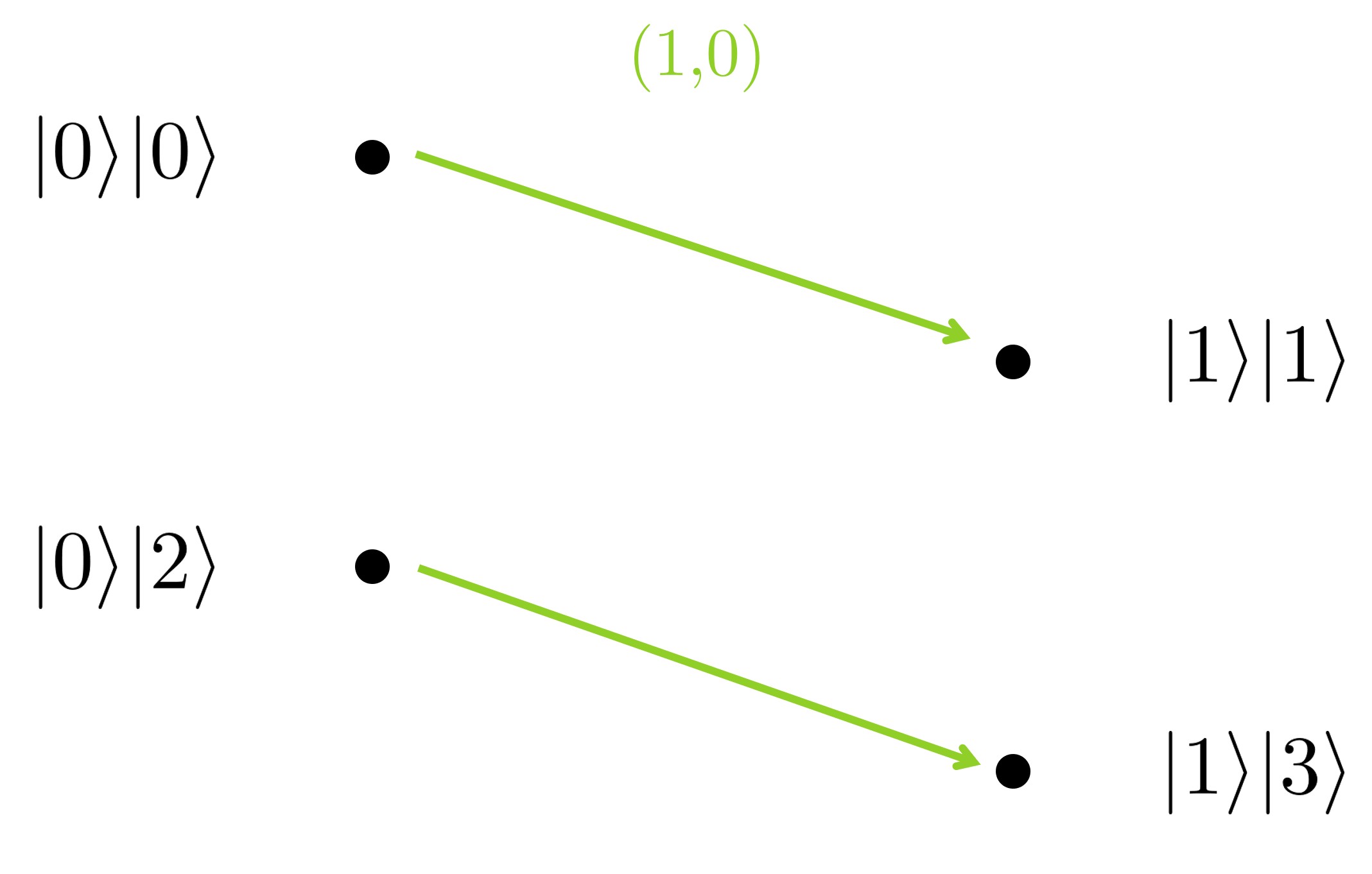}} \\
    \subfloat{\includegraphics[width=0.3\textwidth]{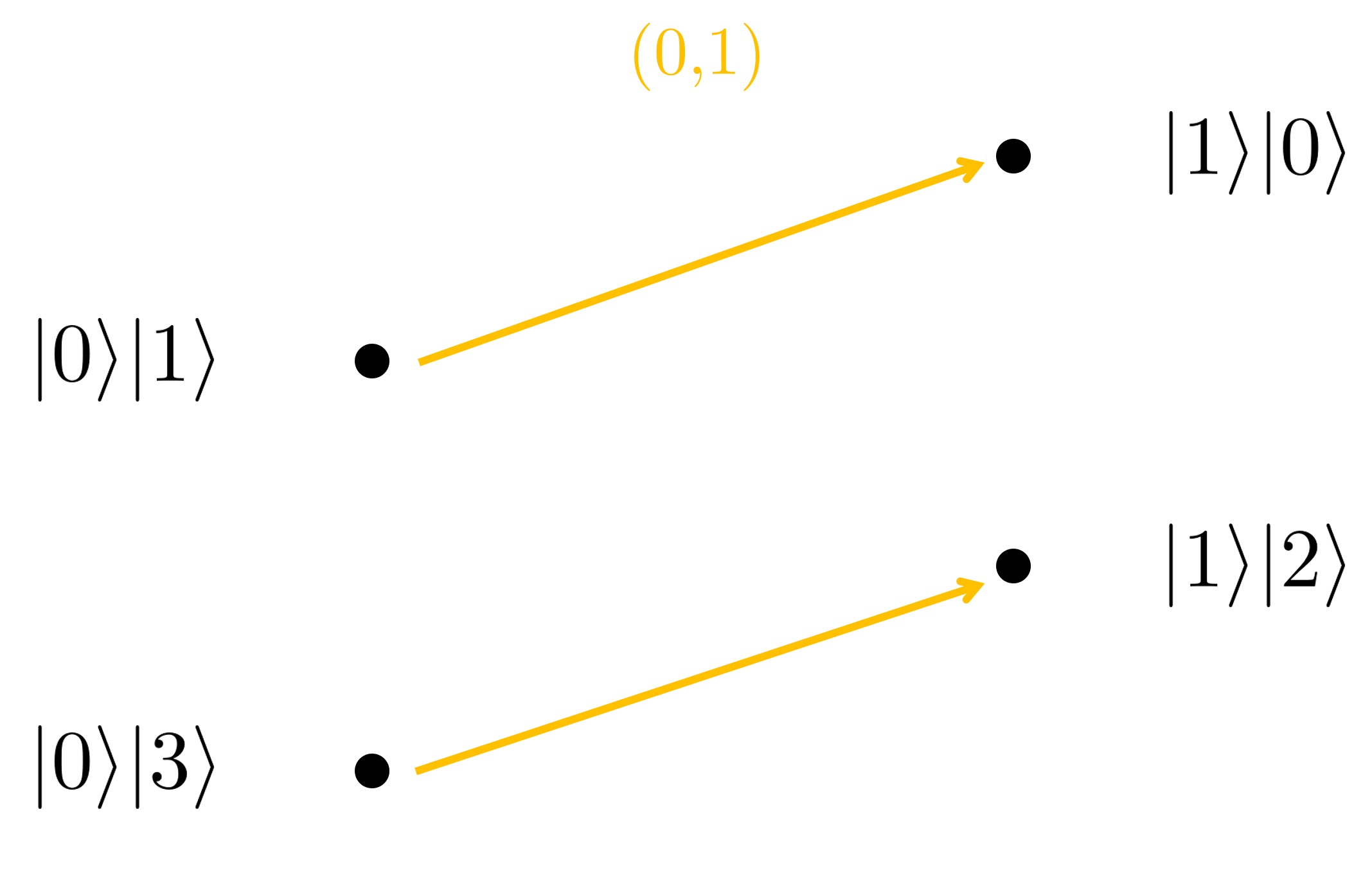}} \hspace{3 em}
    \subfloat{\includegraphics[width=0.3\textwidth]{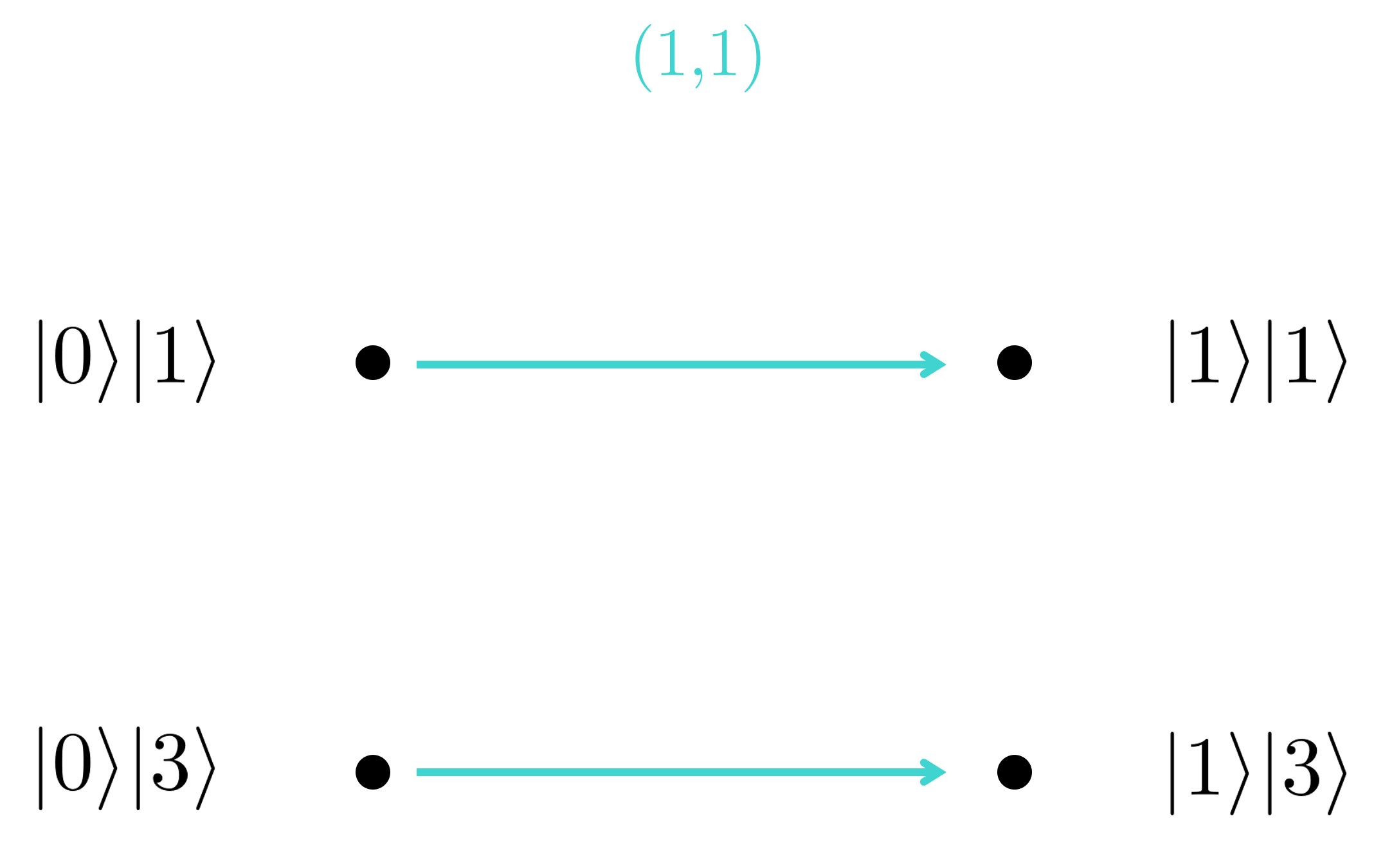}}
    \caption{Colour subgraphs for graph colouring in Fig. \ref{fig:graph2}.}
\end{figure}
The adjacency matrix corresponding to each $G_m^{(c_1, c_2)}$ is a permutation matrix, which is unitary. The edges of each $G_m^{(c_1,c_2)}$ indicate what transitions are present in the permutation matrix corresponding to the colour $(c_1, c_2)$. In the case where a vertex does not appear in $V(G_m^{(c_1,c_2)})$, we will map that vertex back to itself, which ensures that the operator acting on the entire space of states is still a permutation matrix. Later, we will set the amplitude of such transitions to zero, so that they are effectively not implemented. 

We will define a colour oracle $O_C$, which checks checks to see if a vertex $(b, j)$ is in a particular colour graph $G_m^{(c_1, c_2)}$. The colour oracle takes vertex $(b, j) \in \{0,1\} \times \{0, \ldots, 2^n-1\}$ and a colour $(c_1, c_2) \in \{0, \ldots , d-1 \}^2$ as input, and flips the value of another qubit based on whether vertex $(b, j)$ is a vertex in the graph of that colour.
\begin{definition}[Colour oracle] \label{def:O_C}
Let $O_C \in \mathcal{L} (\mathcal{H}_{\log L} \otimes \mathcal{H}_n \otimes \mathcal{H}_{\log d} \otimes \mathcal{H}_{\log d} \otimes \mathcal{H}_1)$ be a unitary operator, such that its action on an arbitrary computational basis state is
    \begin{align}
        O_C &: \ket{m} \ket{b} \ket{j} \ket{c_1} \ket{c_2} \ket{c} \mapsto \begin{cases}
        \ket{m} \ket{b} \ket{j} \ket{c_1} \ket{c_2} \ket{c \oplus 1} & \text{if } (b, j) \in V(G_m^{(c_1,c_2)}) \\
        \ket{m} \ket{b} \ket{j} \ket{c_1} \ket{c_2} \ket{c} & \text{if } (b,j) \notin V(G_m^{(c_1,c_2)}) 
        \end{cases}
    \end{align}
\end{definition}
The colour oracle $O_C$ defined above can be built 8 queries to the index oracle $O_{ind}$ (Definition \ref{def:O_ind}) and $n+1$ additional qubits (which can be reused), as shown in Figure~\ref{fig:oraclecircuit}.

Since every non-zero transition from a $\ket{\chi^{(\ell_m)}_j}$ state to a $\ket{\chi^{(\ell_{m+1})}_q}$ state is assigned exactly colour $(c_1, c_2)$, summing over all the colours will implement every transition. We can write the time evolution operator being implemented at the $m^{\rm th}$ time step, $\hat{A}_m$ (defined in \ref{A_m}), as the sum over all the colours, with the proper associated phase added to each transition:
\begin{align}
    \hat{A}_m &= \sum_{c_1=0}^{d-1} \sum_{c_2=0}^{d-1} \sum_{\substack{j: (0,j) \in \\ V(G_m^{(c_1, c_2)})}} \left|\braket{\chi^{(\ell_{m+1})}_{f_{ind}(m, 0, j, c_1)}}{\chi^{(\ell_m)}_j} \right| e^{i \arg \left( \braket{ \chi_{f_{ind}(m, 0, j, c_1)}^{(\ell_{m+1)}} }{ \chi_j^{(\ell_m)}}  \right) - i \lambda^{(\ell_m)}_j \tau_m \frac{t}{r}} \ketbra{f_{ind}(m, 0, j, c_1)}{j} \label{A_m2}
\end{align}

\subsection{Unitary Decomposition for Paths}

Each transition $\ketbra{q}{j} $ within the time evolution operator for one time step has a corresponding amplitude $\left| \braket{\chi^{(\ell_{m+1})}_q}{\chi^{(\ell_m)}_j} \right|$. We can approximate the relative amplitudes using a technique we call the ``alternating sign trick.'' The idea is that instead of simply applying the permutation matrix corresponding to a colour once in the summation, we will add each operator $2^B$ times, but with a $(-1)^b$ phase to the $b^{\rm th}$ $\ketbra{q}{j}$ transition if $b > 2^B \left|\braket{\chi^{(\ell_{m+1})}_q}{\chi^{(\ell_m)}_j} \right|$. The $(-1)^b$ phase does not affect the unitarity of the operators we are implementing. However, the alternating $+1$ and $-1$ terms will cancel each other out, giving an effective amplitude of approximately $2^B \left| \braket{\chi^{(\ell_{m+1})}_q}{\chi^{(\ell_m)}_j} \right|$ on the $\ketbra{q}{j} $ part of the transition. The same effect will be used for the transitions that have an amplitude of zero, to effectively cancel them out completely. We are taking $2^B$ terms for the alternating sign trick, where $B$ is the number of bits in which we store the inner product magnitudes $\left| \braket{\chi^{(\ell_{m+1})}_q}{\chi^{(\ell_m)}_j} \right|$, since the precision with which we can approximate $\hat{A}_m$ is limited by the inner product precision. 
\begin{definition} \label{def:U}
    Let $\hat{U}_{m,b,c_1,c_2} \in \mathcal{L}(\mathcal{H}_2 \otimes \mathcal{H}_{2^n})$ be a unitary operator, and let $g$ be a function such that
    \begin{align}
        g(m, j, c_1, b) &= \begin{cases}
            1 & b < \left[ 2^B \left| \braket{\chi^{(\ell_{m+1})}_{f_{ind}(\ell_m, j, \ell_{m+1}, c_1)}}{\chi^{(\ell_m)}_j} \right| \right]_B \\
            (-1)^b & \text{otherwise}
        \end{cases},
    \end{align}
    then we define
    \begin{align}
        \hat{U}_{m,b,c_1,c_2} &:= \sum_{\substack{j: (0,j) \in \\ V(G_m^{(c_1, c_2)})}} \bigg( g(m, j, c_1, b) e^{i \arg \left( \braket{ \chi_{f_{ind}(m, 0, j, c_1)}^{(\ell_{m+1)}} }{ \chi_j^{(\ell_m)}}  \right) - i \lambda^{(\ell_m)}_j \tau_m \frac{t}{r}} \ketbra{1}{0} \otimes \ketbra{f_{ind}(m, 0, j, c_1)}{j} \notag \\
        & \quad + (-1)^b \ketbra{0}{1} \otimes \ketbra{j}{f_{ind}(m, 0, j, c_1)} \bigg) + \sum_{\substack{j: (0,j) \notin \\ V(G_m^{(c_1, c_2)})}} (-1)^b \ketbra{0}{0} \otimes \ketbra{j}{j} + \sum_{\substack{q:(1, q) \notin \\ V(G_m^{(c_1,c_2)})}} (-1)^b \ketbra{1}{1} \otimes \ketbra{q}{q}. \label{U_mbc}
    \end{align}
\end{definition}
The above operator maps a state $\ket{0}\ket{j}$ to the state $\ket{1}\ket{q}$ with the corresponding phase if the edge colour between the corresponding two vertices $(0, j)$ and $(1, q)$ is $(c_1,c_2)$, possibly with another $-1$ phase. The operator also maps $\ket{1} \ket{q}$ back to $\ket{0} \ket{j}$ if the corresponding edge colour is $(c_1, c_2)$, with a $(-1)^b$ phase. If a vertex does not have any edges of colour $(c_1,c_2)$, the operator simply acts as the identity, with a $(-1)^b$ phase, on the corresponding state, mapping that state back to itself. $\hat{U}_{m,b,c_1,c_2}$ is a signed permutation operator on the set $\{0,1\} \times \{0,\ldots, 2^n-1\}$, which is unitary. Note that the input state at the $m^{\rm th}$ time step will be of the form $\ket{0} \ket{\psi}$, with the 0 in the first register indicating that $\ket{\psi}$ is written in the $\hat{H}_{\ell_m}$ basis. The second and fourth terms in $\ref{U_mbc}$ will not be relevant when acting on the input state, since the $\ketbra{0}{1}$ and $\ketbra{1}{1}$ acting on the first qubit will give zero. Additionally, once those terms are summed over $b$, the $(-1)^b$ phase will cancel them out anyway. However, we include them in the definition of $\hat{U}_{m,b,c_1,c_2}$ to keep the entire operator unitary.

Let
\begin{align}
    \widetilde{A}_m :&= \sum_{b=0}^{2^B-1} \sum_{c_1=0}^{d-1} \sum_{c_2=0}^{d-1} \frac{1}{2^B} (\hat{X} \otimes \openone) \hat{U}_{m, b, c_1,c_2}\\
    &= \sum_{b=0}^{B-1} \sum_{c_1=0}^{d-1} \sum_{c_2=0}^{d-1} \sum_{\substack{j: (0,j) \in \\ V(G_m^{(c_1, c_2)})}} g(m, j, c_1, b) e^{i \arg \left( \braket{ \chi_{f_{ind}(m, 0, j, c_1)}^{(\ell_{m+1)}} }{ \chi_j^{(\ell_m)}}  \right) - i \lambda^{(\ell_m)}_j \tau_m \frac{t}{r}} \ketbra{0}{0} \otimes \ketbra{f_{ind}(m, 0, j, c_1)}{j}. \label{tildeA}
\end{align}
Most of the terms in $\hat{U}_{m,b,c_1,c_2}$ end up cancelling after the summation over $b$, so $\widetilde{A}$ above only ends up having the transitions that we care about implementing.
\begin{claim} \label{clm:ast}
    Let $\widetilde{A}_m$ and $\hat{A}_m$ be defined as in \eqref{tildeA} and \eqref{A_m2}, respectively. Then,
    \begin{align}
        \left\| \big(\bra{0} \otimes \normalfont{\openone} \big) \widetilde{A}_m \big(\ket{0} \otimes \normalfont{\openone} \big) - \hat{A}_m \right\|_2 \leq \frac{2d^2}{2^B} \label{ast_claim}
    \end{align}
    where $d$ is the sparsity of the Hamiltonian decomposition (Def. \ref{def:dsparse}) and $B$ is the number of bits of precision for the inner product oracle (Def. \ref{def:O_IM}). 
\end{claim}
See Appendix \ref{sec:ast_full_proof} for proof. The claim above tells us that if we have some input state $\ket{0} \ket{\psi}$ and apply $\widetilde{A}_m$, and project onto the first qubit onto the $\ket{0}$ state, we will approximate $\hat{A}_m$ acting on $\ket{\psi}$ to within error $\frac{2d^2}{2^B}$. The probability of projecting onto the $\ket{0}$ subspace is 1, since the other terms in $ (\hat{X} \otimes \openone) \hat{U}_{m, b, c_1,c_2}$ always cancel to 0 due to the $(-1)^b$ phase being summed over. 

\subsection{Linear Combination of Unitaries and Robust Oblivious Amplitude Amplification}
The sum will be implemented using the linear combination of unitaries (LCU) method with robust oblivious amplitude amplification (ROAA). The LCU method is a way of implementing an operator that is written as a linear combination of unitary operators by block-encoding it within a larger unitary, and then projecting onto the correct subspace. In order to boost the probability of success to 1, or near 1, we need to use oblivious amplitude amplification. Oblivious amplitude amplification works similarly to regular amplitude amplification, but does not require reflections about the input state. The original statement and proof of oblivious amplitude amplification \cite{bcck2014} showed that this worked for a block-encoded unitary. However, our operator $(\bra{0} \otimes \openone) \widetilde{A}_m (\ket{0} \otimes \openone)$ is only approximately unitary, since it is an approximation of $\hat{A}_m$, which is exactly unitary. It was shown in \cite{bck2015} that oblivious amplitude amplification can be performed up to $O(\delta)$ precision (called \textit{robust} oblivious amplitude amplification) if the block-encoded operator is $\delta$-close to a unitary. In our case, since we have Claim \ref{clm:ast}, we can use the robust oblivious amplitude amplification result to implement $\widetilde{A}_m$ to within $O \left( \frac{d^2}{2^B} \right)$ error. 

We will define the Prepare and Select operations for the linear combination of unitaries that we are implementing in the following definitions.
\begin{definition}[Prepare operation] \label{def:prep}
    Let \normalfont{PREP}$ \in \mathcal{L}( \mathcal{H}_{2^B} \otimes \mathcal{H}_d \otimes \mathcal{H}_d)$ be a unitary operator,
    \begin{align}
        \text{\normalfont{PREP}} : \ket{0}^{\otimes B} \ket{0}^{\otimes \log d} \ket{0}^{\otimes \log d} \mapsto \frac{1}{d \sqrt{2^B}} \sum_{b=0}^{2^B-1} \sum_{c_1=0}^{d-1} \sum_{c_2=0}^{d-1} \ket{b} \ket{c_1} \ket{c_2}
    \end{align}
\end{definition}
\begin{definition}[Select operation] \label{def:sel}
    Let $\vec{U}_m = (\hat{U}_{m, 0, 0, 0}, \hat{U}_{m, 1, 0, 0}, \ldots \hat{U}_{m, B-1, d-1, d-1})$. Let \normalfont{SEL}$(\vec{U}_m) \in \mathcal{L} (\mathcal{H}_{2^B} \otimes \mathcal{H}_{d} \otimes \mathcal{H}_d \otimes \mathcal{H}_2 \otimes \mathcal{H}_{2^n})$ be a unitary operator,
    \begin{align}
        \text{\normalfont{SEL}}(\vec{U}_m) := \sum_{b=0}^{B-1} \sum_{c_1=0}^{d-1} \sum_{c_2=0}^{d-1} \ketbra{b}{b} \otimes \ketbra{c_1}{c_1} \otimes \ketbra{c_2}{c_2} \otimes  \left((\hat{X} \otimes \normalfont{\openone}) \hat{U}_{m, b, c_1, c_2} \right),
    \end{align}
\end{definition}
We can implement the Prepare operation by simply applying Hadamard gates to all the qubits in the case where $d$ is power of 2. We can build the Select operation out of the previously defined oracles, using an additional $\log M + 3 + n$ ancillary qubits to store the information needed for the oracles; the full algorithm is shown in Appendix \ref{sec:lcu}.

The circuit in Fig. \ref{fig:oraclecircuit} shows how the $O_C$ oracle defined in \ref{def:O_C} can be implemented using eight queries to the $O_{ind}$ oracle and $n+1$ additional qubits.
\begin{figure}[t]
\centering
\resizebox{0.9\textwidth}{!}{
    \Qcircuit @C=1em @R=1em {
        \lstick{r_0 = \ket{m}} & \qw & \multigate{3}{O_{ind}(r_0, r_1, r_2, r_3, r_5)} & \multigate{2}{O_{ind}(r_0, r_1, r_2, r_4, r_5)} & \qw & \multigate{2}{O_{ind}(r_0, r_1, r_5, r_4, r_2)} & \multigate{3}{O_{ind}(r_0, r_1, r_5, r_3, r_2)} & \qw & \multigate{6}{\text{uncompute}} \\
        \lstick{r_1 = \ket{b}} & \ctrl{5} & \ghost{O_{ind}(r_0, r_1, r_2, r_3, r_5)} & \ghost{O_{ind}(r_0, r_1, r_2, r_4, r_5)} & \gate{X} & \ghost {O_{ind}(r_0, r_1, r_5, r_4, r_2)} & \ghost {O_{ind}(r_0, r_1, r_5, r_3, r_2)} &  \qw & \ghost{\text{uncompute}} \\ 
        \lstick{r_2 = \ket{j}} & \qw & \ghost{O_{ind}(r_0, r_1, r_2, r_3, r_5)} & \ghost{O_{ind}(r_0, r_1, r_2, r_4, r_5)} & \qw & \ghost{O_{ind}(r_0, r_1, r_5, r_4, r_2)} & \ghost{O_{ind}(r_0, r_1, r_5, r_3, r_2)} & \ctrlo{5} & \ghost{\text{uncompute}}\\
        \lstick{r_3 = \ket{c_1}} & \qw & \ghost{O_{ind}(r_0, r_1, r_2, r_3, r_5)} & \qwx[-1] \qwx[1] \qw & \qw & \qwx[-1] \qwx[1] \qw & \ghost{O_{ind}(r_0, r_1, r_5, r_3, r_2)} & \qw & \ghost{\text{uncompute}}\\
        \lstick{r_4 = \ket{c_2}} & \qw & \qwx[-1] \qwx[1] \qw & \multigate{1}{\color{white} {O_{ind}(r_0, r_2, r_2, r_4, r_5)}} & \qw & \multigate{1}{{\color{white} O_{ind}(r_0, r_1, r_5, r_4, r_2)}} & \qw \qwx[-1] \qwx[1] & \qw & \ghost{\text{uncompute}} \\
        \lstick{r_5 = \ket{0}^{\otimes n}} & \qw & \gate{{\color{white} O_{ind}(r_0, r_1, r_2, r_3, r_5)}} & \ghost{O_{ind}(r_0, r_1, r_2, r_4, r_5)} & \qw & \ghost{O_{ind}(r_0, r_1, r_5, r_4, r_2)} & \gate{{\color{white} O_{ind}(r_0, r_1, r_5, r_3, r_2)}} & \qw &  \ghost{\text{uncompute}} \\
        \lstick{r_6 = \ket{0}} & \targ & \ctrl{-1} & \ctrlo{-1} & \qw & \ctrl{-1} & \ctrlo{-1} & \qw & \ghost{\text{uncompute}} \\
        \lstick{r_7 = \ket{0}} & \qw & \qw & \qw & \qw & \qw & \qw & \targ & \qw & \qw \gategroup{1}{1}{7}{7}{.7em}{--}
    }}
    \caption{Circuit to implement $O_C$ using two queries to the $O_{ind}$ oracle. The uncompute process is simply all the operations within the dashed box performed again, but in reverse order. The circuit applies 8 queries of the $O_{ind}$ oracle, with the order of the inputs depending on the value of $b$, which indicates whether the input state corresponds to the $\hat{H}_{\ell_m}$ basis or the $\hat{H}_{\ell_{m+1}}$ basis. The value of the qubit in the $r_7$ register gets flipped all the qubits in $r_2$ are zero.} \label{fig:oraclecircuit}
\end{figure}
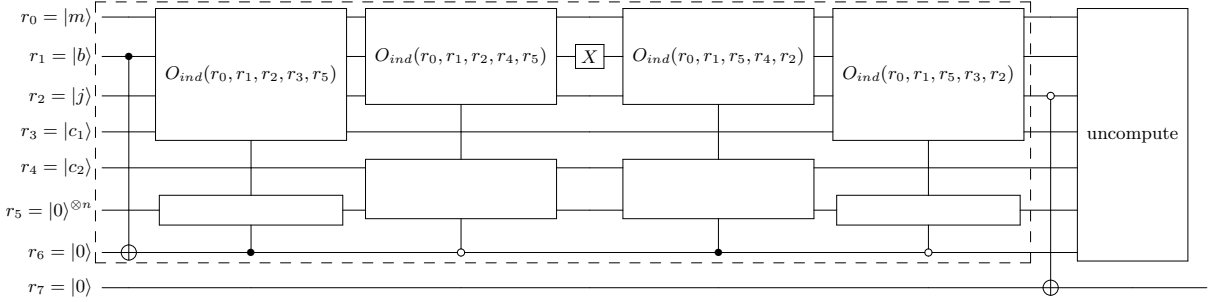

\begin{lemma}[Linear combination of unitaries] \label{lemma:lcu}
    Let
    \begin{align}
        \hat{W} := (\text{\normalfont{PREP}}^\dagger \otimes \normalfont{\openone}) \text{\normalfont{SEL}} (\vec{U}_m) (\text{\normalfont{PREP}} \otimes \normalfont{\openone}),
    \end{align}
    then,
    \begin{align}
        \left( (\ketbra{0}{0})^{\otimes (1 + B + 2 \log d)} \otimes \normalfont{\openone} \right) \hat{W} \ket{0}^{\otimes (B + 2 \log d)} \ket{0} \ket{\psi} = \frac{1}{d^2} \ket{0}^{\otimes (B + 2 \log d)} \widetilde{A}_m \ket{0} \ket{\psi}.
    \end{align}
\end{lemma}
The proof of Lemma \ref{lemma:lcu} can be found in Appendix \ref{sec:lcu_appendix}. A challenge is that the success probability of the operation, in general, is not close to $1$.  We can address this by using robust oblivious amplitude amplificatiom, which is described in the following lemma. we can use the following lemma from \cite{bck2015}:
\begin{lemma}[LCU with ROAA lemma] \label{lemma:roaa_lcu}
    Let $\delta > 0$, $\widetilde{A} = \sum_{a_j=0}^{J-1} a_j \hat{U}_j$, where each $\hat{U}_j$ is a unitary operators acting on a finite dimensional Hilbert space and $a_j \geq 0$, such that $\widetilde{A}$ is $\delta$-close with respect to the spectral norm to $\hat{A}$ for some unitary $\hat{A}$. Then, we can approximate $\widetilde{A}$ to within $O(\delta)$ using $O(a)$ {\normalfont{PREP}} and {\normalfont{SEL}$(\vec{U}_m)$} operations and their inverse operations, where $a:= \sum_{j=0}^{J-1} a_j$.
\end{lemma}
In our case, from Claim \ref{clm:ast}, $(\bra{0}\otimes \openone) \widetilde{A}_m (\ket{0} \otimes \openone)$ is $\delta$-close to the unitary $\hat{A}_m$ where $\delta = \frac{2d^2}{2^B}$. We are summing over $2^B d^2$ unitaries $\hat{U}_{m, b, c_1, c_2}$, where the coefficient on each of the unitaries is $\frac{1}{2^B}$. Let 
\begin{align}
    a = \sum_{b=0}^{2^B-1} \sum_{c_1=0}^{d-1} \sum_{c_2=0}^{d-1} \frac{1}{2^B} = \frac{2^Bd^2}{2^B} = d^2.
\end{align}
By Lemma \ref{lemma:roaa_lcu}, we can approximate $\widetilde{A}_m$ to within $O \left(\frac{d^2}{2^B}\right)$ using $O(d^2)$ Prepare and Select operations and their inverse operations. Each application of Select uses a constant number of oracle queries, $O(B + \log d)$ additional 2-qubit gates and $O(B + \log d)$ extra qubits, and our Prepare operation is simply the Hadamard gate applied to $B + 2 \log d$ qubits, since all the amplitudes on the unitaries in the LCU are the same (note that this uses fewer gates than in \cite{bck2015}, which proves the result for a linear combination of unitaries with arbitrary coefficients). Thus, approximating $\widetilde{A}_m$ to within error $O \left( \frac{d^2}{2^B} \right)$ takes $O(d^2)$ oracle queries, $O \left(d^2 (B + \log d) \right)$ additional two-qubit gates, and $O (B + \log d)$ additional qubits. More details about the implementation are given in Appendix \ref{sec:roaa_appendix}.

\subsection{Error Scaling and Complexity Analysis}
\begin{theorem}
Let $\mathcal{H}$ be a finite-dimensional Hilbert space and let $\hat{H}=\sum_{j=1}^L \alpha_j \hat{H}_j$ where $\alpha_j\ge 0$ and $\hat{H}_j \in \mathcal{L}(\mathcal{H})$ is a Hermitian operator for each $j$ and these $\hat{H}_j$ form a $d$-sparse Hamiltonian decomposition as defined in Definition~\ref{def:sparse}.  Then for any $t>0$ and $\epsilon>0$ a unitary $V$ can be constructed such that $\|(\bra{0} \otimes \openone ) \hat{V} (\ket{0} \otimes \openone) - e^{-iHt}\|\le \epsilon$ using a number of queries to the oracles $O_{ind}, O_{IM}, O_{IP}, O_{EP}$ that scales as
$$
O \left( \frac{L5^{k} \alpha_{\rm comm}^{\frac{1}{2k}} t^{1 + \frac{1}{2k}} d^2}{\epsilon^{\frac{1}{2k}}} \right),
$$
and a number of additional two-qubit gates that scales as
$$
\tilde{O} \left( d^2 L 5^{k} \frac{\alpha_{\rm comm}^{\frac{1}{2k}} t^{1 + \frac{1}{2k}}}{\epsilon^{\frac{1}{2k}}}\right)
$$
\end{theorem}
\begin{proof}
From Lemma \ref{lemma:roaa_lcu}, each iteration of the ROAA process implements $(\bra{0}\otimes \openone) \widetilde{A}_m (\ket{0}\otimes \openone)$ with error $\delta \in O \left( \frac{d^2}{2^B} \right)$, that is,
\begin{align}
    \left\| (\bra{0}^{\otimes (1 + B + 2 \log d)}  \otimes \openone) \widetilde{R}_m (\ket{0}^{\otimes (1 + B + 2 \log d)}  \otimes \openone) - (\bra{0} \otimes \openone) \widetilde{A}_m (\ket{0} \otimes \openone)\right\| \in O \left( \frac{d^2}{2^B} \right),
\end{align}
where $\widetilde{R}_m$ denotes the operation implemented by the robust oblivious amplitude amplification process described in Appendix \ref{sec:roaa_appendix} (i.e., $\widetilde{R}_m := \hat{\Pi} \hat{R}^p \hat{W} \hat{\Pi}$, where $\hat{\Pi}$, $\hat{R}$, $\hat{W}$ and $p$ are defined in Appendix \ref{sec:roaa_appendix}). From Claim \ref{clm:ast}, we also have
\begin{align}
    \left\| \big(\bra{0} \otimes \normalfont{\openone} \big) \widetilde{A}_m \big(\ket{0} \otimes \normalfont{\openone} \big) - \hat{A}_m \right\| \in O \left( \frac{d^2}{2^B} \right)
\end{align}
By a straightforward application of the triangle inequality, we have that the ROAA process implements $\hat{A}_m$ to within error $O \left( \frac{d^2}{2^B} \right)$:
\begin{align}
     & \quad \left\| (\bra{0}^{\otimes (1 + B + 2 \log d)}  \otimes \openone) \widetilde{R}_m (\ket{0}^{\otimes (1 + 2 \log d)}  \otimes \openone) - \hat{A}_m \right\| \\
     &= \left\| (\bra{0}^{\otimes (1 + B + 2 \log d)}  \otimes \openone) \widetilde{R}_m (\ket{0}^{\otimes (1 + B + 2 \log d)}  \otimes \openone) - (\bra{0} \otimes \openone) \widetilde{A}_m (\ket{0} \otimes \openone) + \big(\bra{0} \otimes \normalfont{\openone} \big) \widetilde{A}_m \big(\ket{0} \otimes \normalfont{\openone} \big) - \hat{A}_m \right\| \\
     &\leq \left\| (\bra{0}^{\otimes (1+ B + 2 \log d)}  \otimes \openone) \widetilde{R}_m (\ket{0}^{\otimes (1 + B + 2 \log d)}  \otimes \openone) - (\bra{0} \otimes \openone) \widetilde{A}_m (\ket{0} \otimes \openone) \right\| + \left\| \big(\bra{0} \otimes \normalfont{\openone} \big) \widetilde{A}_m \big(\ket{0} \otimes \normalfont{\openone} \big) - \hat{A}_m \right\| \\
     & \in O \left( \frac{d^2}{2^B} \right)
\end{align}

By Box 4.1 in \cite{nielsenchuang}, the total error of the entire process of $M$ transitions is less than or equal to the sum of the errors from the individual transitions:
\begin{align}
    & \quad \left\| \prod_{m=0}^{M-1} \hat{A}_m  - \prod_{m=0}^{M-1} (\bra{0}^{\otimes (1 + B + 2 \log d)}  \otimes \openone) \widetilde{R}_m (\ket{0}^{\otimes (1 + B + 2 \log d)}  \otimes \openone) \right\| \\
    &\leq \sum_{m=0}^{M-1} \left\| \hat{A}_m - (\bra{0}^{\otimes (1 + B + 2 \log d)}  \otimes \openone) \widetilde{R}_m (\ket{0}^{\otimes (1 + B + 2 \log d)}  \otimes \openone) \right\| \\
    & \in O \left( M \frac{d^2}{2^B} \right) \\
    & \subseteq O \left( \frac{L5^{k}rd^2}{2^B} \right)
\end{align}
Recall that $L$ is the number of terms in the Hamiltonian decomposition, $2k$ is the Trotter order, and $r$ is the number of Trotter time steps. This is the error on implementing the $2k^{\rm th}$-order Trotter product formula in \eqref{trotter}, but there is also error on this product itself compared to the ideal evolution $e^{-i \hat{H} t}$. By Lemma \ref{lemma:trottererror}, this Trotter error is bounded by
\begin{align}
    \left\| e^{-i \hat{H} t} - \prod_{m=0}^{M-1} e^{-i \hat{H}_{\ell_m} \tau_m \frac{t}{r}} \right\| \leq \frac{\alpha_{\rm comm} t^{2k+1}}{r^{2k}}
\end{align}    
where $\alpha_{\rm comm} := \sum_{j_0,\ldots,j_{2k}} \left\|[H_{j_0}, [H_{j_1},\cdots[H_{j_{2k-1}},H_{j_{2k}}]\cdots]] \right\|_\infty$. Thus, the overall error on the entire evolution is
\begin{align}
     & \quad  \left\| e^{-i \hat{H} t} -  \prod_{m=0}^{M-1} (\bra{0}^{\otimes (1 + B + 2 \log d)}  \otimes \openone) \widetilde{R}_m (\ket{0}^{\otimes (1 + B + 2 \log d)}  \otimes \openone)  \right\| \\
     &= \left\| e^{-i \hat{H} t} - \prod_{m=0}^{M-1} e^{-i \hat{H}_{\ell_m} \tau_m \frac{t}{r}} + e^{-i \hat{H}_{\ell_m} \tau_m \frac{t}{r}} -  \prod_{m=0}^{M-1} (\bra{0}^{\otimes (1 + B + 2 \log d)}  \otimes \openone) \widetilde{R}_m (\ket{0}^{\otimes (1 + B + 2 \log d)}  \otimes \openone)  \right\| \\
     &\leq \left\| e^{-i \hat{H} t} - \prod_{m=0}^{M-1} e^{-i \hat{H}_{\ell_m} \tau_m \frac{t}{r}} \right\| + \left\| e^{-i \hat{H}_{\ell_m} \tau_m \frac{t}{r}} - \prod_{m=0}^{M-1} (\bra{0}^{\otimes (1 + B + 2 \log d)}  \otimes \openone) \widetilde{R}_m (\ket{0}^{\otimes (1 + B + 2 \log d)}  \otimes \openone)  \right\| \\
     & \in O \left( \frac{L 5^k r d^2}{2^B} + \frac{\alpha_{\rm comm} t^{2k+1}}{r^{2k}} \right)
\end{align}
If we want the above to be less than some error tolerance $\epsilon$, we can choose $B \in \Theta \left(\log \left( \frac{5^{k} r^{2k+1} L d^2}{\alpha_{\rm comm} t^{2k+1}} \right) \right)$, which gives us an overall error scaling of
\begin{equation}
    \epsilon \in O \left( \frac{\alpha_{\rm comm} t^{2k+1}}{r^{2k}} \right),
\end{equation}
and the number of time steps that we will need is 
\begin{equation}
    r \in O \left( \frac{\alpha_{\rm comm}^{\frac{1}{2k}} t^{1 + \frac{1}{2k}}}{\epsilon^{\frac{1}{2k}}} \right).
\end{equation}
Since we have $M = 2 L 5^{k-1} r$ transition steps for the 2$k^{\rm th}$-order Trotter decomposition (see Lemma \ref{lemma:HamPathInt}), each of which requires $O (d^2)$ oracle queries and $O \left( d^2 (B + \log d) \right)$ additional two-qubit gates, in order to implement the entire time evolution within error $\epsilon$, we have a total query complexity of
\begin{equation}
    O \left( \frac{L5^{k} \alpha_{\rm comm}^{\frac{1}{2k}} t^{1 + \frac{1}{2k}} d^2}{\epsilon^{\frac{1}{2k}}} \right),
\end{equation}
and the number of additional 2-qubit gates needed is
\begin{align}
    & \quad O \left( d^2 L 5^{k} \frac{\alpha_{\rm comm}^{\frac{1}{2k}} t^{1 + \frac{1}{2k}}}{\epsilon^{\frac{1}{2k}}} \left(\log \left(\frac{5^{k} r^{2k+1} L^2 d^2}{\alpha_{\rm comm} t^{2k+1}} \right) + \log d \right) \right) \\
    & \subseteq \tilde{O} \left( d^2 L 5^{k} \frac{\alpha_{\rm comm}^{\frac{1}{2k}} t^{1 + \frac{1}{2k}}}{\epsilon^{\frac{1}{2k}}} \right)
\end{align}
\end{proof}
This shows that we can achieve near-linear simulation time using our Hamiltonian simulation method.  Similarly, we can choose $k\in O(\sqrt{\log(\alpha_{\rm comm}t/\epsilon)})$ then sub-polynomial scaling with $1/\epsilon$ can be achieved.  It is well known that linear scaling with $t$ is the optimal scaling for the number of queries needed to the matrix elements of the Hamiltonian.  In our context, however, our oracles provide information about the inner products and eigenvalues of the Hamiltonians in the decomposition.  

\begin{lemma}
Let $\mathcal{H}$ be a finite dimensional Hilbert space,  $t\in \mathbb{R}$ and let $\hat{H} = \sum_j \alpha_j H_j$ such that each $\hat{H}_j \in \mathcal{L}(\mathcal{H})$ with $\alpha_j \ge 0$ and $\alpha := \sum_j \alpha_j$ and these operators form a $d$-sparse decomposition of $\hat{H}$.
No algorithm exists that can implement a quantum channel $W: \mathcal{L}(\mathcal{H}) \mapsto \mathcal{L}(\mathcal{H})$ such that $D_{\rm Tr}(W,e^{-i\hat{H}t}) \le 1/6$ for all such $\hat{H}$ using a number of queries to $O_{ind},O_{IP},O_{IM},O_{EV}$ that scales as $o(\alpha t)$.
\end{lemma}
\begin{proof}
The proof of this claim follows from the exact same reasoning as the quantum no fast-forwarding theorem.  We consider a Hamiltonian that has a sparse decomposition and cannot be simulated without violating a well known quantum lower bound.  Specifically, we choose to use the parity Hamiltonian whose dynamics can be used to compute the parity of a bitstring.  It can easily be seen from lower bound techniques such as the polynomial method that computing the parity of a bitstring of length $N$ with probability of success greater than $2/3$ requires $\Omega(N)$ queries to an oracle that provides elements of the bitstring.  

The Hamiltonian for this system works as follows.  First let $Y$ be the bitstring whose parity we wish to compute.  Let $\ket{x,j}\in \mathbb{C}^{N} \otimes \mathbb{C}^2$ represent a position for a combined system that tracks the label of the index of the bitstring and the parity of all the bits computed up to point $x$.  We then define
\begin{equation}
\hat{H}= i\sum_{xj} \ketbra{x+1,Y(x)\oplus j}{x,j}- \ketbra{x,j}{x+1,Y(x)\oplus j}
\end{equation}
The Hamiltonian is permutationally equivalent to the Hamiltonian
\begin{equation}
\hat{H}'= i\sum_{xj} \ketbra{x+1,j}{x,j}- \ketbra{x,j}{x+1,j}
\end{equation}
This Hamiltonian can be diagonalized using the quantum Fourier transform and it can be seen from the eigenvalues that there exists $z\in \Theta(N)$ such that $e^{-i\hat{H}' z}\ket{0,0} =\ket{N-1,0}$.  This is permutationally equivalent to the state $\ket{N-1, Y(N-1)}$ if we transform back to the basis used to represent $\hat{H}$.

Next assume that we can implement an operator $W$ that is within trace distance $1/2 - \delta$ of $e^{-i\hat{H}t}$.  The trace distance on channels is defined to be the induced trace norm on density operators.  Specifically if we define $\mathcal{D}\subset \mathcal{L}(\mathcal{H})$ to be the set of valid density operators on the underlying Hilbert space then
\begin{equation}
D_{\rm Tr}(W,e^{-i \hat{H} t}) := \frac{1}{2}\sup_{\rho\in \mathcal{D}}\|W(\rho) - e^{-i \hat{H}t} \rho e^{i\hat{H}t}\|_1.
\end{equation}
We then have that for any operator $\hat{A}$ that the expectation value of the operator is
\begin{equation}
\sup_{\rho \in \mathcal{D}} |{\rm Tr}( W(\rho)A -  e^{-i \hat{H}t} \rho e^{i\hat{H}t}|A)|
\end{equation}
In our context let $A$ be a projector onto a subspace that decides the answer to a particular decision problem.  In that case, we have that $\|A\| =1$ and from the von Neumann trace inequality
\begin{equation}
\sup_{\rho \in \mathcal{D}} |{\rm Tr}( W(\rho)A -  e^{-i \hat{H}t} \rho e^{i\hat{H}t}|\hat{A})| \le 2D_{\rm Tr}(W,e^{-i \hat{H} t})
\end{equation}
Thus ensuring that $\delta<2/3$ guarantees that the trace error in the answer to the decision problem is at most $1/3$.  This is why we ensure that the trace distance is at most $1/6$.

Next we want to consider simulating this Hamiltonian using our oracles.  To make this translation, we will begin by performing a further decomposition of the Hamiltonian.  Let $\hat{H}_{even}$ and $\hat{H}_{odd}$ be two one-sparse Hamiltonians that we decompose the Hamiltonian into.  Specifically,
\begin{align}
\hat{H}_{\rm even}&= i\sum_{xj} \ketbra{2x+1,Y(2x)\oplus j}{2x,j}- \ketbra{2x,j}{2x+1,Y(2x)\oplus j}\\
\hat{H}_{\rm odd}&= i\sum_{xj} \ketbra{2x+2,Y(2x+1)\oplus j}{2x+1,j}- \ketbra{2x+1,j}{2x+2,Y(2x+1)\oplus j}
\end{align}
This is a $d=2$ sparse decomposition and thus it is a valid example of our results.  

Let us now assume that an algorithm exists that can simulate the dynamics within a number of queries to $O_{ev}, O_{ip}, O_{im}, O_{ind}$ that is in $o(t)$.  Next assume that we have an oracle $O_Y$ such that $O_Y\ket{x}\ket{0} = \ket{x}\ket{Y(x)}$.  As $\hat{H}_{\rm even}$ and $\hat{H}_{\rm odd}$ are one sparse, the eigenbasis of both can be chosen such that each eigenvector lies in an irreducible two-dimensional subspace.  Each such eigenvector can be expressed as
\begin{equation}
\ket{\psi_{\pm}} = \frac{\ket{2x,j\oplus Y(2j)} \pm \ket{2x+1,j}}{\sqrt{2}}.
\end{equation}
A similar expression holds for the eigenvectors of $\hat{H}_{\rm odd}$ which we denote $\ket{\phi_{\pm}}$.  
For each such eigenvector of $\hat{H}_{\rm even}$ there are at most $8$ eigenvectors that could have non-zero overlap.  Thus at most a constant number of inner products and eigenvalues need to be computed.  The eigenvalues for $\hat{H}_{\rm even}$ are trivially $\pm 1$.  If the eigenvectors are indexed by this eigenvalue then no queries are needed to compute the eigenvalue here.

The inner products on the other hand require a non-zero number of queries.  In particular, specifically, the inner product between two eigenvectors of $\hat{H}_{\rm even}$ and $\hat{H}_{\rm odd}$ requires at most four queries to $O_Y$.  Thus a single query to $O_{IP}$ and $O_{IM}$ can be simulated using $\Theta(1)$ queries to $O_Y$.  In this case the oracle $O_{\rm ind}$ and $O_{EV}$ can be implemented using no queries as they depend only on the labels for the eigenvectors.  Thus the total number of queries to $O_Y$ needed to simulate a single query to these oracles is in $O(1)$.  Thus if there existed a quantum algorithm for simulating Hamiltonian dynamics using the path integral representation that requires $o(t)$ queries then this parity simulation algorithm would require $o(N)$ queries to $O_Y$, which violates the $\Omega(N)$ quantum lower bound for computing parity with probability greater than $2/3$~\cite{beals2001quantum}.  Thus linear scaling with $t$ is optimal if $\alpha \in O(1)$.

Finally, the fact that $\alpha t$ scaling is optimal follows immediately from the fact that if we were to change the Hamiltonian such that $\hat{H} \mapsto \alpha \hat{H}$ then we would only need simulation time $t\mapsto t/\alpha$.  Thus repeating the same argument we arrive at the conclusion that $o(\alpha t)$ scaling is optimal.
\end{proof}
\section{Long Time Hamiltonian Path Integral Algorithm} \label{sec:LTHamPathInt}
As mentioned earlier, the Hamiltonian path integral method is not restricted to solely the short time step case.  In contrast, we can consider a different approach that instead divides the paths into much longer segments and uses arguments reminiscent of the stationary phase approximation or the adiabatic theorem to evaluate the time evolution.

The central observation behind these approaches can be seen using the following argument.  Let us consider all paths that begin in a particular value of $j_0$ and end in a particular value of $j_{M-1}$.  The sum over all such paths in the path integral that gives the coefficient of such a transition takes the following form
\begin{equation}
    \sum_{j_1=0}^{2^n-1} \ldots \sum_{j_{M-2}=0}^{2^n-1}
 e^{-i \left( \sum_{m=0}^{M-1} \lambda_{j_m}^{(\ell_m)}\tau_m \right) \frac{t}{r}} \prod_{u=1}^{M-2} \braket{\chi_{j_u}^{(\ell_u)}}{\chi_{j_{u-1}}^{(\ell_{u-1})}}.
\end{equation}
If we envision a situation where the sum over eigenvalues multiplied by the evolution times for each path is a random variable that is approximately uniformly distributed mod $2\pi$, then the resulting sum will be expected to be small because of phase cancellation between the paths.  This means that not all sets of paths are likely to have equal importance to the overall sum; further, if we can isolate paths that do not contribute substantially to the path we can remove them from consideration and thereby substantially reduce the cost of a simulation.

Here we will use tools from adiabatic perturbation theory~\cite{jansen2007bounds,cheung2011improved} to find ways of combining the paths together into groups within which the impacts of the phase cancellation between the paths can be computed.  The adiabatic theorem specifically states that if we have a slowly varying time-dependent Hamiltonian then an initial eigenstate will approximately remain in an initial eigenstate of the Hamiltonian throughout the evolution.  This implies that the evolution is dominated by the zero-jump path which does not transition from the instantaneous eigenstates throughout the evolution.  This is very useful because it provides a natural grouping of the paths.  The challenge behind applying this, however, is that it specifically requires a time-dependent Hamiltonian.  We show in the following that, in some cases, a time-independent Hamiltonian can be written as an adiabatic time-dependent Hamiltonian in the interaction frame.
\begin{proposition}
Let $\hat{H}=\hat{A}+\hat{B}$ for Hermitian $\hat{A},\hat{B}$ in $\mathbb{C}^{2^n\times 2^n}$ and assume that $\hat{B}$ is gapped and non-degenerate.  Let $\ket{j}$ be an eigenvector of $\hat{B}$ with eigenvalue $E_j$ and assume that for any distinct eigenvalue $E_k$ that $|E_j-E_k|\ge \gamma_{\min}$.  We then have that for evolution time $T$ if $\max_{j=0,1,2}(\|\hat{A}^{j+1}\|T^{j}) \le O(\epsilon \gamma_{\min}^2/\|\hat B\| )$ then the error in the adiabatic approximation for the evolution $\mathcal{T}e^{-i \int_0^T\hat H_{\rm int}(t)\mathrm{d}t}\ket{j}$ (quantified by the projection of the evolution of an instantaneous eigenvector at $t=0$ for $H_{\rm int}$ onto the subspace orthogonal to this at $t=T$) for the interaction frame Hamiltonian $\hat{H}_{\rm int}:= e^{i\hat{A} t} \hat{B} e^{-i\hat{A}t}$ is at most $O(\epsilon)$.
\end{proposition}
\begin{proof}
First, note that because $\hat H_{\rm int}(t)$ is a unitary conjugation of the time-independent Hamiltonian $\hat{B}$.  This implies that the instantaneous eigenvalues of $\hat H_{\rm int}(t)$ are time-independent and so the instantaneous eigenvalue gap of $H_{\rm int}(t)$ is $\gamma_{\min}$ for all $t$.

Next note that the derivatives of $\hat H_{\rm int}$ are of the form
\begin{align}
    \|\partial_t \hat H_{\rm int}(t)\| &\in O( \|[\hat{A},\hat{B}]\|)\subseteq O(\|\hat{A}\| \|\hat{B}\|)\\
    \|\partial_t^2 \hat H_{\rm int}(t)\| &\in O( \|[\hat{A},[\hat{A},\hat{B}]]\|)\subseteq O(\|\hat{A}\|^2 \|\hat{B}\|)\\
    \|\partial_t^3 \hat H_{\rm int}(t)\| &\in O( \|[\hat{A},[\hat{A},[\hat{A},\hat{B}]]]\|)\subseteq O(\|\hat{A}\|^3 \|\hat{B}\|)
\end{align}

Theorem 4 of Jansen, Seiler and Ruskai~\cite{jansen2007bounds} shows that the error in the adiabatic approximation for the Hamiltonian $H_{\rm int}$ is
\begin{equation}
\epsilon_{adiabatic}=O\left(\frac{\|\partial_t \hat H_{\rm int}(t)\|T + \|\partial_t^2 \hat H_{\rm int}(t)\|T^2+\|\partial_t^3 \hat H_{\rm int}(t)\|T^3}{\gamma_{\min}^2 T} \right).
\end{equation}
Our proof immediately follows from the above results and the assumption that $\max_{j=0,1,2}(\|\hat{A}^{j+1}\|T^{j}) \le O(\epsilon \gamma_{\min}^2/\|\hat B\| )$.
\end{proof}
This shows that if we have a strongly gapped Hamiltonian $\hat{B}$ then  we can perform a transformation into an interaction frame where the Hamiltonian is time-independent and further that the resulting Hamiltonian will obey the adiabatic approximation.

Since we can without loss of generality analyze the time evolution in the interaction frame as a time-dependent Hamiltonian we will focus our attention in the following on time-dependent Hamiltonians $\hat{H}(t)$ with instantaneous eigenvectors $\ket{\chi_j(t)}$ such that $\hat{H}(t) \ket{\chi_j(t)} = \lambda_j(t) \ket{\chi_j(t)}$ for $j \in \{0, \ldots, 2^n=1\}$.

For any such time-dependent Hamiltonian we can follow the exact same methodology as that used in the short time evolution to define the path integral.  In particular,
\begin{align}
    \mathcal{T}e^{-i \int_0^T \hat{H}(t) \mathrm{d}t} &:= \lim_{r\rightarrow \infty} \prod_{j=0}^{r-1} e^{-i \hat{H}(j\frac{t}{r}) \frac{t}{r}}\nonumber\\
    &=\lim_{r\rightarrow \infty} \prod_{j=0}^{r-1} \left(\sum_k e^{-i \lambda_k(j\frac{t}{r}) \frac{t}{r}}\ketbra{\chi_k(jt/r)}{\chi_k(jt/r)}\right)
\end{align}
This immediately takes the same form as Lemma~\ref{lemma:HamPathInt}.  However, unlike that case it is helpful here to consider the limit as $r\rightarrow \infty$.  This limit is straightforward to evaluate and is given
by Theorem 3 of~\cite{cheung2011improved} to be the following (for simplicity, we use a change of variables of $s=t/T$, and will write $\lambda(s)$ and $\hat{H}(s)$ instead of $\lambda(sT)$ and $\hat{H}(sT)$).

\begin{lemma}[Theorem 3 of~\cite{cheung2011improved}] \label{lemma:TimeDepHam}
Let $\hat{H}(t)$ be a time-dependent Hermitian operator that is at least twice differentiable on $[0,1]$ and assume that the instantaneous eigenbasis of $\hat{H}(t)$ is chosen such that $\forall~j,k~\braket{\dot\chi_j(t)}{\chi_k(t)} \propto (1-\delta_{jk})$. Then,
\begin{align*}
&\mathcal{T}e^{-i T \int_0^1 \hat{H}(s) \mathrm{d}s}=e^{-i T \int_0^1 \lambda_j(t) \mathrm{d} s } \sum_j \ketbra{\chi_j(1)}
{\chi_j(0)}\\
&\quad+\sum_{j_1\ne j_0}\sum_{j_0} \ketbra{\chi_{j_1}(1)}{\chi_{j_0}(0)}\int_{0}^{1} \beta_{j_1,j_0}(s_1)e^{-iT \left(\int_{s_1}^1 \lambda_{j_1}(u) \mathrm{d}u+\int_{0}^{s_1} \lambda_{j_0}(u) \mathrm{d}u \right)}\mathrm{d}s_1\\
&\quad+\sum_{j_2\ne j_1}\sum_{j_1\ne j_0}\sum_{j_0} \ketbra{\chi_{j_2}(1)}{\chi_{j_0}(0)}\int_{0}^{1} \int_{0}^{s_2} \beta_{j_2,j_1}(s_2)\beta_{j_1,j_0}(s_1)e^{-iT \left(\int_{s_2}^1 \lambda_{j_2}(u) \mathrm{d}u+\int_{s_1}^{s_2} \lambda_{j_1}(u) \mathrm{d}u+\int_{0}^{s_1} \lambda_{j_0}(u) \mathrm{d}u \right)}\mathrm{d}s_2\mathrm{d}s_1 \\
&\quad+\cdots
\end{align*}
where 
$$
\beta_{jk}(s) := \lim_{\delta \rightarrow 0} \frac{\braket{\chi_j(s+\delta)}{\chi_k(s)}}{\delta}= \frac{\bra{\chi_j(s)}\dot{\hat{H}}(s) \ket{\chi_k(s)}}{\lambda_j(s) - \lambda_k(s)} \label{eq:beta}
$$
\end{lemma}
This form of the path integral expansion is useful here because it categorizes the paths by the number of jumps that the path makes between the the initial and final instantaneous eigenstate.  Specifically, the first term consists of no jumps and corresponds to all states remaining in the instantaneous eigenbasis throughout the evolution.  This in physical terms describes a perfect adiabatic evolution with no transitions in the state.  The next term describes the sum over all paths that transition at time $t_1$ from instantaneous eigenstate $j_0$ to $j_1$ and the following term describes the sum over all paths that  jump twice.  In cases where the gaps are sufficiently large relative to the rate of change of the Hamiltonian, the resultant series takes the form of a geometric series and thus can be truncated at a finite number of jumps without incurring substantial error.

The error due to truncating these integrals is given by the following lemma.
\begin{lemma}\label{lemma:Cbd}
    Let $\gamma_{\min}$ be the minimum value of $|\lambda_j(s) -\lambda_k(s)|$ and let $$\aleph = \frac{6 \max\|\partial_s \hat{H} \|^3}{\gamma_{\min}^3}+\frac{ \max\|\partial_s \hat{H} \|\max\|\partial_s^2 \hat{H} \|+ 2\max\|\partial_s \hat{H} \|^2}{\gamma_{\min}^2}$$
    then for $m\ge 0$ denote $C_p$ to be the sum over all $p$-jump paths.  We then have that for any positive integer $m$
    \begin{align*}
        \|C_{2m}\| &\le \frac{\aleph^m}{m!(\gamma_{\min} T)^m}\\
        \|C_{2m+1}\| &\le \frac{\aleph^m}{m!(\gamma_{\min} T)^m} \left(\frac{\aleph}{\max \|\partial_s \hat{H}(t)\| T} \right)
    \end{align*}
\end{lemma}
\begin{proof}
    The result is shown directly in the proof of Lemma 2 of~\cite{cheung2011improved} in (51) and (54).
\end{proof}
This shows that provided the gap is sufficiently large, we can guarantee the contribution to the path integral from terms with $3,4,5,\ldots$ jumps are small.  This justifies truncating the path integral representation at a finite number of jumps in cases where the gap is large enough so that $\aleph /{\gamma_{\min} T}$ is small.

In cases where the number of jumps are suppressed for a path by the gap, it is often convenient to use approximations to the path integral that arise from integration by parts or the stationary phase approximation (if the derivative of the eigenvalue gap evaluates to zero).  

In general, if the above nearly adiabatic evolution is not sufficient to describe the dynamics of the system we can evaluate the integrals directly.
\begin{lemma}[Integration Lemma] \label{lemma:integration}
Under the above assumptions that the Hamiltonian is smooth and all transitions are gapped by $\gamma_{\min}$, we have that for $\Gamma=O\left(\frac{\max\{\|\dot{\hat{H}}\|,\|\ddot{\hat{H}}\|, \|\dddot{\hat{H}}\|\}}{\gamma_{\min}}  \right)$, the time-evolution operator can be expressed as
\begin{align*}
    \mathcal{T} e^{-i T \int_0^1 \hat{H}(s) \mathrm{d}s} =& \sum_{j_0} \ketbra{\chi_{j_0}(1)}{\chi_{j_0}(0)}e^{-i T\int_0^1\lambda_{j_0}(u) \mathrm{d}u}\left( 1 + \sum_{j_1\ne j_0}\int_0^1 \frac{\beta_{j_0,j_1}(s_2) \beta_{j_1,j_0}(s_2)}{iT \gamma_{j_1,j_0}(s_2)}ds_2\right)\nonumber\\
    &\qquad+\sum_{j_1 \ne j_0}e^{-i T\int_0^1 \lambda_{j_0}(u) \mathrm{d}u}\ketbra{j_1(1)}{j_0(0)}\left.\left(\frac{\beta_{j_1,j_0}(s) e^{-iT\int_0^s \gamma_{j_1,j_0}(u) \mathrm{d}u}}{-iT \gamma_{j_1,j_0}(s)}\right)\right|_{s=0}^1 + O\left(\frac{\Gamma^4}{\gamma_{\min}^2T^2} \right)
\end{align*}
where $\gamma_{j_1, j_0} (s)$ is the eigenvalue gap at time $s$ between levels $j_1$ and $j_0$
\end{lemma}
\begin{proof}
In order to derive an expression for the matrix elements of the time evolution operator we need to evaluate the leading order scaling of the dominant path integrals in the evolution operator.  The dominant contribution in the limit of large $T$ comes from the one-jump paths and a subset of two-jump paths where the paths return to the initial state.

Let $\mathcal{T} e^{-i T \int_0^1 \hat{H}(s) \mathrm{d}s} = C_0 + \mathcal{E}$, where $C_0$ is the sum over all zero-jump paths and $\mathcal{E}$ is all the remaining terms. From Theorem 4 of~\cite{cheung2011improved},
\begin{align}
    &\left\|\sum_{j_0\ne j_1} \ketbra{j_1(1)}{j_0(0)}\bra{j_1(1)}\mathcal{E}\ket{j_0(0)} - \sum_{j_1 \ne j_0}e^{-i T\int_0^1 \lambda_{j_0}(u) \mathrm{d}u}\ketbra{j_1(1)}{j_0(0)}\left.\left(\frac{\beta_{j_1,j_0}(s) e^{-iT\int_0^s \gamma_{j_1,j_0}(u) \mathrm{d}u}}{-i\gamma_{j_1,j_0}(s) T}\right)\right|_{s=0}^1\right\|\nonumber\\
    &\qquad\in O\left(\frac{\Gamma^4}{\gamma_{\min}^2 T^2} \right),
\end{align}
where 
\begin{equation}
    \Gamma=O\left(\frac{\max\{\|\dot{\hat{H}}\|,\|\ddot{\hat{H}}\|, \|\dddot{\hat{H}}\|\}}{\gamma_{\min}}  \right)
\end{equation}
This previous work shows that we can approximate the dynamics of all one-jump paths.  This provides us with much of the material that we need to approximate the off-diagonal matrix elements of the evolution operator.

While this result is sufficient for bounding the error in the adiabatic approximation, it is insufficient here because it does not discuss the construction of the diagonal elements.  The diagonal element calculation involves us considering two-jump terms that return to the initial eigenvector.

\begin{align}
    \sum_{j_1\ne j_0}\sum_{j_0} e^{-i T\int_0^1\lambda_{j_0}(u) \mathrm{d}u}\ketbra{\chi_{j_0}(1)}{\chi_{j_0}(0)}\int_{0}^{1} \int_{0}^{s_2} \beta_{j_0,j_1}(s_2)\beta_{j_1,j_0}(s_1)e^{-iT \int_{s_1}^{s_2} \gamma_{j_1,j_0}(u) \mathrm{d}u }\mathrm{d}s_2\mathrm{d}s_1 \label{eq:2jump}
\end{align}
We can evaluate the integral using integration by parts again as per the previous analysis.  However, something interesting happens with the integrals because of a phase cancellation effect.
\begin{align}
    &\int_{0}^{1} \int_{0}^{s_2} \beta_{j_0,j_1}(s_2)\beta_{j_1,j_0}(s_1)e^{-iT \int_{s_1}^{s_2} \gamma_{j_1,j_0}(u) \mathrm{d}u }\mathrm{d}s_2\mathrm{d}s_1 \nonumber\\
    &= \int_{0}^{1} \beta_{j_0,j_1}(s_2) e^{-i T \int_0^{s_2} \gamma_{j_1,j_0}(u) \mathrm{d}u} \int_{0}^{s_2} \beta_{j_1,j_0}(s_1)e^{iT \int_{0}^{s_1} \gamma_{j_1,j_0}(u) \mathrm{d}u }\mathrm{d}s_2\mathrm{d}s_1 \nonumber\\
    &= \int_{0}^{1} \beta_{j_0,j_1}(s_2) e^{-i T \int_0^{s_2} \gamma_{j_1,j_0}(u) \mathrm{d}u} \Biggr(\left. \frac{\beta_{j_1,j_0}(s_1)e^{iT \int_{0}^{s_1} \gamma_{j_1,j_0}(u) \mathrm{d}u }}{i\gamma_{j_1,j_0}(s_1) T}\right|_{s_1=0}^{s_2} \nonumber\\
    &\qquad-\int_{0}^{s_2} \left( \frac{\partial}{\partial_{s_1}}\frac{\beta_{j_1,j_0}(s_1)}{i\gamma_{j_1,j_0}(s_1) T}\right)e^{iT\int_0^{s_1} \gamma_{j_1,j_0}(u) \mathrm{d}u}{d}s_1 \Biggr)ds_2\nonumber\\
    &= \int_0^1 \frac{\beta_{j_0,j_1}(s_2) \beta_{j_1,j_0}(s_2)}{i\gamma_{j_1,j_0}(s_2)T}ds_2 - \int_0^1 \frac{\beta_{j_0,j_1}(s_2) \beta_{j_1,j_0}(s_2)e^{-iT\int_0^{s_2} \gamma_{j_1,j_0}(u) \mathrm{d}u}}{i\gamma_{j_1,j_0}(0)T}ds_2\nonumber\\
    &\qquad - \int_0^1 \beta_{j_0,j_1}(s_2) e^{-iT\int_0^{s_2} \gamma_{j_1,j_0}(u) \mathrm{d}u}\int_{0}^{s_2} \left( \frac{\partial}{\partial_{s_1}}\frac{\beta_{j_1,j_0}(s_1)}{i\gamma_{j_1,j_0}(s_1) T}\right)e^{iT\int_0^{s_1} \gamma_{j_1,j_0}(u) \mathrm{d}u}{d}s_1 ds_2\nonumber\\
    &=\int_0^1 \frac{\beta_{j_0,j_1}(s_2) \beta_{j_1,j_0}(s_2)}{i\gamma_{j_1,j_0}(s_2)T}ds_2 -\left.\frac{\beta_{j_0,j_1}(s_2) \beta_{j_1,j_0}(s_2)e^{-iT\int_0^{s_2} \gamma_{j_1,j_0}(u) \mathrm{d}u}}{\gamma_{j_1,j_0}(0)\gamma_{j_1,j_0}(s_2)T^2}\right|_0^1 \nonumber\\
    &\qquad + \int_0^1 \partial_{s_2}\left(\frac{\beta_{j_0,j_1}(s_2) \beta_{j_1,j_0}(s_2)}{\gamma_{j_1,j_0}(0)\gamma_{j_1,j_0}(s_2)T^2}\right)e^{-iT\int_0^{s_2} \gamma_{j_1,j_0}(u) \mathrm{d}u} ds_2  \nonumber\\
&\qquad-\left.\frac{\beta_{j_0,j_1}(s_2) e^{-iT\int_0^{s_2} \gamma_{j_1,j_0}(u) \mathrm{d}u}}{-iT\gamma_{j_1,j_0}(s_2)}\int_{0}^{s_2} \left( \frac{\partial}{\partial_{s_1}}\frac{\beta_{j_1,j_0}(s_1)}{i\gamma_{j_1,j_0}(s_1) T}\right)e^{iT\int_0^{s_1} \gamma_{j_1,j_0}(u) \mathrm{d}u}{d}s_1 \right|_{s_2=0}^1\nonumber\\
&\qquad+\int_0^1\frac{\partial}{\partial s_2}\left(\frac{\beta_{j_0,j_1}(s_2) }{-iT\gamma_{j_1,j_0}(s_2)}\int_{0}^{s_2} \left( \frac{\partial}{\partial_{s_1}}\frac{\beta_{j_1,j_0}(s_1)}{i\gamma_{j_1,j_0}(s_1) T}\right)e^{iT\int_0^{s_1} \gamma_{j_1,j_0}(u) \mathrm{d}u}{d}s_1\right) e^{-iT\int_0^{s_2} \gamma_{j_1,j_0}(u) \mathrm{d}u} ds_2
\end{align}
This shows that the dominant term at $O(1/T)$ because of phase interference between the paths that jump twice over the interval.  We can use this reasoning to find an asymptotic expansion to the two-jump terms in~\eqref{eq:2jump} that can be shown using the Cauchy-Schwarz inequality to be of the form 
\begin{equation}
    \sum_{j_0} e^{-i T\int_0^1\lambda_{j_0}(u) \mathrm{d}u}\ketbra{\chi_{j_0}(1)}{\chi_{j_0}(0)}\sum_{j_1\ne j_0}\int_0^1 \frac{\beta_{j_0,j_1}(s_2) \beta_{j_1,j_0}(s_2)}{i\gamma_{j_1,j_0}(s_2)T}ds_2 + O\left(\frac{\Gamma^2+ \Gamma^3 + \Gamma^4}{\gamma_{\min}^2T^2} \right) 
\end{equation}
From Theorem 4 of~\cite{cheung2011improved}, the sum over all paths in $C_2$ that jump back to $j_0$ is in $O(\Gamma^4/\gamma_{\min}^2T^2)$ and so are comparable to the higher order terms neglected in the path integral decomposition.
Thus we have from the triangle inequality that 
\begin{align}
    &\mathcal{T} e^{-i T \int_0^1 \hat{H}(s) \mathrm{d}s} = \sum_{j_0} \ketbra{\chi_{j_0}(1)}{\chi_{j_0}(0)}e^{-i T\int_0^1\lambda_{j_0}(u) \mathrm{d}u}\left( 1 + \sum_{j_1\ne j_0}\int_0^1 \frac{\beta_{j_0,j_1}(s_2) \beta_{j_1,j_0}(s_2)}{iT \gamma_{j_1,j_0}(s_2)}ds_2\right)\nonumber\\
    &\quad+\sum_{j_1 \ne j_0}e^{-i T \int_0^1 \lambda_{j_0}(u) \mathrm{d}u}\ketbra{j_1(1)}{j_0(0)}\left.\left(\frac{\beta_{j_1,j_0}(s) e^{-iT\int_0^s \gamma_{j_1,j_0}(u) \mathrm{d}u}}{-i T \gamma_{j_1,j_0}(s)}\right)\right|_{s=0}^1 + O\left(\frac{\Gamma^4}{\gamma_{\min}^2T^2} \right) + \sum_{m\ge 3} \|C_m\|~\label{eq:pathintApprox}
\end{align}
Next note from Lemma~\ref{lemma:Cbd} $\aleph \in O(\Gamma^3)$ which implies that 
\begin{equation}
    \sum_{m\ge 3} \|C_m\| \in O\left(\frac{\Gamma^5}{\gamma_{\min}^2 T^2} \right).
\end{equation}
The above sum is asymptotically smaller than the lower-order contributions and thus can be neglected from the sum.  The result then follows immediately from~\eqref{eq:pathintApprox}.
\end{proof}
The above result shows that if $T$ is sufficiently large, then we can simplify the integral over our paths to only consider the end points of the interval.  This procedure can be iterated many times to achieve a higher and higher order expression for the value of the path integral resulting in an expression that will depend on the value of the $\beta,\gamma$ and their derivatives near the boundary of the region of integration.

 Similar to the formulation in the previous section, we will divide the total evolution time into time steps and assume oracle access to the amplitudes and phases of the overlaps between eigenstates of adjacent Hamiltonian terms, as well as the eigenvalues as a phase. However, here, instead of ``adjacent'' terms being adjacent in the Hamiltonian decomposition in Definition \ref{def:Ham}, they will be the Hamiltonian evaluated at adjacent time steps. Let $\ell \in \{0, \ldots, r\}$ denote the time step number at reduced time $\frac{\ell}{r} \in \{0, \frac{1}{r}, \frac{2}{r}, \ldots, 1 \}$. We will use an extra register to store $\ell$, similar to how an extra register was used to store the Hamiltonian term number in the time-independent case. Thus, a state $\ket{\chi_j \left(\frac{\ell}{r} \right)}$ can be represented in $\log (r+1) + n$ qubits as $\ket{\ell} \ket{j}$.

\begin{definition}[$d$-sparsity for time-dependent Hamiltonians]
Let $\hat{H}$ be a Hamiltonian as defined in Lemma \ref{lemma:TimeDepHam}. We refer to the Hamiltonian as $d$-sparse if, for each eigenvector $\ket{\chi_{j_0} \left( \frac{\ell}{r} \right)}$ where $0 \leq \ell < L-1$, there are at most $d$ values of $j_1$ such that $\beta_{j_0, j_1} \left( \frac{\ell}{r} \right) \neq 0$, and for each eigenvector $\ket{\chi_{j_1} \left( \frac{\ell}{r} \right)}$, there are at most $d$ values of $j_0$ such that $\beta_{j_0, j_1} \left( \frac{\ell}{r} \right) \neq 0$. \label{def:dsparse2}
\end{definition}

Our first goal is to consider the path integral expansion for the time-ordered operator exponential truncated at $O(1/T^{m})$.  First, we see from Lemma~\ref{lemma:Cbd} that we can neglect any path that has at least $2m$ jumps at this order.  Similarly, each path with $k>0$ jumps needs to be approximated using integration by parts to a similar order.  It is straightforward to verify that if integration by parts is applied $2p$ times then we generate a term of order $O(1/T^p)$ and thus it suffices to apply integration by parts a total of $2m$ times for each of the sets of integrals, distributing the $2m$ integrals over each of the $k$ jump paths that we are interested in.  From the multinomial theorem implies that the maximum number of such terms that can arise is $2^{2m}$ for any $k\in \{1, \ldots, 2m\}$.  As we are treating $m$ as a constant, we can ignore the growth of the number of terms with $m$ due to the growth of terms that arise from the product rule.  Taking
\begin{equation}
    \Gamma \ge \sup_p \left( \partial_s^p \hat{H}(s) \right)/\gamma_{\min},
\end{equation}
if we denote the truncated expansion of the ordered operator exponential to be ${\rm Trunc}(\mathcal{T} e^{-i T \int_0^1 \hat{H}(s) ds})$ then the results of Theorem 4 of~\cite{cheung2011improved} shows that
\begin{equation}
    \|\mathcal{T} e^{-iT \int_0^1 \hat{H}(s) ds } - {\rm Trunc}(\mathcal{T} e^{-i T \int_0^1 \hat{H}(s) ds}) \| = O\left(\frac{\Gamma^{2m}}{\gamma_{\min}^m T^m}\right)
\end{equation}
As $m$ is a constant, the total error and in turn the truncation error in the path integral can be made to scale arbitrarily well with $T$ for constant but large $m$. Here, since we have $T \in \Omega \left( \frac{\Gamma^2}{\gamma_{\min} \sqrt{\epsilon}} \right)$ for some given error tolerance $\epsilon$, we can truncate the path integral at $m=2$.

Thus, we can approximate the long-time evolution using sums of paths that do not transition, paths that transition once only, and paths that transition twice but end up back on the original path. Specifically, the time evolution operator that we will approximate is
\begin{align}
    &\sum_{j_0} \ketbra{\chi_{j_0}(1)}{\chi_{j_0}(0)}e^{-i T\int_0^1\lambda_{j_0}(u) \mathrm{d}u}\left( 1 + \sum_{j_1\ne j_0}\int_0^1 \frac{\beta_{j_0,j_1}(s') \beta_{j_1,j_0}(s')}{iT \gamma_{j_1,j_0}(s')} \mathrm{d} s'\right)\nonumber\\
    &\quad+\sum_{j_1 \ne j_0}e^{-i T \int_0^1 \lambda_{j_0}(u) \mathrm{d}u}\ketbra{j_1(1)}{j_0(0)}\left.\left(\frac{\beta_{j_1,j_0}(s) e^{-iT\int_0^s \gamma_{j_1,j_0}(u) \mathrm{d}u}}{-i T \gamma_{j_1,j_0}(s)}\right)\right|_{s=0}^1, \label{eq:TDOperator}
\end{align}
which Lemma \ref{lemma:integration} tells us is equal to the exact time evolution operator $\mathcal{T} e^{-i T \int_0^1 \hat{H}(s) \mathrm{d}s}$ up to error $O\left( \frac{\Gamma^4}{\gamma_{\rm min}^2 T^2} \right)$.

In the discretized time setting, denoting the time step number by $\ell \in \{0, \ldots, r\}$ to represent reduced times $\{ 0, \frac{\ell}{r}, \frac{2}{r} , \ldots, 1 \}$ and replacing the integrals with the trapezoidal Riemann sum, the discrete approximate time evolution operator is
\begin{align}
    &\bar{U}:=  \sum_{j_0} \ketbra{\chi_{j_0}(1)}{\chi_{j_0}(0)} e^{-i \frac{T}{r} \left(  \frac{1}{2}\lambda_{j_0} (0) + \frac{1}{2} \lambda_{j_0} (1) + \sum_{\ell=1}^{r-1} \lambda_{j_0} \left( \frac{\ell}{r} \right) \right)} \\
    & \quad + \sum_{j_0} \ketbra{\chi_{j_0}(1)}{\chi_{j_0}(0)} e^{-i \frac{T}{r} \left(  \frac{1}{2}\lambda_{j_0} (0) + \frac{1}{2} \lambda_{j_0} (1) + \sum_{\ell=1}^{r-1} \lambda_{j_0} \left( \frac{\ell}{r} \right) \right)}\sum_{j_1\ne j_0} \frac{1}{r} \Bigg( \frac{1}{2} \frac{\beta_{j_0,j_1} \left( 0 \right) \beta_{j_1,j_0}\left( 0 \right)}{iT \gamma_{j_1,j_0}\left( 0 \right)} + \frac{1}{2} \frac{\beta_{j_0,j_1} \left( 1 \right) \beta_{j_1,j_0}\left( 1 \right)}{iT \gamma_{j_1,j_0}\left( 1 \right)} \nonumber \\
    & \quad + \sum_{\ell=1}^{r-1} \frac{\beta_{j_0,j_1} \left( \frac{\ell}{r} \right) \beta_{j_1,j_0}\left( \frac{\ell}{r} \right)}{iT \gamma_{j_1,j_0}\left( \frac{\ell}{r} \right)}\Bigg) \Biggr)\\
    & \quad + \sum_{j_1 \ne j_0} \ketbra{j_1(1)}{j_0(0)} e^{-i \frac{T}{r} \left( \frac{1}{2} \lambda_{j_0}(0) + \frac{1}{2} \lambda_{j_0} (1) + \sum_{\ell=1}^{r-1} \lambda_{j_0} \left(\frac{\ell}{r} \right) \right)} \Bigg(\frac{i \beta_{j_1,j_0}(1) e^{-i \frac{T}{r} \left( \frac{1}{2} \gamma_{j_1, j_0}(0) + \frac{1}{2} \gamma_{j_1, j_0} (1) + \sum_{\ell=1}^{r-1} \gamma_{j_1, j_0} \left( \frac{\ell}{r} \right) \right) }}{T \gamma_{j_1,j_0}(1)} \nonumber\\
    & \quad - \frac{i \beta_{j_1, j_0} (0)}{T \gamma_{j_1, j_0} (0)} \Bigg)
\end{align}
We can approximately implement $\bar{U}$ using the techniques from Section \ref{sec:HamPathInt}. The implementation of the long-time path integral is very similar to the short-time path integral, so we will not go through all the entire algorithm again here. Instead, see Appendix \ref{sec:LTPIAppendix} for details. The main result for the long-time path integral is the following.

\begin{theorem}
    Let $\hat{H}(t)$ be a time-dependent Hamiltonian satisfying the $d$-sparsity definition in Def. \ref{def:dsparse2}, with a minimum gap of $\gamma_{\min}$ between any two eigenvalues, and let $\Gamma \in O \left( \frac{\max \{ \| \dot{\hat{H}} \|, \| \ddot{\hat{H}} \|, \| \dddot{\hat{H}} \| \} }{\gamma_{\min}} \right)$. Then, we can approximately implement the time evolution operator $\mathcal{T} e^{-i T \int_0^1 \hat{H}(s) \mathrm{d} s}$ up to error $\epsilon$ for $T \in \Omega  \left( \frac{\Gamma^2}{\gamma_{\min} \sqrt{\epsilon}} \right)$ using $O \left( \frac{d^2}{\sqrt{\epsilon}} \sqrt[3]{\frac{\Gamma^2 \Lambda}{\gamma_{\min}}} \right)$ queries to the oracles defined in Appendix \ref{sec:LTPIAppendix}.
\end{theorem}

\section{Discussion}
We have introduced three quantum algorithms inspired by the path integral formulation--the short-time Hamiltonian path integral algorithm, the long-time Hamiltonian path integral algorithm, and the Lagrangian path integral algorithm--and derived the query and gate complexity scalings for each.

The Lagrangian path integral is perhaps the most interesting algorithm of the three because it does not require any information about a Hamiltonian. Although the derivation begins with a Hamiltonian, the derived expression resembles the discretized Feynman path integral with a Lagrangian, and that expression can be directly simulated without any assumptions about any underlying Hamiltonian. In quantum field theory, the Lagrangian is often treated as the starting point of the theory rather than something derived from the Hamiltonian. Here, we assume that we have a Lagrangian that is quadratic in the velocity and a potential that only depends on the position. This is standard for a single-particle quantum mechanical system, but the Lagrangian may take a different form for other types of systems. One future direction will be trying to apply the Lagrangian path integral algorithm for simulating quantum field theories.  Further, applications of these ideas to constrained quantum dynamics could also be a significant application of the Lagrangian simulation methods presented here.

The short-time Hamiltonian path integral is a path integral reformulation of the higher-order Trotter product formula, but with a different access model. Here, we assumed access to the inner products between eigenstates of the Hamiltonian terms, and the overall complexity scales quadratically with $d$, the maximum number of non-zero inner products between the eigenstates of two adjacent Hamiltonian terms. The notion of sparsity here differs from the usual notion of row-sparsity or locality that is often used in Hamiltonian simulation. For example, consider a Hamiltonian written as a linear combination of $n$-qubit Pauli operators,
\begin{align}
    \hat{H} = \sum_{\ell=0}^{L-1} \hat{H}_\ell = \sum_{\ell=0}^{L-1} h_\ell \bigotimes_{q=0}^{n-1} \hat{P}^{(\ell)}_q
\end{align}
where $\hat{P}^{(\ell)}_q$ is a Pauli operator $\hat{X}$, $\hat{Y}$, or $\hat{Z}$, or identity operator $\openone$, acting on the $q^{\rm th}$ qubit. For two adjacent terms $\hat{H}_\ell$ and $\hat{H}_{\ell+1}$, if $[\hat{P}^{(\ell)}_q, \hat{P}_q^{(\ell+1)}] \neq 0$ for $d'$ values of $q \in \{0, \ldots, n-1\}$, we can write the bases for $\hat{H}_\ell$ and $\hat{H}_{\ell+1}$ so that each basis of $\hat{H}_\ell$ has a non-zero inner product with $2^{d'}$ basis states of $\hat{H}_{\ell+1}$, and vice versa. Specifically, we can use basis vectors $\{\ket{b^{(\ell)}_{0, \pm}} \ket{b^{(\ell)}_{1, \pm}} \ldots \ket{b^{(\ell)}_{n-1, \pm}}\}$ so that $\ket{b^{(\ell)}_{q, +}}$ is the $+1$ eigenstate of $\hat{P}_q^{(\ell)}$ and $\ket{b^{(\ell)}_{q, -}}$ is the $-1$ eigenstate if $\hat{P}_q^{(\ell)}$ is a Pauli operator. If $\hat{P}_q^{(\ell)}$ is the identity operator, we can choose $\ket{b_{q,\pm}^{(\ell)}}$ to be the eigenstates of $\hat{P}_q^{(\ell+1)}$ instead, since any single-qubit orthonormal basis is an eigenbasis for the identity operator (if $\hat{P}_q^{(\ell+1)}$ is also the identity, we can choose the $\hat{P}_q^{(\ell-1)}$ basis or just the $Z$ basis). Then, defining the basis for $\hat{P}^{(\ell+1)}$ similarly, there are $2^{d'}$ non-zero inner products with $\hat{H}_{\ell+1}$ basis states for each $\hat{H}_\ell$ basis state, and $2^{d'}$ non-zero inner products with $\hat{H}_\ell$ basis states for each $\hat{H}_{\ell+1}$ basis state. In the most general case, where all the $\hat{P}_q^{(\ell)}$ and $\hat{P}_q^{(\ell+1)}$ terms can be non-commuting, we get up to $2^n$ non-zero inner products. However, if the two Hamiltonian terms mostly commute qubitwise, $d'$ can be a small constant.

Of course, this is only the number of non-zero inner products between two of the Hamiltonian terms, and we need to consider the maximum value of $d'$ between all pairs of adjacent terms in the Hamiltonian decomposition. The ordering of the terms, as well as which basis we pick for the identity operators, will make a difference in what value of $d$ we can attain for the decomposition; we have not focused on finding ways of decomposing Hamiltonians in this work. It is worth noting that the final overall complexity scales quadratically with $d$ but only linearly in $L$, so it may be more optimal to decompose the Hamiltonian into terms with more qubitwise-commuting parts, at the cost of increasing the total number of terms. This algorithm can be thought of as a good generic algorithm that can be used to simulate any quantum algorithm using the Hamiltonian formalism, and unlike conventional Hamiltonian simulation approaches it does not require that we explicitly diagonalize the Hamiltonian terms in the Hamiltonian simulation which may have advantages.

The long-time Hamiltonian simulation algorithm on the other hand is the opposite of the previous one.  It is unlikely to be a very practical algorithm in cases where the system gaps are small, but in the case of near-adibatic evolution it can achieve better scaling than methods such as the truncated Dyson series simulation method~\cite{low2018hamiltonian,kieferova2019simulating}.  This work is relevant because it shows that if we have transitions that are large relative to the energy scale of the system that we're interested in, then we can use integration by parts to approximately sum over all such paths and avoid needing to perform each segment of the sum using LCU.  Further, this suggests that by intelligently breaking up the paths the comprise the time evolution into groups that can be analytically summed then we can substantially reduce the costs and further rely on methods like the short-time evolution simulation method to handle the sum over all paths excluding the long-time paths handled using the aforementioned method.

Perhaps the most significant research direction revealed by our work involves optimizing the Lagrangian simulation approach. 
 While our approach for the Lagrangian simulation method provides polynomial scaling, it fails to provide near-linear scaling with the evolution time $T$.  This is a consequence of the use of the lowest-order Trotter formula and likely can be improved; although the improvements are unlikely to follow by simply replacing the formulas with their higher order counterparts because of challenges posed by the Gauss sum formulas.  On a related note, the use of Gauss sum formulas here ties our spatial and temporal resolution for the quantum simulation, which means that our Lagrangian path integral methods can only approach zero error in the continuum limit.  Identifying a discrete analogue of this would be an important step forward because it would allow us to apply these ideas for discrete quantum systems rather than only continuous systems.

\section*{Acknowledgements}
SS acknowledges the Canadian National Science and Engineering Research Council (NSERC) for support.  NW would like to acknowledge funding for this work from Google Inc. This material is primarily supported
by the U.S. Department of Energy, Office of Science, National Quantum Information Science Research Centers, Co-
design Center for Quantum Advantage (C2QA) under contract number DE-SC0012704 (PNNL FWP 76274).

\bibliographystyle{unsrt}
\bibliography{references}

\appendix

\section{Derivation of Feynman path integral in the continuum} \label{sec:appendix_a}

Here for completeness we provide a derivation of the Feynman path integral in the continuum from the Hamiltonian.  We add this to contrast with the discrete derivative given above.
We will assume that we have a Hamiltonian of the standard form,
\begin{equation}
    \hat{H} = \frac{\hat{p}^2}{2m} + \hat{V}(x)
\end{equation}
where $\hat{p}$ is the momentum operator, $m$ is the mass, and $\hat{V}$ is a potential that only depends on the position, $x$. The solution to the Schrodinger equation is
\begin{equation}
    \ket{\psi(t)} = e^{-i \hat{H} t} \ket{\psi(0)}
\end{equation}
For any two particular position eigenstates $\ket{x_a}$ and $\ket{x_b}$, the transition amplitude of $e^{-i \hat{H} t}$ between the two states is
\begin{align}
    \bra{x_b} e^{-i \hat{H}t} \ket{x_a} &= \bra{x_b} e^{-i \left( \frac{\hat{p}^2}{2m} + \hat{V} \right) t} \ket{x_a}
\end{align}
We can divide the total time evolution into $r$ time slices. The evolution of the system for time $t$ is equivalent to applying the time evolution operator $r$ times, each for time $\frac{t}{r}$:
\begin{align}
    \bra{x_b} e^{-i \hat{H}t} \ket{x_a} &= \bra{x_b} \left( e^{-i \left( \frac{\hat{p}^2}{2m} + \hat{V} \right) \frac{t}{r}} \right)^r \ket{x_a}
\end{align}
If we look at just one time slice, and apply the Baker-Campbell-Hausdorff formula, we get
\begin{align}
    e^{-i \left( \frac{\hat{p}^2}{2m} + \hat{V} \right) \frac{t}{r}} &= \bra{x_b} e^{-i \frac{\hat{p}^2}{2m} \frac{t}{r}} e^{-i \hat{V} \frac{t}{r}} e^{-i \frac{t^2}{2r^2} \left[ \frac{\hat{p}^2}{2m}, \hat{V} \right] - \frac{t^3}{r^3} \left( \frac{1}{6} \left[ \hat{V}, \left[ \hat{V}, \frac{\hat{p}^2}{2m} \right] \right] - \frac{1}{3} \left[ \left[ \hat{V}, \frac{\hat{p}^2}{2m} \right], \frac{\hat{p}^2}{2m} \right] \right) + \ldots} \ket{x_a}
\end{align}
where the $\ldots$ indicates terms of higher order of $\frac{t}{r}$. If all the nested commutators of $\frac{\hat{p}^2}{2m}$ and $\hat{V}$ are bounded, the third exponential term goes to the identity in the $r \rightarrow \infty$ limit, and we get
\begin{align}
    \bra{x_b} e^{-i \left( \frac{\hat{p}^2}{2m} + \hat{V}(x) \right) t} \ket{x_a} &= \lim_{r \rightarrow \infty} \bra{x_b} \left( e^{-i \frac{\hat{p}^2}{2m} \frac{t}{r}} e^{-i \hat{V} \frac{t}{r}} \right)^r \ket{x_a} \label{path_integral_1}
\end{align}
The above equation does not hold for all potentials $\hat{V}(x)$; in fact, there are many useful ones, such as the Coulomb potential, for which nested commutators are not bounded and the convergence of the Baker-Campbell-Hausdorff expansion cannot be proven. There are ways of dealing with these potentials, such as by using the Duru-Kleinert transformation \cite{kleinert1990}, but here, for simplicity, we will assume that we are working with a nicely-behaved system where everything is bounded and \eqref{path_integral_1} holds.

Between each time slice, we will insert a complete set of intermediate states, which resolve to the identity:
\begin{equation}
    \int d x \ketbra{x}{x} = \openone
\end{equation}
Then, \eqref{path_integral_1} becomes
\begin{align}
    \lim_{r \rightarrow \infty} \int dk_1 \int dk_2 ... \int dk_{r-1}\bra{x_b} x_{r-1} \rangle \bra{x_{r-1}} e^{-i \frac{\hat{p}^2}{2m} \frac{t}{r}} e^{-i \hat{V}  \frac{t}{r}} \ket{x_r} \bra{x_r} e^{-i \frac{\hat{p}^2}{2m} \frac{t}{r}} e^{-i \hat{V}  \frac{t}{r}}  \ldots e^{-i \frac{\hat{p}^2}{2m} \frac{t}{r}} e^{-i \hat{V}  \frac{t}{r}} \ket{x_1} \langle x_1 \ket{x_a} \label{path_integral_2}
\end{align}
We have a product of $r$ factors, each of the form
\begin{align}
    \bra{x_{k+1}} e^{-i \frac{\hat{p}^2}{2m} \frac{t}{r}} e^{-i \hat{V} \frac{t}{r}} \ket{x_k} \label{one_ts}
\end{align}
where $x_0 \equiv x_a$ and $x_r \equiv x_b$, integrated over $x_k$ for all $k \in \{1, \ldots , r-1 \}$. Each $\ket{x_k}$ is the state at the end of the $k$th time slice, and the expression in \eqref{one_ts} is just the transition amplitude for the time evolution during that time slice. Thus, this expression tells us that the transition amplitude for the total time evolution is the sum of transition amplitudes for all possible paths $(x_a, x_1, \ldots, x_{r-1}, x_b )$. Now, if we insert another resolution of the identity between $e^{-i \frac{\hat{p}^2}{2m} \frac{t}{r}}$ and $e^{-i \hat{V} \frac{t}{r}}$, \eqref{one_ts} becomes
\begin{align}
    \int dp_k \bra{x_{k+1}} e^{-i \frac{\hat{p}^2}{2m} \frac{t}{r}} \ket{p_k} \bra{p_k} e^{-i \hat{V} \frac{t}{r}} \ket{x_k} &= \int dp_k \bra{x_{k+1}} e^{-i \frac{p_k^2}{2m} \frac{t}{r}} \ket{p_k} \bra{p_k} e^{-i V(x_k) \frac{t}{r}} \ket{x_k}\\
    &= \int dp_k \bra{x_{k+1}} p_k \rangle \bra{p_k} x_k \rangle e^{-i \frac{p_k^2}{2m} \frac{t}{r}}  e^{-i V(x_k)\frac{t}{r}}\\
    &= \frac{1}{2 \pi} \int d p_k \; e^{i p_k x_{k+1}} e^{i p_k x_k} e^{-i \frac{p_k^2}{2m} \frac{t}{r}} e^{-i V(x_k)\frac{t}{r}} \\
    &= \frac{1}{2 \pi} e^{-i V(x_k) \frac{t}{r}} \int d p_k \; e^{-i \frac{p_k^2}{2m} \frac{t}{r} + i p_k (x_{k+1} - x_k)} \\
    &= \sqrt{\frac{-i m r}{2 \pi t}} e^{i \left( \frac{m r}{2 t} (x_{k+1} - x_k)^2 - V(x_k) \frac{t}{r} \right)}\\
    &= \sqrt{\frac{-i m r}{2 \pi t}} e^{i \left( \frac{mt}{2r} \left( \frac{x_{k+1} - x_k}{\frac{t}{r}} \right)^2 - V(x_k) \frac{t}{r} \right)}
\end{align}
If we put this back into \eqref{path_integral_2}, then we get
\begin{align}
    & \lim_{r \rightarrow \infty} \left( \frac{-imr}{2 \pi t} \right)^{r/2} \int dk_1 ... \int dk_{r-1} \; e^{i \left( \frac{mt}{2r} \left( \frac{x_r - x_{r-1}}{\frac{t}{r}} \right)^2 - V(x_r) \frac{t}{r} \right)}  \ldots e^{i \left( \frac{mt}{2r} \left( \frac{x_1 - x_0}{\frac{t}{r}} \right)^2 - V(x_0) \frac{t}{r} \right)} \\
    = \; & \lim_{r \rightarrow \infty} \left( \frac{-imr}{2 \pi t} \right)^{r/2} \int dk_1 ... \int dk_{r-1} \; e^{i \sum_{k=0}^{r-1} \left( \frac{mt}{2r} \left( \frac{x_{k+1} - x_k}{\frac{t}{r}} \right)^2 - V(x_k) \frac{t}{r} \right)} \label{path_integral_3}
\end{align}
In the limit where $r \rightarrow \infty$, $\frac{x_{k+1} - x_k}{\frac{t}{r}}$ becomes $\dot{x}_k$, so the phase becomes
\begin{align}
    \lim_{r \rightarrow \infty} \frac{t}{r} \sum_{k=0}^{r-1} \left( \frac{m \dot{x}_k^2}{2} - V(x_k) \right) = \int_0^t L(\bm{x}(t'), \bm{\dot{x}}(t')) dt'
\end{align}\
where $\bm{x}$ is a particular path, which is now a function of continuous time $t'$ in the $r \rightarrow \infty$ limit. The integrals over all $x_k$'s becomes an integral over all possible paths $\bm{x}(t)$ from $x_a$ to $x_b$, and the expression in \eqref{path_integral_3} becomes
\begin{align}
    \frac{1}{Z} \int \mathcal{D} \bm{x} \; e^{i \int_0^t L (\bm{x}(t'), \bm{\dot{x}}(t') dt} 
\end{align}
where $Z$ is a normalization constant.

\section{Trotter Error Bound Proof} \label{sec:t_error_proof}
\begin{proof}[Proof of Lemma~\ref{lem:Trotter_Lagrange}]
Let us consider the expectation value of a POVM element $M_k$ that is supported on the feasible domain.  The corresponding probabilities for the Trotterized and non-Trotterized expectations are given below.
\begin{align}
    |p(k) - \widetilde{p}(k)| &= \left| \langle \psi | U^\dagger \hat{M}_k U|\psi\rangle - \langle \psi | \widetilde{U}^\dagger \hat{M}_k \widetilde{U} |\psi\rangle \right| \\
    & \leq  \left| \langle \psi | \hat{M}_k (U - \widetilde{U}) |\psi\rangle + \langle \psi| (U - \widetilde{U}) \hat{M}_k |\psi\rangle \right| \\
    & \leq 2 \left|\langle \psi| \hat{M}_k (U - \widetilde{U}) |\psi \rangle \right| \\
    & =2 \left|\langle \psi| \hat{M}_k \int_0^\tau \mathrm{d}t_1 \int_0^{t_1} \mathrm{d}t_2 e^{-i (\tau - t_1)H} e^{-i t_1 \hat{V}} e^{i t_2 \hat{V}} \left[ i \hat{V}, i T \right] e^{-i t_2 \hat{V}} e^{-i t_1 \hat{K}} |\psi \rangle \right| \\
    & = 2\left|\langle \psi| \hat{M}_k \int_0^\tau \mathrm{d}t_1 \int_0^{t_1} \mathrm{d}t_2 e^{-i (\tau - t_1)H} e^{-i t_1 \hat{V}} e^{i t_2 \hat{V}} \left( \hat{V} \hat{K} - \hat{K} \hat{V} \right) e^{-i t_2 \hat{V}} e^{-i t_1 T} |\psi \rangle \right| \\
    & \leq 2 \Bigg( \int_0^\tau \mathrm{d}t_1 \int_0^{t_1} \mathrm{d}t_2 \left| \bra{\psi} M_k e^{-i (\tau - t_1)H} e^{-i t_1 \hat{V}} e^{i t_2 \hat{V}} \hat{V} \hat{K} e^{-i t_2 \hat{V}} e^{-i t_1 \hat{K}} |\psi \rangle\right| \notag\\
    & \quad + \int_0^\tau \mathrm{d}t_1 \int_0^{t_1} \mathrm{d}t_2 \left| \bra{\psi} \hat{M}_k e^{-i (\tau - t_1)H} e^{-i t_1 \hat{V}} e^{i t_2 \hat{V}} \hat{K} \hat{V} e^{-i t_2 \hat{V}} e^{-i t_1 \hat{K}} |\psi \rangle\right| \Bigg) \\
    & \leq 2 \Bigg( \int_0^\tau \mathrm{d}t_1 \int_0^{t_1} \mathrm{d}t_2 \left\|  \hat{V} e^{i (t_1-t_2) \hat{V}} e^{i (\tau - t_1)H} \hat{M}_k\ket{\psi}\right\| \left\| \hat{K} e^{-i t_2 \hat{V}} e^{-i t_1 \hat{K}} |\psi \rangle\right\| \notag\\
    & \quad + \int_0^\tau \mathrm{d}t_1 \int_0^{t_1} \mathrm{d}t_2 \left\|  \hat{K} e^{i (t_1 - t_2) \hat{V}} e^{i (\tau - t_1)H} \hat{M}_k \ket{\psi} \right\| \left\| \hat{V} e^{-i t_2 \hat{V}} e^{-i t_1 \hat{K}}|\psi \rangle\right\| \Bigg)\\
    & \leq 2\Bigg( \int_0^\tau \mathrm{d}t_1 \int_0^{t_1} \mathrm{d}t_2 V_{\max} \left\| \hat{K} e^{-i t_2 \hat{V}} e^{-i t_1 \hat{K}} |\psi \rangle\right\| \notag\\
    & \quad + \int_0^\tau \mathrm{d}t_1 \int_0^{t_1} \mathrm{d}t_2 \left\|  \hat{K} e^{i (t_1 - t_2) \hat{V}} e^{i (\tau - t_1)H} M_k\ket{\psi} \right\|V_{\max} \Bigg) \label{tbound1}
\end{align}
Since $\hat{M}_k$ by assumption is an operator with domain and range on the feasible space $|\psi\rangle$ and $e^{i (\tau - t_1) H} M_k\ket{\psi}$ both are in the feasible space $S_n$. in Def. \ref{def:S_n}.
As $\hat{M}_k$ is a positive semi-definite operator such that $\sum_k \hat{M}_k = I$ from the definition of a POVM, we must have that $\|\hat{M}_k\|$ on the feasible space is at most $1$ as if it had an eigenvalue greater than $1$ it would ${\rm Tr}(\hat{\rho} \hat{M}_k)$ would not have a probability interpretation in general.  Thus we have that under the above assumptions
\begin{align}
    \max_{\ket{\psi}}\left\| \hat{K} e^{-i t_2 \hat{V}} e^{-i t_1 \hat{K}}|\psi \rangle\right\| &\le K_{\max}\\
    \left\|  \hat{K} e^{i (t_1 - t_2) \hat{V}} e^{i (\tau - t_1)H} M_k\ket{\psi} \right\| &\le \max_{\ket{\psi}}\left\| \hat{K} e^{-i t_2 \hat{V}} e^{-i t_1 \hat{K}}|\psi \rangle\right\| \le K_{\max}
\end{align}
Then, \eqref{tbound1} becomes
\begin{align}
    & \quad 2 \Bigg( \int_0^\tau \mathrm{d}t_1 \int_0^{t_1} \mathrm{d}t_2 V_{\max} \left\| \hat{K}  e^{-i t_1 \hat{K}} \sum_{k=0}^{2^{\eta D n}-1} a_{k} \ket{p_k} \right\|\nonumber\\
    &\qquad+ \int_0^\tau \mathrm{d}t_1 \int_0^{t_1} \mathrm{d}t_2 \left\|  \hat{K}  \sum_{k=0}^{2^{\eta D n}-1} a_{k} \ket{p_k} \ \right\|V_{\max} \Bigg) \\
    &\le {2 \tau^2 V_{\max} K_{\max}}
\end{align}

This is an upper bound on the error from one time step; the total error on the whole time evolution is (by Box 4.1 in \cite{nielsenchuang}) simply $r$ times the error from one time step. Thus, the total error $\epsilon$ is upper bounded by 
\begin{equation}
    r{2 \tau^2 V_{\max} K_{\max}} 
    = \frac{2 T^2 V_{\max} K_{\max}}{r} 
    \label{epsilon_scaling}
\end{equation}
\end{proof}

\section{Proof of Claim \ref{clm:ast}} \label{sec:ast_full_proof}
\begin{proof}[Proof of Claim \ref{clm:ast}]
    \begin{align}
        & \big(\bra{0} \otimes \normalfont{\openone} \big) \widetilde{A}_m \big(\ket{0} \otimes \normalfont{\openone} \big) - \hat{A}_m \notag \\
        &= \big(\bra{0} \otimes \normalfont{\openone} \big) \left(\sum_{b=0}^{2^B-1} \sum_{c_1=0}^{d-1} \sum_{c_2=0}^{d-1} \frac{1}{2^B} (\hat{X} \otimes \openone) \hat{U}_{m, b, c_1,c_2} \right) \big(\ket{0} \otimes \normalfont{\openone} \big) - \hat{A}_m \\
        &= \sum_{b=0}^{2^B-1} \sum_{c_1=0}^{d-1} \sum_{c_2=0}^{d-1} \frac{1}{2^B} \sum_{\substack{j: (0,j) \in \\ V(G_m^{(c_1, c_2)})}} g(m, j, c_1, b) e^{i \arg \left( \braket{ \chi_{f_{ind}(m, 0, j, c_1)}^{(\ell_{m+1)}} }{ \chi_j^{(\ell_m)}}  \right) - i \lambda^{(\ell_m)}_j \tau_m \frac{t}{r}}  \ketbra{f_{ind}(m, 0, j, c_1)}{j} - \hat{A}_m
    \end{align}
    Here we drop all the terms in $\hat{U}_{m,b,c_1,c_2}$ where the projections of the first qubit are zero (i.e., where the operator on the first qubit is $\ketbra{0}{0}$, $\ketbra{1}{1}$ or $\ketbra{0}{1}$). Then, applying the definition of $g(m,j,c_1,b)$ (see Definition \ref{def:U}) and splitting the $b$ sum accordingly, we get
    \begin{align}
        & \sum_{c_1=0}^{d-1} \sum_{c_2=0}^{d-1} \frac{1}{2^B} \sum_{\substack{j: (0,j) \in \\ V(G_m^{(c_1, c_2)})}} \notag \\
        & \quad \Bigg[ \sum_{b=0}^{\left[ 2^B \left|\braket{\chi^{(\ell_{m+1})}_{f_{ind}(\ell_m, j, \ell_{m+1}, c_1)}}{\chi^{(\ell_m)}_j} \right| \right]_B-1} e^{i \arg \left( \braket{ \chi_{f_{ind}(m, 0, j, c_1)}^{(\ell_{m+1)}} }{ \chi_j^{(\ell_m)}}  \right) - i \lambda^{(\ell_m)}_j \tau_m \frac{t}{r}} \ketbra{f_{ind}(m, 0, j, c_1)}{j} \notag \\
        & \quad + \sum_{\substack{b= \\ \left[ 2^B \left| \braket{\chi^{(\ell_{m+1})}_{f_{ind}(\ell_m, j, \ell_{m+1}, c_1)}}{\chi^{(\ell_m)}_j} \right| \right]_B}}^{2^B-1} (-1)^b e^{i \arg \left( \braket{ \chi_{f_{ind}(m, 0, j, c_1)}^{(\ell_{m+1)}} }{ \chi_j^{(\ell_m)}}  \right) - i \lambda^{(\ell_m)}_j \tau_m \frac{t}{r}} \ketbra{f_{ind}(m, 0, j, c_1)}{j} \Bigg] - \hat{A}_m \notag \\
        &= \sum_{c_1=0}^{d-1} \sum_{c_2=0}^{d-1}  \sum_{\substack{j: (0,j) \in \\ V(G_m^{(c_1, c_2)})}} \notag \\
        & \quad \Bigg[ \frac{1}{2^B} \left[ 2^B \left|\braket{\chi^{(\ell_{m+1})}_{f_{ind}(\ell_m, j, \ell_{m+1}, c_1)}}{\chi^{(\ell_m)}_j} \right| \right]_B e^{i \arg \left( \braket{ \chi_{f_{ind}(m, 0, j, c_1)}^{(\ell_{m+1)}} }{ \chi_j^{(\ell_m)}}  \right) - i \lambda^{(\ell_m)}_j \tau_m \frac{t}{r}} \ketbra{f_{ind}(m, 0, j, c_1)}{j} \notag \\
        & \quad + \frac{1}{2^B} \sum_{\substack{b= \\ \left[ 2^B \left| \braket{\chi^{(\ell_{m+1})}_{f_{ind}(\ell_m, j, \ell_{m+1}, c_1)}}{\chi^{(\ell_m)}_j} \right| \right]_B}}^{2^B-1}  (-1)^b e^{i \arg \left( \braket{ \chi_{f_{ind}(m, 0, j, c_1)}^{(\ell_{m+1)}} }{ \chi_j^{(\ell_m)}}  \right) - i \lambda^{(\ell_m)}_j \tau_m \frac{t}{r}} \ketbra{f_{ind}(m, 0, j, c_1)}{j} \Bigg] - \hat{A}_m \label{ast_proof1}
    \end{align}
    The second term is in the square brackets is 0 if the number of terms being summed over is even, and $\frac{1}{2^B} e^{i \arg \left( \braket{ \chi_{f_{ind}(m, 0, j, c_1)}^{(\ell_{m+1)}} }{ \chi_j^{(\ell_m)}}  \right) - i \lambda^{(\ell_m)}_j \tau_m \frac{t}{r}} \ketbra{f_{ind}(m, 0, j, c_1)}{j}$ if it is odd. Thus, the absolute value of the coefficient of the $\ketbra{f_{ind}(j, 0, j, c_1}{j}$ operator is in the set
    \begin{equation}
        \left\{ \frac{1}{2^B} \left[ 2^B \left| \braket{\chi^{(\ell_{m+1})}_{f_{ind}(\ell_m, j, \ell_{m+1}, c_1)}}{\chi^{(\ell_m)}_j} \right| \right]_B, \frac{1}{2^B} \left( \left[ 2^B \left| \braket{\chi^{(\ell_{m+1})}_{f_{ind}(\ell_m, j, \ell_{m+1}, c_1)}}{\chi^{(\ell_m)}_j} \right| \right]_B + 1 \right) \right\},
    \end{equation} which differs from the ideal magnitude of $\left| \braket{\chi^{(\ell_{m+1})}_{f_{ind}(\ell_m, j, \ell_{m+1}, c_1)}}{\chi^{(\ell_m)}_j} \right|$ by at most $\frac{2}{2^B}$. This is because the encoding of $2^B \left| \braket{\chi^{(\ell_{m+1})}_{f_{ind}(\ell_m, j, \ell_{m+1}, c_1)}}{\chi^{(\ell_m)}_j} \right|$ into $B$ bits rounds to an integer, so
    \begin{align}
        \left| \left| \braket{\chi^{(\ell_{m+1})}_{f_{ind}(\ell_m, j, \ell_{m+1}, c_1)}}{\chi^{(\ell_m)}_j} \right| - \frac{1}{2^B} \left[ 2^B \left| \braket{\chi^{(\ell_{m+1})}_{f_{ind}(\ell_m, j, \ell_{m+1}, c_1)}}{\chi^{(\ell_m)}_j} \right| \right]_B\right| \leq \frac{1}{2^B},
    \end{align} so the total error in the magnitude of the coefficient is at most $\frac{1}{2^B} + \frac{1}{2^B} = \frac{2}{2^B}$. Thus, for some $\Delta(m,j,c_1,c_2) \in \mathbb{C}$ such that $|\Delta| \leq \frac{2}{2^B}$, we can write \eqref{ast_proof1} as
    \begin{align}
        & \sum_{c_1=0}^{d-1} \sum_{c_2=0}^{d-1}  \sum_{\substack{j: (0,j) \in \\ V(G_m^{(c_1, c_2)})}} \left|\braket{\chi^{(\ell_{m+1})}_{f_{ind}(m, 0, j, c_1)}}{\chi^{(\ell_m)}_j} \right| e^{i \arg \left( \braket{ \chi_{f_{ind}(m, 0, j, c_1)}^{(\ell_{m+1)}} }{ \chi_j^{(\ell_m)}}  \right) - i \lambda^{(\ell_m)}_j \tau_m \frac{t}{r}} \ketbra{f_{ind}(m, 0, j, c_1)}{j} \notag \\
        & \quad + \Delta (m,j,c_1,c_2) \ketbra{f_{ind}(m, 0, j, c_1)}{j} - \hat{A}_m \notag \\
        &= \hat{A}_m + \sum_{c_1=0}^{d-1} \sum_{c_2=0}^{d-1}  \sum_{\substack{j: (0,j) \in \\ V(G_m^{(c_1, c_2)})}}  \Delta (m,j,c_1,c_2) \ketbra{f_{ind}(m, 0, j, c_1)}{j}  - \hat{A}_m \label{ast_proof3} \\ 
        &=  \sum_{c_1=0}^{d-1} \sum_{c_2=0}^{d-1}  \sum_{\substack{j: (0,j) \in \\ V(G_m^{(c_1, c_2)})}}  \Delta (m,j,c_1,c_2) \ketbra{f_{ind}(m, 0, j, c_1)}{j}
    \end{align}
    The spectral norm of the above operator is
    \begin{align}
        & \left\| \sum_{c_1=0}^{d-1} \sum_{c_2=0}^{d-1}  \sum_{\substack{j: (0,j) \in \\ V(G_m^{(c_1, c_2)})}}  \Delta (m,j,c_1,c_2) \ketbra{f_{ind}(m, 0, j, c_1)}{j} \right\|_2 \notag \\
        & \leq \sum_{c_1=0}^{d-1} \sum_{c_2=0}^{d-1} \left\|   \sum_{\substack{j: (0,j) \in \\ V(G_m^{(c_1, c_2)})}}  \Delta (m,j,c_1,c_2) \ketbra{f_{ind}(m, 0, j, c_1)}{j}  \right\|_2 \\
        &= \sum_{c_1=0}^{d-1} \sum_{c_2=0}^{d-1} \sup_{\ket{x}}\left\| \left( \sum_{\substack{j: (0,j) \in \\ V(G_m^{(c_1, c_2)})}}  \Delta (m,j,c_1,c_2) \ketbra{f_{ind}(m, 0, j, c_1)}{j} \right) \ket{x} \right\|_2 \label{ast_proof2}
    \end{align}
    If we write $\ket{x}$ as a linear combination of computational basis states, i.e.,
    \begin{align}
        \ket{x} := \sum_{j=0}^{2^n-1} a_j \ket{j},
    \end{align}
    where $\sum_{j=0}^{2^n-1} |a_j|^2 = 1$, then \eqref{ast_proof2} becomes
    \begin{align}
        & \quad \sum_{c_1=0}^{d-1} \sum_{c_2=0}^{d-1} \sup_{\substack{\{a_j\}_j: \\ \sum_{j=0}^{2^n-1} |a_j|^2 =1}} \left\| \sum_{\substack{j: (0,j) \in \\ V(G_m^{(c_1, c_2)})}}  \Delta (m,j,c_1,c_2) a_j \ket{f_{ind}(m, 0, j, c_1)} \braket{j}{j} \right\|_2 \\
        &= \sum_{c_1=0}^{d-1} \sum_{c_2=0}^{d-1} \sup_{\substack{\{a_j\}_j: \\ \sum_{j=0}^{2^n-1} |a_j|^2 =1}} \left\| \sum_{\substack{j: (0,j) \in \\ V(G_m^{(c_1, c_2)})}}  \Delta (m,j,c_1,c_2) a_j \ket{f_{ind}(m, 0, j, c_1)} \right\|_2 \label{2norm}
    \end{align}
    For a fixed $(c_1,c_2)$, $f_{ind}(m,0,j,c_1)$ is a one-to-one function of $j$, therefore the Euclidean norm in \ref{2norm} simplifies to 
    \begin{align}
        &\sum_{c_1=0}^{d-1} \sum_{c_2=0}^{d-1} \sup_{\substack{\{a_j\}_j: \\ \sum_{j=0}^{2^n-1} |a_j|^2 =1}} \sqrt{\sum_{\substack{j: (0,j) \in \\ V(G_m^{(c_1, c_2)})}} \left| \Delta (m,j,c_1,c_2) a_j \right|^2 } \notag \\
        &\leq \sum_{c_1=0}^{d-1} \sum_{c_2=0}^{d-1}  \frac{2}{2^B} \sup_{\substack{\{a_j\}_j: \\ \sum_{j=0}^{2^n-1} |a_j|^2 =1}} \sqrt{\sum_{\substack{j: (0,j) \in \\ V(G_m^{(c_1, c_2)})}} \left| a_j \right|^2 } \\
        &\leq \sum_{c_1=0}^{d-1} \sum_{c_2=0}^{d-1}  \frac{2}{2^B} \\
        &= \frac{2d^2}{2^B}
    \end{align}
\end{proof}
\section{Implementation of Alternating Sign Trick with Linear Combination of Unitaries} \label{sec:lcu}
Let $\ket{\psi_m} = \sum_{j=0}^{2^n-1} a_j \ket{j}$ denote the state $\ket{\psi}$ at the $m^{\rm th}$ time step. The full input state, with all the ancillas, is $\ket{m}_0 \ket{0}_1 \ket{\psi}_2 \ket{0}_3 \ket{0}_4 ... \ket{0}_{9}$ where the subscripts denote the register numbers. The following table describes what each register is used for: 
\begin{center}
    \begin{tabular}{ |c|c|c|c| } 
        \hline
        Register \# & Initial State & Size (Qubits) & Description \\ 
        \hline
        0 & $\ket{m}$ & $\log M$ & stores time step number $m$ \\
        1 & $\ket{0}$ & 1 & stores whether state corresponds to $\hat{H}_{\ell_m}$ or $\hat{H}_{\ell_{m+1}}$ basis \\ 
        2 & $\ket{\psi_m}$ & $n$ & stores basis state indices for $\ket{\psi}$ \\
        3 & $\ket{0}$ & $B$ & stores $b$ index for alternating sign trick\\ 
        4 & $\ket{0}$ & $\log d$ & stores $c_1$ \\ 
        5 & $\ket{0}$ & $\log d$ & stores $c_2$ \\ 
        6 & $\ket{0}$ & 1 & stores whether vertex is in $G_m^{(c_1,c_2)}$ \\
        7 & $\ket{0}$ & $n$ & stores $f_{ind}(m, 0, j, c_1)$ \\ 
        8 & $\ket{0}$ & $B$ & stores value of $\left[ 2^B \big|\langle \chi^{(\ell_{m+1})}_{f_{ind}(m, 0, j, c_1)} | \chi^{(\ell_m)}_j \rangle \big| \right]_B $ \\
        9 & $\ket{0}$ & 1 & stores whether $b \geq  \big| 2^B \langle \chi_{f_{ind}(m, 0, j, c_1)} | j \rangle \big|$\\
        \hline
    \end{tabular}
\end{center}
In order to check whether $b \geq \big| 2^B \langle \chi_{f_{ind}(m, 0, j, c_1)} | j \rangle \big|$ for the alternating sign trick, we will have to define a circuit that compares two bitstrings and checks whether the second is greater than the first.
\begin{definition}[Comparison circuit]
    Let $\hat{U}_{\leq} \in \mathcal{L}(\mathcal{H}_{B} \otimes \mathcal{H}_{B} \otimes \mathcal{H}_2)$ be a unitary operator,
    \begin{align}
        \hat{U}_{\leq} : \ket{b_1} \ket{b_2} \ket{c} \mapsto \ket{b_1} \ket{b_2} \ket{c \oplus (b_2 \leq b_1)}
    \end{align}
\end{definition}
The comparison circuit takes two bitstrings $b_1, b_2$ of length $B$ and flips the sign of a third register if the $b_1 \geq b_2$. By Lemma 24 in \cite{ssdmsw2023}, the above circuit can be constructed using $5B-2$ Toffoli gates and $B$ additional qubits (see \cite{ssdmsw2023} for full circuit). 
\begin{algorithm}[H]
\begin{algorithmic}[1]
\State Apply Hadamard gates to $r_3, r_4$ and $r_5$ \Comment{prepares uniform superposition of indices to be summed over}
\end{algorithmic}
\caption{Prepare}
\end{algorithm}

\begin{algorithm}[H]
\caption{Select} \label{alg:sel}
\begin{algorithmic}[1]
\item Apply $O_C(r_0, r_1, r_2, r_4, r_5, r_6)$ \Comment{checks whether $\ket{j} \in G_m^{(c_1,c_2)}$}
\item Apply $O_{ind} (r_0, r_1, r_2, r_4, r_7)$, controlled on $r_6$ \Comment{computes index of state that $\ket{j}$ gets mapped to if it's in $G_m^{(c_1, c_2)}$}
\item Apply $O_{IM}(r_0, r_2, r_7, r_8)$, controlled on $r_6$ \Comment{computes inner product magnitude if edge is in $G_m^{(c_1, c_2)}$}
\item Apply $U_{\leq} (r_3, r_8, r_9)$ \Comment{marks states where $b \geq 2^B \big| \langle \chi_{f_{ind}(m, 0, j, c_1)} | j \rangle \big|$}
\item Apply $O_{IP}(r_0, r_2, r_7)$ and $O_P(r_0, r_2)$, controlled on $r_6$ and 0-controlled on $r_1$ \Comment{computes phases} 
\item Swap $r_2$ and $r_7$ and apply $X(r_1)$, controlled on $r_6$ \Comment{implements transition if $\big((0,j), (1,f_{ind}(m, 0, j, c_1)\big) \in E(G_m^{(c_1, c_2)})$}
\item Apply $(-1)^b$ to $r_0$, controlled on $r_9$ \Comment{cancels out states with 0 inner product, applies amplitudes}
\item Apply $\hat{U}_{\leq} (r_3, r_8, r_9)$ \Comment{Sets $r_9$ back to 0}
\item Apply $O_{IM}(r_0, r_2, r_7, r_8)$ controlled on $r_6$ \Comment{sets $r_8$ back to 0}
\item Apply $O_{ind}(r_0, r_1, r_2, r_5, r_7)$ controlled on $r_6$ 
\Comment{sets $r_7$ back to 0}
\item Apply $O_C(r_0, r_1, r_2, r_4, r_5, r_6)$ \Comment{Sets $r_6$ back to 0}
\item Apply $X(r_1)$ \Comment{Sets $r_1$ back to 0}
\end{algorithmic}
\end{algorithm}
The Prepare operation simply prepares a superposition of all the $2^B d^2$ computational basis states by applying Hadamard gates to each of the $B + 2 \log d$ qubits. 

Recall that the Select operation is 
\begin{align}
    \text{\normalfont{SEL}}(\vec{U}_m) := \sum_{b=0}^{B-1} \sum_{c_1=0}^{d-1} \sum_{c_2=0}^{d-1} \ketbra{b}{b} \otimes \ketbra{c_1}{c_1} \otimes \ketbra{c_2}{c_2} \otimes  \left((\hat{X} \otimes \normalfont{\openone}) \hat{U}_{m, b, c_1, c_2} \right),
\end{align}
where 
\begin{align}
    \hat{U}_{m,b,c_1,c_2} &:= \sum_{\substack{j: (0,j) \in \\ V(G_m^{(c_1, c_2)})}} \bigg( g(m, j, c_1, b) e^{i \arg \left( \braket{ \chi_{f_{ind}(m, 0, j, c_1)}^{(\ell_{m+1)}} }{ \chi_j^{(\ell_m)}}  \right) - i \lambda^{(\ell_m)}_j \tau_m \frac{t}{r}} \ketbra{1}{0} \otimes \ketbra{f_{ind}(m, 0, j, c_1)}{j} \notag \\
    & \quad + (-1)^b \ketbra{0}{1} \otimes \ketbra{j}{f_{ind}(m, 0, j, c_1)} \bigg) + \sum_{\substack{j: (0,j) \notin \\ V(G_m^{(c_1, c_2)})}} (-1)^b \ketbra{0}{0} \otimes \ketbra{j}{j} + \sum_{\substack{q:(1, q) \notin \\ V(G_m^{(c_1,c_2)})}} (-1)^b \ketbra{1}{1} \otimes \ketbra{q}{q},
\end{align}
\begin{align}
    g(m, j, c_1, b) &= \begin{cases}
        1 & b < \left[ 2^B \left| \braket{\chi^{(\ell_{m+1})}_{f_{ind}(\ell_m, j, \ell_{m+1}, c_1)}}{\chi^{(\ell_m)}_j} \right| \right]_B \\
        (-1)^b & \text{otherwise}
    \end{cases}.
\end{align}
To see how the Select algorithm in Algorithm \ref{alg:sel} implements the Select operation defined in Def. \ref{def:sel} on registers $r_3, r_4, r_5, r_1, r_2$, consider the action of each step of the algorithm on an input state where the second register $r_2$ is a computational basis state; we only need to consider the computational basis states because any state $\ket{\psi_m}$ will be a linear combination of basis states. We will only show Algorithm \ref{alg:sel} implements Select with input states where $r_1$ is in the 0 state, because the input state will always be of that form, however, a similar calculation shows that it implements the desired transitions for states where $r_1$ is in the 1 state (the main difference is that $O_{ind}$ is applied to $r_4$ instead of $r_5$ in Step 10 to set $r_7$ back to 0). The subscript on each register indicates the register number.
\begin{align}
    \text{Input: } & \ket{m}_0 \ket{0}_1 \ket{j}_2 \ket{b}_3 \ket{c_1}_4 \ket{c_2}_5 \ket{0}_6 \ket{0}_7 \ket{0}_8 \ket{0}_9 \\
    \xrightarrow{\text{Step 1}} &  \begin{cases} \ket{m}_0 \ket{0}_1 \ket{j}_2 \ket{b}_3 \ket{c_1}_4 \ket{c_2}_5 \ket{0}_6 \ket{0}_7 \ket{0}_8 \ket{0}_9 & j \notin G_m^{(c_1, c_2)}\\
    \ket{m}_0 \ket{0}_1 \ket{j}_2 \ket{b}_3 \ket{c_1}_4 \ket{c_2}_5 \ket{1}_6 \ket{0}_7 \ket{0}_8 \ket{0}_9 & j \in G_m^{(c_1, c_2)}
    \end{cases} \\
    \xrightarrow{\text{Step 2}} &  \begin{cases} \ket{m}_0 \ket{0}_1 \ket{j}_2 \ket{b}_3 \ket{c_1}_4 \ket{c_2}_5 \ket{0}_6 \ket{0}_7 \ket{0}_8 \ket{0}_9 & j \notin G_m^{(c_1, c_2)}\\
    \ket{m}_0 \ket{0}_1 \ket{j}_2 \ket{b}_3 \ket{c_1}_4 \ket{c_2}_5 \ket{1}_6 \ket{f_{ind}(m, 0, j, c_1)}_7 \ket{0}_8 \ket{0}_9 & j \in G_m^{(c_1, c_2)}
    \end{cases} \\
    \xrightarrow{\text{Step 3}} &  \begin{cases} \ket{m}_0 \ket{0}_1 \ket{j}_2 \ket{b}_3 \ket{c_1}_4 \ket{c_2}_5 \ket{0}_6 \ket{0}_7 \ket{0}_8 \ket{0}_9 & j \notin G_m^{(c_1, c_2)}\\
    \ket{m}_0 \ket{0}_1 \ket{j}_2 \ket{b}_3 \ket{c_1}_4 \ket{c_2}_5 \ket{1}_6 \ket{f_{ind}(m, b, j, c_1)}_7 \ket{\left[ 2^B \left| \braket{\chi_{f_{ind}(m, 0, j, c_1)}^{(\ell_{m+1})}}{\chi_j^{(\ell_m)}} \right| \right]_B}_8 \ket{0}_9 & j \in G_m^{(c_1, c_2)} 
    \end{cases} \\
    \xrightarrow{\text{Step 4}} &  \begin{cases} \ket{m}_0 \ket{0}_1 \ket{j}_2 \ket{b}_3 \ket{c_1}_4 \ket{c_2}_5 \ket{0}_6 \ket{0}_7 \ket{0}_8 \ket{1}_9 & j \notin G_m^{(c_1, c_2)} \vspace{1em} \\
    \ket{m}_0 \ket{0}_1 \ket{j}_2 \ket{b}_3 \ket{c_1}_4 \ket{c_2}_5 \ket{1}_6 \ket{f_{ind}(m, 0, j, c_1)}_7 \ket{\left[ 2^B \left| \braket{\chi_{f_{ind}(m, 0, j, c_1)}^{(\ell_{m+1})}}{\chi_j^{(\ell_m)}} \right| \right]_B}_8 \\
    \quad \ket{ \left(b > \left[ 2^B \left| \braket{\chi_{f_{ind}(m, 0, j, c_1)}^{(\ell_{m+1})}}{\chi_j^{(\ell_m)}} \right| \right]_B \right)}_9 & j \in G_m^{(c_1, c_2)} 
    \end{cases} \\
    \xrightarrow{\text{Step 5}} &  \begin{cases} \ket{m}_0 \ket{0}_1 \ket{j}_2 \ket{b}_3 \ket{c_1}_4 \ket{c_2}_5 \ket{0}_6 \ket{0}_7 \ket{0}_8 \ket{1}_9 & j \notin G_m^{(c_1, c_2)} \vspace{1em} \\
    e^{i \arg \left(\braket{\chi_{f_{ind}(m, 0, j, c_1)}^{(\ell_{m+1})}}{\chi_j^{(\ell_m)}} \right) - i \lambda_j^{(\ell_m)} \tau_m \frac{t}{r}} \ket{m}_0 \ket{0}_1 \ket{j}_2 \ket{b}_3 \ket{c_1}_4 \ket{c_2}_5 \ket{1}_6 \ket{f_{ind}(m, 0, j, c_1)}_7 \\
    \quad \ket{\left[ 2^B \left| \braket{\chi_{f_{ind}(m, 0, j, c_1)}^{(\ell_{m+1})}}{\chi_j^{(\ell_m)}} \right| \right]_B}_8 \ket{ \left(b > \left[ 2^B \left| \braket{\chi_{f_{ind}(m, 0, j, c_1)}^{(\ell_{m+1})}}{\chi_j^{(\ell_m)}} \right| \right]_B \right)}_9 & j \in G_m^{(c_1, c_2)} 
    \end{cases} \\
    \xrightarrow{\text{Step 6}} &  \begin{cases} \ket{m}_0 \ket{0}_1 \ket{j}_2 \ket{b}_3 \ket{c_1}_4 \ket{c_2}_5 \ket{0}_6 \ket{0}_7 \ket{0}_8 \ket{1}_9 & j \notin G_m^{(c_1, c_2)} \vspace{1em} \\
    e^{i \arg \left(\braket{\chi_{f_{ind}(m, 0, j, c_1)}^{(\ell_{m+1})}}{\chi_j^{(\ell_m)}} \right) - i \lambda_j^{(\ell_m)} \tau_m \frac{t}{r}} \ket{m}_0 \ket{1}_1 \ket{f_{ind}(m, 0, j, c_1)}_2 \ket{b}_3 \ket{c_1}_4 \ket{c_2}_5 \ket{1}_6 \ket{j}_7 \\
    \quad \ket{\left[ 2^B \left| \braket{\chi_{f_{ind}(m, 0, j, c_1)}^{(\ell_{m+1})}}{\chi_j^{(\ell_m)}} \right| \right]_B}_8 \ket{ \left(b > \left[ 2^B \left| \braket{\chi_{f_{ind}(m, 0, j, c_1)}^{(\ell_{m+1})}}{\chi_j^{(\ell_m)}} \right| \right]_B \right)}_9 & j \in G_m^{(c_1, c_2)} 
    \end{cases} \\
    \xrightarrow{\text{Step 7}} &  \begin{cases} (-1)^b \ket{m}_0 \ket{1}_1 \ket{j}_2 \ket{b}_3 \ket{c_1}_4 \ket{c_2}_5 \ket{0}_6 \ket{0}_7 \ket{0}_8 \ket{1}_9 & j \notin G_m^{(c_1, c_2)} \vspace{1em}\\
    e^{i \arg \left(\braket{\chi_{f_{ind}(m, 0, j, c_1)}^{(\ell_{m+1})}}{\chi_j^{(\ell_m)}} \right) - i \lambda_j^{(\ell_m)} \tau_m \frac{t}{r}} \ket{m}_0 \\
    \quad \ket{1}_1 \ket{f_{ind}(m, 0, j, c_1)}_2 \ket{b}_3 \ket{c_1}_4 \ket{c_2}_5 \ket{1}_6 \ket{j}_7 \\
    \quad \ket{\left[ 2^B \left| \braket{\chi_{f_{ind}(m, 0, j, c_1)}^{(\ell_{m+1})}}{\chi_j^{(\ell_m)}} \right| \right]_B}_8 \\
    \quad \ket{ \left(b > \left[ 2^B \left| \braket{\chi_{f_{ind}(m, 0, j, c_1)}^{(\ell_{m+1})}}{\chi_j^{(\ell_m)}} \right| \right]_B 
    \right)}_9 & j \in G_m^{(c_1, c_2)}, b < \left[ 2^B \left| \braket{\chi_{f_{ind}(m, 0, j, c_1)}^{(\ell_{m+1})}}{\chi_j^{(\ell_m)}} \right| \right]_B \vspace{1em} \\
    (-1)^b e^{i \arg \left(\braket{\chi_{f_{ind}(m, 0, j, c_1)}^{(\ell_{m+1})}}{\chi_j^{(\ell_m)}} \right) - i \lambda_j^{(\ell_m)} \tau_m \frac{t}{r}} \ket{m}_0 \ket{1}_1 \\
    \quad \ket{f_{ind}(m, 0, j, c_1)}_2 \ket{b}_3 \ket{c_1}_4 \ket{c_2}_5 \ket{1}_6 \ket{j}_7 \\
    \quad \ket{\left[ 2^B \left| \braket{\chi_{f_{ind}(m, 0, j, c_1)}^{(\ell_{m+1})}}{\chi_j^{(\ell_m)}} \right| \right]_B}_8 \\
    \quad \ket{ \left(b > \left[ 2^B \left| \braket{\chi_{f_{ind}(m, 0, j, c_1)}^{(\ell_{m+1})}}{\chi_j^{(\ell_m)}} \right| \right]_B \right)}_9 & j \in G_m^{(c_1, c_2)}, b \geq \left[ 2^B \left| \braket{\chi_{f_{ind}(m, 0, j, c_1)}^{(\ell_{m+1})}}{\chi_j^{(\ell_m)}} \right| \right]_B
    \end{cases} \\
    =&  \begin{cases} (-1)^b \ket{m}_0 \ket{1}_1 \ket{j}_2 \ket{b}_3 \ket{c_1}_4 \ket{c_2}_5 \ket{0}_6 \ket{0}_7 \ket{0}_8 \ket{1}_9 & j \notin G_m^{(c_1, c_2)} \vspace{1em} \\
    g(m,j,c_1,b) e^{i \arg \left(\braket{\chi_{f_{ind}(m, 0, j, c_1)}^{(\ell_{m+1})}}{\chi_j^{(\ell_m)}} \right) - i \lambda_j^{(\ell_m)} \tau_m \frac{t}{r}} \ket{m}_0 \\
    \quad \ket{1}_1 \ket{f_{ind}(m, 0, j, c_1)}_2 \ket{b}_3 \ket{c_1}_4 \ket{c_2}_5 \ket{1}_6 \ket{j}_7 \ket{\left[ 2^B \left| \braket{\chi_{f_{ind}(m, 0, j, c_1)}^{(\ell_{m+1})}}{\chi_j^{(\ell_m)}} \right| \right]_B}_8 \\
    \quad \ket{ \left(b > \left[ 2^B \left| \braket{\chi_{f_{ind}(m, 0, j, c_1)}^{(\ell_{m+1})}}{\chi_j^{(\ell_m)}} \right| \right]_B 
    \right)}_9 & j \in G_m^{(c_1, c_2)}
    \end{cases} \\
    \xrightarrow{\text{Step 8}} &  \begin{cases} (-1)^b \ket{m}_0 \ket{1}_1 \ket{j}_2 \ket{b}_3 \ket{c_1}_4 \ket{c_2}_5 \ket{0}_6 \ket{0}_7 \ket{0}_8 \ket{0}_9 & j \notin G_m^{(c_1, c_2)} \vspace{1em} \\
    g(m,j,c_1,b) e^{i \arg \left(\braket{\chi_{f_{ind}(m, 0, j, c_1)}^{(\ell_{m+1})}}{\chi_j^{(\ell_m)}} \right) - i \lambda_j^{(\ell_m)} \tau_m \frac{t}{r}} \ket{m}_0 \ket{1}_1 \ket{f_{ind}(m, 0, j, c_1)}_2 \\
    \quad \ket{b}_3 \ket{c_1}_4 \ket{c_2}_5 \ket{1}_6 \ket{j}_7 \ket{\left[ 2^B \left| \braket{\chi_{f_{ind}(m, 0, j, c_1)}^{(\ell_{m+1})}}{\chi_j^{(\ell_m)}} \right| \right]_B}_8 \ket{0}_9 & j \in G_m^{(c_1, c_2)}
    \end{cases} \\
    \xrightarrow{\text{Step 9}} &  \begin{cases} (-1)^b \ket{m}_0 \ket{1}_1 \ket{j}_2 \ket{b}_3 \ket{c_1}_4 \ket{c_2}_5 \ket{0}_6 \ket{0}_7 \ket{0}_8 \ket{0}_9 & j \notin G_m^{(c_1, c_2)} \vspace{1em} \\
    g(m,j,c_1,b) e^{i \arg \left(\braket{\chi_{f_{ind}(m, 0, j, c_1)}^{(\ell_{m+1})}}{\chi_j^{(\ell_m)}} \right) - i \lambda_j^{(\ell_m)} \tau_m \frac{t}{r}} \ket{m}_0 \ket{1}_1 \ket{f_{ind}(m, 0, j, c_1)}_2 \\
    \quad \ket{b}_3 \ket{c_1}_4 \ket{c_2}_5 \ket{1}_6 \ket{j}_7 \ket{0}_8 \ket{0}_9 & j \in G_m^{(c_1, c_2)}
    \end{cases} \\
    \xrightarrow{\text{Step 10}} &  \begin{cases} (-1)^b \ket{m}_0 \ket{1}_1 \ket{j}_2 \ket{b}_3 \ket{c_1}_4 \ket{c_2}_5 \ket{0}_6 \ket{0}_7 \ket{0}_8 \ket{0}_9 & j \notin G_m^{(c_1, c_2)} \vspace{1em} \\
    g(m,j,c_1,b) e^{i \arg \left(\braket{\chi_{f_{ind}(m, 0, j, c_1)}^{(\ell_{m+1})}}{\chi_j^{(\ell_m)}} \right) - i \lambda_j^{(\ell_m)} \tau_m \frac{t}{r}} \ket{m}_0 \ket{1}_1 \ket{f_{ind}(m, 0, j, c_1)}_2 \\
    \quad \ket{b}_3 \ket{c_1}_4 \ket{c_2}_5 \ket{1}_6 \ket{0}_7 \ket{0}_8 \ket{0}_9 & j \in G_m^{(c_1, c_2)}
    \end{cases} \\
    \xrightarrow{\text{Step 11}} &  \begin{cases} (-1)^b \ket{m}_0 \ket{1}_1 \ket{j}_2 \ket{b}_3 \ket{c_1}_4 \ket{c_2}_5 \ket{0}_6 \ket{0}_7 \ket{0}_8 \ket{0}_9 & j \notin G_m^{(c_1, c_2)} \vspace{1em} \\
    g(m,j,c_1,b) e^{i \arg \left(\braket{\chi_{f_{ind}(m, 0, j, c_1)}^{(\ell_{m+1})}}{\chi_j^{(\ell_m)}} \right) - i \lambda_j^{(\ell_m)} \tau_m \frac{t}{r}} \ket{m}_0 \ket{1}_1 \ket{f_{ind}(m, 0, j, c_1)}_2 \\
    \quad \ket{b}_3 \ket{c_1}_4 \ket{c_2}_5 \ket{0}_6 \ket{0}_7 \ket{0}_8 \ket{0}_9 & j \in G_m^{(c_1, c_2)}
    \end{cases} \\
    \xrightarrow{\text{Step 11}} &  \begin{cases} (-1)^b \ket{m}_0 \ket{1}_1 \ket{j}_2 \ket{b}_3 \ket{c_1}_4 \ket{c_2}_5 \ket{0}_6 \ket{0}_7 \ket{0}_8 \ket{0}_9 & j \notin G_m^{(c_1, c_2)} \vspace{1em} \\
    g(m,j,c_1,b) e^{i \arg \left(\braket{\chi_{f_{ind}(m, 0, j, c_1)}^{(\ell_{m+1})}}{\chi_j^{(\ell_m)}} \right) - i \lambda_j^{(\ell_m)} \tau_m \frac{t}{r}} \ket{m}_0 \ket{1}_1 \ket{f_{ind}(m, 0, j, c_1)}_2 \\
    \quad \ket{b}_3 \ket{c_1}_4 \ket{c_2}_5 \ket{0}_6 \ket{0}_7 \ket{0}_8 \ket{0}_9 & j \in G_m^{(c_1, c_2)}
    \end{cases} \\
    \xrightarrow{\text{Step 12}} &  \begin{cases} (-1)^b \ket{m}_0 \ket{0}_1 \ket{j}_2 \ket{b}_3 \ket{c_1}_4 \ket{c_2}_5 \ket{0}_6 \ket{0}_7 \ket{0}_8 \ket{0}_9 & j \notin G_m^{(c_1, c_2)} \vspace{1em} \\
    g(m,j,c_1,b) e^{i \arg \left(\braket{\chi_{f_{ind}(m, 0, j, c_1)}^{(\ell_{m+1})}}{\chi_j^{(\ell_m)}} \right) - i \lambda_j^{(\ell_m)} \tau_m \frac{t}{r}} \ket{m}_0 \ket{0}_1 \ket{f_{ind}(m, 0, j, c_1)}_2 \\
    \quad \ket{b}_3 \ket{c_1}_4 \ket{c_2}_5 \ket{0}_6 \ket{0}_7 \ket{0}_8 \ket{0}_9 & j \in G_m^{(c_1, c_2)}
    \end{cases}
\end{align}
If we consider the action of the above algorithm on registers $r_3, r_4, r_5, r_1, r_2$, it maps a state $\ket{b} \ket{c_1} \ket{c_2} \ket{0} \ket{j}$ to $(-1)^b \ket{b} \ket{c_1} \ket{c_2} \ket{0} \ket{j} $ if $j \notin G_m^{(c_1, c_2)}$, and $g(m,j,c_1,b) e^{i \arg \left(\braket{\chi_{f_{ind}(m, 0, j, c_1)}^{(\ell_{m+1})}}{\chi_j^{(\ell_m)}} \right) - i \lambda_j^{(\ell_m)} \tau_m \frac{t}{r}} \ket{b} \ket{c_1} \ket{c_2} \ket{0} \ket{j} $ if $j \in G_m^{(c_1, c_2)}$, which is exactly the action of SEL$(\vec{U}_m)$, as required.

\section{Proof of Lemma \ref{lemma:lcu}} \label{sec:lcu_appendix}
\begin{proof}
    Lemma \ref{lemma:lcu} claims that 
    \begin{align}
        \left( (\ketbra{0}{0})^{\otimes (1 + B + 2 \log d)} \otimes \normalfont{\openone} \right) \hat{W} \ket{0}^{\otimes (B + 2 \log d)} \ket{0} \ket{\psi} = \frac{1}{d^2} \ket{0}^{\otimes (B + 2 \log d)} \widetilde{A}_m \ket{0} \ket{\psi}
    \end{align}
    where $\hat{W} := (\text{\normalfont{PREP}}^\dagger \otimes \normalfont{\openone}) \text{\normalfont{SEL}} (\vec{U}_m) (\text{\normalfont{PREP}} \otimes \normalfont{\openone})$. If we plug the definitions of PREP \eqref{def:prep} and SEL \eqref{def:sel} into the left side, we get
    \begin{align}
        & \quad \left( \big(\ketbra{0}{0} \big)^{\otimes \log (1 + B + 2 \log d)} \otimes \openone \right) (\text{\normalfont{PREP}}^\dagger \otimes \openone) \left(\text{\normalfont{SEL}}(\vec{U}_m) \right) (\text{\normalfont{PREP}} \otimes \openone) \ket{0}^{\otimes (B + 2 \log d)} \ket{0} \ket{\psi} \\
        &= \left( \frac{1}{d \sqrt{2^B}} \ket{0}^{\otimes (B + 2 \log d)} \left( \sum_{b=0}^{2^B-1} \sum_{c_1=0}^{d-1} \sum_{c_2=0}^{d-1} \bra{b} \bra{c_1} \bra{c_2} \right) \otimes \ketbra{0}{0} \otimes \openone \right) \notag \\
        & \quad \left(\text{\normalfont{SEL}}(\vec{U}_m) \right) \left( \frac{1}{d \sqrt{2^B}} \sum_{b=0}^{2^B-1} \sum_{c_1=0}^{d-1} \sum_{c_2=0}^{d-1} \ket{b} \ket{c_1} \ket{c_2} \right) \ket{0} \ket{\psi} \\
        &= \frac{1}{d^2} \left( \ket{0}^{\otimes (B + 2 \log d)} \left( \sum_{b=0}^{2^B-1} \sum_{c_1=0}^{d-1} \sum_{c_2=0}^{d-1} \bra{b} \bra{c_1} \bra{c_2} \right) \otimes \ketbra{0}{0} \otimes \openone \right)  \left( \sum_{b=0}^{2^B-1} \ket{b} \ket{c_1} \ket{c_2} \frac{1}{2^B} (\hat{X} \otimes \openone) \hat{U}_{m,b,c_1,c_2}  \ket{0} \ket{\psi} \right)\\
        &= \frac{1}{d^2} \ket{0}^{\otimes (B + 2 \log d)} \sum_{b=0}^{2^B-1} \sum_{c_1=0}^{d-1} \sum_{c_2=0}^{d-1} \frac{1}{2^B} (\hat{X} \otimes \openone) \hat{U}_{m,b,c_1,c_2}  \ket{0} \ket{\psi} \\
        &= \frac{1}{d^2} \ket{0}^{\otimes (B + 2 \log d)} \widetilde{A}_m \ket{0} \ket{\psi}
    \end{align}
\end{proof}

\section{Robust Oblivious Amplitude Amplification} \label{sec:roaa_appendix}
We provide here for completeness, a thorough discussion of robust oblivious amplitude amplification and define more carefully some of the notation used earlier in its application to Hamiltonian path integral simulations. The implementation of the process is given in the following lemma, adapted from \cite{bck2015}:
\begin{lemma}
    Let $\hat{W}$ denote the circuit implementing $(\text{\normalfont{PREP}}^\dagger \otimes \normalfont{\openone} )\text{\normalfont{SEL}}(\vec{U}_m)(\text{\normalfont{PREP}} \otimes \normalfont{\openone})$, as shown in Figure \ref{fig:lcucircuit}, where {\normalfont{PREP}} and {\normalfont{SEL}} are defined in \ref{def:prep} and \ref{def:sel}, respectively. Let $\hat{R} := -  (\normalfont{\openone} - 2 \hat{W} \hat{\Pi} \hat{W}^\dagger) (\normalfont{\openone} - 2 \hat{\Pi})$, where $\hat{\Pi}:= (\ketbra{0}{0})^{\otimes (1 + B + 2 \log d)} \otimes \normalfont{\openone}$. Then,
\begin{align*}
    \left\| \left( \hat{\Pi} \hat{R}^p \hat{W} \hat{\Pi} \right) \ket{0}^{\otimes (B + 2 \log d)} \ket{0} \ket{\psi} -  \left( (\ketbra{0}{0} )^{\otimes (B + 2 \log d)} \otimes \widetilde{A}_m \right) \ket{0}^{\otimes (B + 2 \log d)} \ket{0} \ket{\psi} \right\| \in O\left( \frac{d^2}{2^B} \right)
\end{align*} 
where $p \in \mathbb{N}$ such that $\sin\left( \frac{\pi}{2(2p+1)} \right) = \frac{1}{d^2}$, $\widetilde{A}_m$ is defined as in \eqref{tildeA}, $d$ is the sparsity of the Hamiltonian decomposition (Def. \ref{def:dsparse}), and $B$ is the number of bits of precision for the inner product oracle (Def. \ref{def:O_IM}). 
\end{lemma}
For a rigorous proof, see proof of Lemma 6 from \cite{bck2015}. The rough idea of how the robust oblivious amplitude amplification process works is that each iteration of $\hat{R}$ (which is the product of two reflections) implements, up to $\frac{1}{2^B}$ error, a rotation of the state by an angle of $\frac{\pi}{2p+1}$ in the two-dimensional subspace spanned by $\hat{W} \hat{\Pi} \ket{0}^{\otimes (1 + B + 2 \log d)} \ket{\psi}$ and $\ket{0}^{\otimes (B + 2 \log d)} \widetilde{A}_m \ket{0} \ket{\psi}$. After $p$ rotations by that angle, the state has approximately been rotated to $\ket{0}^{\otimes (B + 2 \log d)} \widetilde{A}_m \ket{0} \ket{\psi}$. It works exactly the same as oblivious amplitude amplification except that there is some error that comes in due to the non-unitarity of the operator $(\bra{0} \otimes \openone) \widetilde{A}_m (\ket{0} \otimes \openone)$.

Each iteration of $\hat{R}$ takes one application of $\hat{W}$ and one application of $\hat{W}^\dagger$, so one iteration of $\hat{R}$ uses a constant number of Select($\vec{U}_m$) and Prepare operations and their inverse operations. Since $\sin \left( \frac{\pi}{2(2p+1)} \right) = \frac{1}{d^2}$, we have $p \in O(d^2)$, which means that the final error in implementing $(\bra{0} \otimes \openone) \widetilde{A}_m (\ket{0} \otimes \openone)$ after $p$ iterations of $\hat{R}$ is $\in O \left( \frac{d^2}{2^B} \right)$, using $O(d^2)$ Select, Prepare and their inverses. Each application of Select uses a constant number of oracle queries, as well as $O(B + \log d)$ additional two-qubit gates and $B$ extra qubits for the comparison circuit (which can be reused for each application of $\hat{W}$), and each application of Prepare requires $B + 2\log d$ two-qubit gates. Thus, approximating $\widetilde{A}_m$ to within error $O \left( \frac{d^2}{2^B} \right)$ takes $O(d^2)$ oracle queries, $O \left(d^2 (B + \log d) \right)$ additional two-qubit gates, and $O (B + \log d)$ additional qubits.

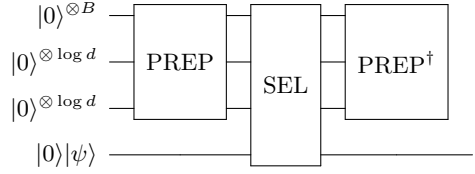
\begin{figure}
$\begin{array}{c}
    \Qcircuit @C=1em @R=1em {
        \lstick{\ket{0}^{\otimes B}} & \multigate{2}{\text{PREP}} & \multigate{3}{\text{SEL}} & \multigate{2}{\text{PREP}^\dagger} \\ 
        \lstick{\ket{0}^{\otimes \log d}} & \ghost{\text{PREP}} & \ghost{{\text{SEL}}} & \ghost{\text{PREP}^\dagger}\\
        \lstick{\ket{0}^{\otimes \log d}} & \ghost{\text{PREP}} & \ghost{\text{SEL}} & \ghost{\text{PREP}^\dagger}\\
        \lstick{\ket{0} \ket{\psi}} & \qw & \ghost{\text{SEL}} & \qw & \qw
    }
    \end{array}$
    \caption{Circuit to implement $\hat{W}$.} \label{fig:lcucircuit}
\end{figure}

\section{Long-time Hamiltonian path integral implementation} \label{sec:LTPIAppendix}
This section goes through the implementation of the time-dependent path integral in more detail. The ideas are almost exactly the same as the short-time path integral in \ref{sec:HamPathInt}, but we will provide some more details here.

Recall that the discretized truncated time evolution operator is
\begin{align}
    \bar{U}:= & \sum_{j_0} \ketbra{\chi_{j_0}(1)}{\chi_{j_0}(0)} e^{-i \frac{T}{r} \left(  \frac{1}{2}\lambda_{j_0} (0) + \frac{1}{2} \lambda_{j_0} (1) + \sum_{\ell=1}^{r-1} \lambda_{j_0} \left( \frac{\ell}{r} \right) \right)} \label{eq:0jump} \\
    & \quad + \sum_{j_0} \ketbra{\chi_{j_0}(1)}{\chi_{j_0}(0)} e^{-i \frac{T}{r} \left(  \frac{1}{2}\lambda_{j_0} (0) + \frac{1}{2} \lambda_{j_0} (1) + \sum_{\ell=1}^{r-1} \lambda_{j_0} \left( \frac{\ell}{r} \right) \right)}\sum_{j_1\ne j_0} \frac{1}{r} \Bigg( \frac{1}{2} \frac{\beta_{j_0,j_1} \left( 0 \right) \beta_{j_1,j_0}\left( 0 \right)}{iT \gamma_{j_1,j_0}\left( 0 \right)}  \nonumber \\
    & \quad + \frac{1}{2} \frac{\beta_{j_0,j_1} \left( 1 \right) \beta_{j_1,j_0}\left( 1 \right)}{iT \gamma_{j_1,j_0}\left( 1 \right)} + \sum_{\ell=0}^{r-1} \frac{\beta_{j_0,j_1} \left( \frac{\ell}{r} \right) \beta_{j_1,j_0}\left( \frac{\ell}{r} \right)}{iT \gamma_{j_1,j_0}\left( \frac{\ell}{r} \right)}\Bigg) \Biggr)\\
    & \quad + \sum_{j_1 \ne j_0} \ketbra{\chi_{j_1}(1)}{\chi_{j_0} (0)} e^{-i \frac{T}{r} \left( \frac{1}{2} \lambda_{j_0}(0) + \frac{1}{2} \lambda_{j_0} (1) + \sum_{k=1}^{r-1} \lambda_{j_0} \left(\frac{k}{r} \right) \right)} \Bigg(\frac{i \beta_{j_1,j_0}(1) e^{-i \frac{T}{r} \left( \frac{1}{2} \gamma_{j_1, j_0}(0) + \frac{1}{2} \gamma_{j_1, j_0} (1) + \sum_{\ell=1}^{r-1} \gamma_{j_1, j_0} \left( \frac{\ell}{r} \right) \right) }}{T \gamma_{j_1,j_0}(1)} \nonumber\\
    & \quad - \frac{i \beta_{j_1, j_0} (0)}{T \gamma_{j_1, j_0} (0)} \Bigg),
\end{align}
where
\begin{align}
\beta_{jk}(s) := \lim_{\delta \rightarrow 0} \frac{\braket{\chi_j(s+\delta)}{\chi_k(s)}}{\delta}= \frac{\bra{\chi_j(s)}\dot{\hat{H}}(s) \ket{\chi_k(s)}}{\lambda_j(s) - \lambda_k(s)} \label{eq:beta2}
\end{align}
and \begin{align}
    \gamma_{j_1, j_0} (s) := \lambda_{j_1}(s) - \lambda_{j_0}(s)
\end{align}
The Riemann sum error in the above expression from approximating the integrals in Eq. \eqref{eq:TDOperator} is
\begin{align}
    \|{\rm Trunc}(\mathcal{T} e^{-i T \int_0^1 \hat{H}(s) ds}) - \bar{U} \| \in O \left( \frac{T}{r^3} \Lambda \right),
\end{align}
where
\begin{align}
    \Lambda:= \max_j \left( \max |\lambda_j ''(s)| \right).
\end{align}
Using the triangle inequality, we have
\begin{align}
    & \quad \|\mathcal{T} e^{-i T \int_0^1 \hat{H}(s) ds} - \bar{U} \| \\
    &= \left\|\mathcal{T} e^{-i T \int_0^1 \hat{H}(s) ds} - \mathrm{Trunc} ( \mathcal{T} e^{-i T \int_0^1 \hat{H}(s) ds} ) - \left( \bar{U} - \mathrm{Trunc} )\mathcal{T} e^{-i T \int_0^1 \hat{H}(s) ds} ) \right) \right\| \\
    & \leq \left \|\mathcal{T} e^{-i T \int_0^1 \hat{H}(s) ds} - \mathrm{Trunc} ( \mathcal{T} e^{-i T \int_0^1 \hat{H}(s) ds} ) \right\| + \left\| \bar{U} - \mathrm{Trunc}(\mathcal{T} e^{-i T \int_0^1 \hat{H}(s) ds}) \right\| \\
    & \in O \left( \frac{\Gamma^4}{\gamma^2_{\rm min} T^2} + \frac{T}{r^3} \Lambda \right) \label{eq:lt_error2}
\end{align}
We assume that the Hamiltonian satisfies the $d$-sparse definition in Def. \ref{def:dsparse2}, and that we can enumerate the non-zero $\beta_{j_0, j_1}$ values in the same way at any reduced time $\frac{\ell}{r}$, i.e., the $0^{\rm th}$ non-zero $\beta_{j_0, j_1} \left( \frac{\ell}{r} \right)$ is the one with the smallest value of $j_1$, and the $\beta_{j_0, j_1}$ values are enumerated by ascending $j_1$. The non-zero $\beta_{j_0, j_1}  \left( \frac{\ell}{r} \right)$ values for $j_1$ are enumerated similarly, by ascending $j_0$. As with the time-independent case, we will represent the eigenvectors of $\hat{H} \left( \frac{\ell}{r} \right)$ with computational basis states $\ket{\ell} \ket{b} \ket{j}$, with the first register storing the time step number, the second register storing whether we are in the $\hat{H} \left( \frac{\ell-1}{r} \right)$ basis or the $\hat{H} \left( \frac{\ell}{r} \right)$ basis, and the third register storing the basis state number. In this new representation, the discretized truncated time evolution operator becomes 
\begin{align}
    \bar{V}:= & \sum_{j_0} \ketbra{j_0}{j_0} e^{-i \frac{T}{r} \left(  \frac{1}{2}\lambda_{j_0} (0) + \frac{1}{2} \lambda_{j_0} (1) + \sum_{\ell=1}^{r-1} \lambda_{j_0} \left( \frac{\ell}{r} \right) \right)} \\
    & \quad + \sum_{j_0} \ketbra{j_0}{j_0} e^{-i \frac{T}{r} \left(  \frac{1}{2}\lambda_{j_0} (0) + \frac{1}{2} \lambda_{j_0} (1) + \sum_{\ell=1}^{r-1} \lambda_{j_0} \left( \frac{\ell}{r} \right) \right)} \sum_{j_1 \neq j_0} \frac{1}{r} \Bigg( \frac{1}{2} \frac{\beta_{j_0,j_1} \left( 0 \right) \beta_{j_1,j_0}\left( 0 \right)}{iT \gamma_{j_1,j_0}\left( 0 \right)}  \nonumber \\
    & \quad + \frac{1}{2} \frac{\beta_{j_0,j_1} \left( 1 \right) \beta_{j_1,j_0}\left( 1 \right)}{iT \gamma_{j_1,j_0}\left( 1 \right)} + \sum_{\ell=0}^{r-1} \frac{\beta_{j_0,j_1} \left( \frac{\ell}{r} \right) \beta_{j_1,j_0}\left( \frac{\ell}{r} \right)}{iT \gamma_{j_1,j_0}\left( \frac{\ell}{r} \right)}\Bigg) \Biggr)\\
    & \quad + \sum_{j_1 \ne j_0} \ketbra{j_1}{j_0} e^{-i \frac{T}{r} \left( \frac{1}{2} \lambda_{j_0}(0) + \frac{1}{2} \lambda_{j_0} (1) + \sum_{\ell=1}^{r-1} \lambda_{j_0} \left(\frac{\ell}{r} \right) \right)} \Bigg(\frac{i \beta_{j_1,j_0}(1) e^{-i \frac{T}{r} \left( \frac{1}{2} \gamma_{j_1, j_0}(0) + \frac{1}{2} \gamma_{j_1, j_0} (1) + \sum_{\ell=1}^{r-1} \gamma_{j_1, j_0} \left( \frac{\ell}{r} \right) \right) }}{T \gamma_{j_1,j_0}(1)} \nonumber\\
    & \quad - \frac{i \beta_{j_1, j_0} (0)}{T \gamma_{j_1, j_0} (0)} \Bigg)
\end{align}
acting on the $\ket{j}$ register of the state. Since the Hamiltonian satisfies the $d$-sparse defintion in Def. \ref{def:dsparse2}, the time-dependent index function can be defined very similarly to the short-time one in Def.  \ref{def:f_ind}.
\begin{definition}[Time-dependent index function]
    Let $f_{ind,TD}: \{0, \ldots, r-1 \} \times \{ 0, 1\} \times \{ 0, \ldots 2^n-1 \} \times \{ 0, \ldots, d-1 \} \rightarrow \{ 0, \ldots , 2^n-1 \}$, $f_{ind,TD}(\ell, b, j_0, p)$ is defined to be the index $j_1$ such that $\beta_{j_0, j_1} \left( \frac{\ell}{r} \right)$ is the $p^{\rm th}$ non-zero $\beta_{j_0, j_1} \left( \frac{\ell}{r} \right)$ for $j_0$ if $b=0$ and the $p^{\rm th}$ non-zero $\beta_{j_0, j_1} \left( \frac{\ell}{r} \right)$ value for $j_1$ if $b=1$, according to the enumeration defined above.
\end{definition} 
 We can use the same graph colouring scheme from Section \ref{sec:graphcolouring} to sort the non-zero transitions into $d^2$ different colours, defining the graphs in a similar manner to the short-time case, with edges between the vertices corresponding to two states $j_0$ and $j_1$ if $\beta_{j_1,j_0}$ is nonzero. Then, we can write $\bar{V}$ as
\begin{align}
    \bar{V}:= & \sum_{j_0} \ketbra{j_0}{j_0} e^{-i \frac{T}{r} \left(  \frac{1}{2}\lambda_{j_0} (0) + \frac{1}{2} \lambda_{j_0} (1) + \sum_{\ell=1}^{r-1} \lambda_{j_0} \left( \frac{\ell}{r} \right) \right)} \\
    & \quad + \sum_{c_1=0}^{d-1} \sum_{c_2=0}^{d-1} \sum_{\substack{j_0: (0,j_0) \in \\ V(G_\ell^{(c_1, c_2)})}} e^{-i \frac{T}{r} \left(  \frac{1}{2}\lambda_{j_0} (0) + \frac{1}{2} \lambda_{j_0} (1) + \sum_{\ell=1}^{r-1} \lambda_{j_0} \left( \frac{\ell}{r} \right) \right)} \frac{1}{r} \Bigg( \frac{1}{2} \frac{\beta_{j_0,f_{ind,TD} (\ell,0,j_0,c_1)} \left( 0 \right) \beta_{f_{ind,TD} (\ell,0,j_0,c_1),j_0}\left( 0 \right)}{iT \gamma_{j_1,j_0}\left( 0 \right)}  \nonumber \\
    & \quad + \frac{1}{2} \frac{\beta_{j_0,j_1} \left( 1 \right) \beta_{f_{ind,TD} (\ell,0,j_0,c_1),j_0}\left( 1 \right)}{iT \gamma_{j_1,j_0}\left( 1 \right)} + \sum_{\ell=0}^{r-1} \frac{\beta_{j_0,f_{ind,TD} (\ell,0,j_0,c_1)} \left( \frac{\ell}{r} \right) \beta_{f_{ind,TD} (\ell,0,j_0,c_1),j_0}\left( \frac{\ell}{r} \right)}{iT \gamma_{f_{ind,TD} (\ell,0,j_0,c_1),j_0}\left( \frac{\ell}{r} \right)}\Bigg) \Biggr)\\
    & \quad + \ketbra{f_{ind,TD} (\ell,0,j_0,c_1)}{j_0} e^{-i \frac{T}{r} \left( \frac{1}{2} \lambda_{j_0}(0) + \frac{1}{2} \lambda_{j_0} (1) + \sum_{\ell=1}^{r-1} \lambda_{j_0} \left(\frac{\ell}{r} \right) \right)} \nonumber \\
    & \quad \Bigg(\frac{i \beta_{f_{ind,TD} (\ell,0,j_0,c_1),j_0}(1) e^{-i \frac{T}{r} \left( \frac{1}{2} \gamma_{f_{ind,TD} (\ell,0,j_0,c_1), j_0}(0) + \frac{1}{2} \gamma_{f_{ind,TD} (\ell,0,j_0,c_1), j_0} (1) + \sum_{\ell=1}^{r-1} \gamma_{f_{ind,TD} (\ell,0,j_0,c_1), j_0} \left( \frac{\ell}{r} \right) \right) }}{T \gamma_{f_{ind,TD} (\ell,0,j_0,c_1),j_0}(1)} \nonumber\\
    & \quad - \frac{i \beta_{f_{ind,TD} (\ell,0,j_0,c_1), j_0} (0)}{T \gamma_{f_{ind,TD} (\ell,0,j_0,c_1), j_0} (0)} \Bigg)
\end{align}
Unlike the short-time Hamiltonian path integral, which required repeated applications of the short-time step time evolution operator, this operator $\bar{V}$ implements the entire time evolution from start to finish.

We will define the index oracle as in  Def. \ref{def:O_ind}, but replacing the inner product in the definition of the index function with $\beta_{j_0, j_1}$. The inner product and eigenvalues oracles in Section \ref{sec:HamPathInt} will be replaced with oracles that implement the corresponding amplitudes and phases for the time-dependent time evolution:

\begin{definition}[Time-dependent index oracle]
    Let $O_{ind, TD} \in \mathcal{L} (\mathcal{H}_r \otimes \mathcal{H}_2 \otimes \mathcal{H}_{2^n} \otimes \mathcal{H}_d \otimes \mathcal{H}_{2^n})$ be a unitary operator, such that its action on an arbitrary computational basis state is
    \begin{align}
        O_{ind, TD} : \ket{\ell} \ket{b} \ket{j} \ket{p} \ket{c} \mapsto \ket{\ell} \ket{b} \ket{j} \ket{p} \ket{c \oplus f_{ind,TD}(\ell, b, j, p)}.
    \end{align}
\end{definition}

To simplify the transition amplitudes, let
 \begin{align}
     \eta_{j_0, j_1} (\ell) := \begin{cases}
         -\frac{i \beta_{j_1, j_0} \left( \frac{\ell}{r} \right)}{T \gamma_{j_1, j_0} \left( \frac{\ell}{r} \right)} & \ell=0 \\
         \frac{i \beta_{j_1, j_0} \left( \frac{\ell}{r} \right)}{T \gamma_{j_1, j_0} \left( \frac{\ell}{r} \right)} & \ell=r \\
         0 & \text{otherwise}
     \end{cases} ,
 \end{align}
 and let
 \begin{align}
     \zeta_{j_0, j_1} (\ell) := \begin{cases}
        \frac{i \beta_{j_0, j_1} \left( \frac{\ell}{r} \right) \beta_{j_1, j_0} \left( \frac{\ell}{r} \right) }{2 r T \gamma_{j_1, j_0} \left( \frac{\ell}{r} \right)} & \ell=0 \text{ or } \ell=r\\
         \frac{i \beta_{j_0, j_1}\left( \frac{\ell}{r} \right)  \beta_{j_1, j_0}\left( \frac{\ell}{r} \right) }{r T \gamma_{j_1, j_0} \left( \frac{\ell}{r} \right)} & \text{otherwise}
     \end{cases}.
 \end{align}

\begin{definition}[$\eta$ magnitude oracle]
    Let $B \in \mathbb{N}$. Let $O_{\eta M} \in \mathcal{L} (\mathcal{H}_{r+1} \otimes \mathcal{H}_{2^n} \otimes \mathcal{H}_{2^n} \otimes \mathcal{H}_B)$ be a unitary operator, such that its action on an arbitrary computational basis state is
    \begin{align}
        O_\eta : \ket{\ell} \ket{j_0}\ket{j_1} \ket{c} \mapsto \ket{\ell} \ket{j_0} \ket{j_1} \ket{c \oplus \left[ 2^B |\eta_{j_0, j_1}(\ell)| \right]_B},
    \end{align}
    where $\beta$ is defined in Eq. \eqref{eq:beta2} and $[\cdot]_B$ denotes the bitstring representation in $B$ bits.
\end{definition}

\begin{definition}[$\eta$ phase oracle]
    Let $B \in \mathbb{N}$. Let $O_{\eta M} \in \mathcal{L} (\mathcal{H}_{r+1} \otimes \mathcal{H}_{2^n} \otimes \mathcal{H}_{2^n})$ be a unitary operator, such that its action on an arbitrary computational basis state is
    \begin{align}
        O_{\eta P} : \ket{\ell} \ket{j_0} \ket{j_1} \mapsto e^{i \mathrm{arg} \left( \eta_{j_0, j_1}(\ell) \right)} \ket{\ell} \ket{j_0} \ket{j_1}.
    \end{align}
 \end{definition}

\begin{definition}[$\zeta$ magnitude oracle]
Let $B \in \mathbb{N}$. Let $O_{\zeta M} \in \mathcal{L} (\mathcal{H}_{r+1} \otimes \mathcal{H}_{2^n} \otimes \mathcal{H}_{2^n} \otimes \mathcal{H}_B)$ be a unitary operator, such that its action on an arbitrary computational basis state is
    \begin{align}
        O_\zeta : \ket{\ell} \ket{j_0}\ket{j_1} \ket{c} \mapsto \ket{\ell} \ket{j_0} \ket{j_1} \ket{c \oplus \left[ 2^B |\zeta_{j_0, j_1}(\ell)| \right]_B},
    \end{align}
    where $\beta$ is defined in Eq. \eqref{eq:beta2} and $[\cdot]_B$ denotes the bitstring representation in $B$ bits.
\end{definition}

\begin{definition}[$\zeta$ phase oracle]
    Let $B \in \mathbb{N}$. Let $O_{\eta M} \in \mathcal{L} (\mathcal{H}_{r+1} \otimes \mathcal{H}_{2^n} \otimes \mathcal{H}_{2^n})$ be a unitary operator, such that its action on an arbitrary computational basis state is
    \begin{align}
        O_{\zeta P} : \ket{\ell} \ket{j_0} \ket{j_1} \mapsto e^{i \mathrm{arg} \left( \zeta_{j_0, j_1}(\ell) \right)} \ket{\ell} \ket{j_0} \ket{j_1}.
    \end{align}
 \end{definition} 

\begin{definition}[Time-dependent eigenvalue phase oracle]
    Let $O_{TEP} \in \mathcal{L} (\mathcal{H}_{r+1} \otimes \mathcal{H}_{2^n})$ be a unitary operator, such that its action on an arbitrary basis state is
    \begin{align}
        O_{TEP} : \ket{j} \mapsto e^{-i \frac{T}{r} \left(  \frac{1}{2}\lambda_{j} (\ell) + \frac{1}{2} \lambda_{j} (1) + \sum_{\ell=1}^{r-1} \lambda_{j} \left( \frac{\ell}{r} \right) \right)} \ket{j}
    \end{align}
\end{definition}

\begin{definition}[$\gamma$ oracle]
    Let $O_\gamma \in \mathcal{L} (\mathcal{H}_{2^n} \otimes \mathcal{H}_{2^n})$ be a unitary operator, such that its action on an arbitrary basis state is
    \begin{align}
        O_\gamma :\ket{j_0} \ket{j_1} \mapsto e^{-i \frac{T}{r}  \left( \frac{1}{2} \gamma_{j_1, j_0}(0) + \frac{1}{2} \gamma_{j_1, j_0} (1) + \sum_{\ell=1}^{r-1} \gamma_{j_1, j_0} \left( \frac{\ell}{r} \right) \right) } \ket{j_0} \ket{j_1} 
    \end{align}
\end{definition}

We will consider three different parts of the operator separately. The first part, \eqref{eq:0jump}, implements the zero-jump paths. This part is 
\begin{align}
    \bar{V}_0 := \sum_{j_0} \ketbra{j_0}{j_0} e^{-i \frac{T}{r} \left( \frac{1}{2} \lambda_{j_0}(0) + \frac{1}{2} \lambda_{j_0}(1) + \sum_{\ell=1}^{r-1} \lambda_{j_0} \left( \frac{\ell}{r} \right) \right)}
\end{align}
Note that the zero-jump path can be implemented using one query to the oracle $O_{TEP}$.

The part of the operator that implements the one-jump paths is
\begin{align}
    \bar{V}_1 :&= \sum_{j_0=0}^{2^n-1} \sum_{j_1=0}^{2^n-1} \ketbra{j_1}{j_0} e^{-i \frac{T}{r} \left( \frac{1}{2} \lambda_{j_0}(0) + \frac{1}{2} \lambda_{j_0} (1) + \sum_{\ell=1}^{r-1} \lambda_{j_0} \left(\frac{\ell}{r} \right) \right)} \Bigg(\frac{i \beta_{j_1,j_0}(1) e^{-i \frac{T}{r} \left( \frac{1}{2} \gamma_{j_1, j_0}(0) + \frac{1}{2} \gamma_{j_1, j_0} (1) + \sum_{\ell=1}^{r-1} \gamma_{j_1, j_0} \left( \frac{\ell}{r} \right) \right) }}{T \gamma_{j_1,j_0}(1)} \nonumber\\
    & \quad - \frac{i \beta_{j_1, j_0} (0)}{T \gamma_{j_1, j_0} (0)} \Bigg)  \\
    &= \sum_{s=0}^{1} \sum_{j_0=0}^{2^n-1} \sum_{j_1=0}^{2^n-1} \ketbra{j_1}{j_0} e^{-i \frac{T}{r} \left( \frac{1}{2} \lambda_{j_0}(0) + \frac{1}{2} \lambda_{j_0} (1) + \sum_{\ell=1}^{r-1} \lambda_{j_0} \left(\frac{\ell}{r} \right) \right)} \Bigg( \eta_{j_1,j_0}(sr) e^{-i \frac{s T}{r}  \left( \frac{1}{2} \gamma_{j_1, j_0}(0) + \frac{1}{2} \gamma_{j_1, j_0} (1) + \sum_{\ell=1}^{r-1} \gamma_{j_1, j_0} \left( \frac{\ell}{r} \right) \right) }\Bigg) \label{eq:1jump}
\end{align}
Because we defined $\beta_{j_0, j_1}$ to be zero when $j_0 = j_1$, the second summation over $j_1 \neq j_0$ has been changed to a sum over all $j_1$, since the amplitude of the $j_0 = j_1$ paths is zero. We have an operator that implements a transition where each computational basis state is mapped to at most $d$ other computational basis states, and each computational basis state has at most $d$ computational basis states that are mapped to it. This is exactly the same scenario as the time-independent case, so the same technique of using graph colouring to split the transitions into $d^2$ groups will be used. The same colour oracle defined in Def. \ref{def:O_C} can be used, except now the graph that is being coloured is one with vertices representing $\ket{\chi_{j_0} \left( \frac{\ell}{r} \right)}$ and $\ket{\chi_{j_1} \left( \frac{\ell}{r} \right)}$ states, and the edges representing non-zero $\beta_{j_0, j_1} \left( \frac{\ell}{r} \right)$ values.

Using this graph colouring, we can implement $\bar{V}_1$ up to $O \left( \frac{d^2}{2^B} \right)$ error using the same alternating sign trick with linear combination of unitaries as in the time-independent case, querying the phase oracles as required to get the correct phases on the transitions. The $O \left( \frac{d^2}{2^B} \right)$ error comes in due to the precision of storing $\beta_{j_0, j_1} \left( \frac{\ell}{r} \right)$ in $B$ bits. The proof of this is pretty much the same as the proof of Claim \ref{clm:ast} in Appendix \ref{sec:ast_full_proof}. 

Finally, we will consider the two-jump paths that begin and end at the same state, i.e., paths stay on $\ket{\chi_{j_0}(t)}$ for time $t \in [0, t')$ and $\ket{\chi_{j_1}(t)}$ for time $t \in [t', T]$, where $j_0 \neq j_1$. The part of the operator that implements these paths is
\begin{align}
    \bar{V}_2 :&= \sum_{j_0=0}^{2^n-1} \sum_{j_1=0}^{2^n-1} \ketbra{j_0}{j_0} e^{-i \frac{T}{r} \left(  \frac{1}{2}\lambda_{j_0} (0) + \frac{1}{2} \lambda_{j_0} (1) + \sum_{\ell=1}^{r-1} \lambda_{j_0} \left( \frac{\ell}{r} \right) \right)}  \frac{1}{r} \Bigg( \frac{1}{2} \frac{\beta_{j_0,j_1} \left( 0 \right) \beta_{j_1,j_0}\left( 0 \right)}{iT \gamma_{j_1,j_0}\left( 0 \right)}  \nonumber \\
    & \quad + \frac{1}{2} \frac{\beta_{j_0,j_1} \left( 1 \right) \beta_{j_1,j_0}\left( 1 \right)}{iT \gamma_{j_1,j_0}\left( 1 \right)} + \sum_{\ell=1}^{r-1} \frac{\beta_{j_0,j_1} \left( \frac{\ell}{r} \right) \beta_{j_1,j_0}\left( \frac{\ell}{r} \right)}{iT \gamma_{j_1,j_0}\left( \frac{\ell}{r} \right)}\Bigg) \Biggr) \\
    &= \sum_{j_0=0}^{2^n-1} \sum_{j_1=0}^{2^n-1} \sum_{\ell=0}^{r} \ketbra{j_0}{j_0} e^{-i \frac{T}{r} \left(  \frac{1}{2}\lambda_{j_0} (0) + \frac{1}{2} \lambda_{j_0} (1) + \sum_{\ell'=1}^{r-1} \lambda_{j_0} \left( \frac{\ell'}{r} \right) \right)} \zeta_{j_0, j_1} (\ell)
\end{align}
Unlike the one-jump case, we also have a sum over $\ell$, which comes from the Riemann sum discretization of the integral
\begin{align}
    \int_0^1 \frac{\beta_{j_0, j_1} (s) \beta_{j_1, j_0} (s)}{i T \gamma_{j_1, j_0} (s)} \mathrm{d}s
\end{align}
in Eq. \ref{eq:TDOperator}. This does not affect the graph colouring, since the $d$-sparsity holds for all $\ell$, and the extra sum is easily implemented as part of the linear combination of unitaries process. However, now the error in implementing $\bar{V}_2$ using the LCU/alternating sign trick method scales with $\frac{\ell d^2}{2^B}$.

The entire sum $\bar{V} = \bar{V}_0 + \bar{V}_1 + \bar{V}_2$ can be approximately implemented using the linear combination of unitaries technique, since each $\bar{V}_i$ can be written as a linear combination of unitaries using the alternating sign trick, and $\bar{V}$ is simply the sum of these $\bar{V}_i$ terms. As with the time-independent case, we will need to use robust oblivious amplitude amplification in order to boost the amplitude of the correct final state arbitrarily close to 1.

Let $\ket{\psi_\ell} = \sum_{j=0}^{2^n-1} a_j \ket{\chi_{j} \left( \frac{\ell}{r} \right)}$ be the state at the $\ell^{\rm th}$ time step. As mentioned earlier, we will represent the state on a quantum computer as
\begin{align}
    \sum_{j=0}^{2^n-1} a_j \ket{\ell} \ket{0} \ket{j}
\end{align}
The Prepare and Select operations for the linear combination of unitaries are defined as follows:
\begin{definition}[Prepare operation]
    Let $\mathcal{H}_{\rm anc}$ be a finite dimensional space over our ancillary qubits, which we further specify as $\mathcal{H}_{\rm anc}=(\mathcal{H}_{r+1} \otimes \mathcal{H}_4 \otimes \mathcal{H}_{2^B} \otimes \mathcal{H}_d \otimes \mathcal{H}_d)$, then let the Prepare operation  $ {\rm PREP}\in \mathcal{L}(\mathcal{H}_{\rm anc})$
    be a unitary operator,
    \begin{align}
        \text{\normalfont{PREP}} &:  \ket{0}^{\otimes \log (r+1)} \ket{0}^{\otimes 2} \ket{0}^{B} \ket{0}^{\otimes \log d} \ket{0}^{\otimes \log d} \notag \\
        & \mapsto \frac{1}{d \sqrt{(r+1)2^B}} \sum_{\ell=0}^{r} \sum_{b=0}^{2^B-1} \sum_{c_1=0}^{d-1} \sum_{c_2=0}^{d-1} \ket{\ell} \frac{1}{\sqrt{\frac{1}{(r+1)d^2} + 2}} \left( \frac{1}{\sqrt{r+1}d} \ket{0} + \ket{1} + \ket{2} \right) \ket{b} \ket{c_1} \ket{c_2}
    \end{align}
\end{definition}
\begin{definition}[Select operation] \label{def:sel2}
    Let \normalfont{SEL} $\in \mathcal{L} (\mathcal{H}_{\rm anc} 
    \otimes \mathcal{H}_2 \otimes \mathcal{H}_{2^n})$ be a unitary operator,
    \begin{align}
        \text{\normalfont{SEL}} := \sum_{\ell=0}^{r} \sum_{p=0}^{2} \sum_{b=0}^{2^B-1} \sum_{c_1=0}^{d-1} \sum_{c_2=0}^{d-1} \ketbra{\ell}{\ell} \otimes \ketbra{p}{p} \otimes \ketbra{b}{b} \otimes \ketbra{c_1}{c_1} \otimes \ketbra{c_2}{c_2} \otimes \left( (\hat{X} \otimes \normalfont{\openone}) \widetilde{V}_{\ell, p, b, c_1, c_2} \right)
    \end{align}
    where $\widetilde{V}_{\ell, p, b, c_1, c_2}$ is the operator that implements the transition corresponding to the edge of colour $(c_1, c_2)$ with the proper phase for the $\bar{V}_p$ operator. More specifically,
    \begin{align}
         \widetilde{V}_{\ell, 0, b, c_1, c_2}&:= \sum_{j=0}^{2^n-1} \left(  e^{-i \frac{T}{r} \left( \frac{1}{2} \lambda_j (0) + \frac{1}{2} \lambda_j (1) + \sum_{\ell'=1}^{r-1} \lambda_j \left( \frac{\ell'}{r} \right) \right)} \ketbra{1}{0} + (-1)^b \ketbra{0}{1} \right) \otimes \ketbra{j}{j},
    \end{align}
    \begin{align}
        \widetilde{V}_{0, 1, b, c_1, c_2}&:= \sum_{\substack{j: (0,j) \in \\ V(G_\ell^{(c_1, c_2)})}} \bigg( g_1(m, j, c_1, b) e^{-i \frac{T}{r} \left( \frac{1}{2} \lambda_j (0) + \frac{1}{2} \lambda_j (1) + \sum_{\ell'=1}^{r-1} \lambda_j \left( \frac{\ell'}{r} \right) \right)} e^{i \arg (\eta_{j, f_{ind, TD}(\ell, 0, j, c_1)}(\ell))} \notag \\
        & \quad \ketbra{1}{0} \otimes \ketbra{f_{ind,TD}(\ell, 0, j, c_1)}{j}  + (-1)^b \ketbra{0}{1} \otimes \ketbra{j}{f_{ind,TD}(\ell, 0, j, c_1)} \bigg) + \sum_{\substack{j: (0,j) \notin \\ V(G_\ell^{(c_1, c_2)})}} (-1)^b \ketbra{0}{0} \otimes \ketbra{j}{j} \notag \\
        & \quad + \sum_{\substack{q:(1, q) \notin \\ V(G_\ell^{(c_1,c_2)})}} (-1)^b \ketbra{1}{1} \otimes \ketbra{q}{q},
    \end{align}
    \begin{align}
        \widetilde{V}_{r, 1, b, c_1, c_2}&:= \sum_{\substack{j: (0,j) \in \\ V(G_\ell^{(c_1, c_2)})}} \bigg( g_1(m, j, c_1, b) e^{-i \frac{T}{r} \left( \frac{1}{2} \lambda_j (0) + \frac{1}{2} \lambda_j (1) + \sum_{\ell'=1}^{r-1} \lambda_j \left( \frac{\ell'}{r} \right) \right)} e^{i \arg (\eta_{j, f_{ind, TD}(\ell, 0, j, c_1)}(\ell))} \notag \\
        & \quad  \left( e^{-i \frac{T}{r} \left( \frac{1}{2} \gamma_{f_{ind,TD}(\ell,0,j,c_1),j)}(0) +  \frac{1}{2} \gamma_{f_{ind,TD}(\ell,0,j,c_1),j)}(1) + \sum_{\ell'=0}^{r-1}  \gamma_{f_{ind,TD}(\ell,0,j,c_1),j}\left( \frac{\ell'}{r} \right)  \right)} \right) \ketbra{1}{0} \notag \\
        & \quad \otimes \ketbra{f_{ind,TD}(\ell, 0, j, c_1)}{j}  + (-1)^b \ketbra{0}{1} \otimes \ketbra{j}{f_{ind,TD}(\ell, 0, j, c_1)} \bigg) + \sum_{\substack{j: (0,j) \notin \\ V(G_\ell^{(c_1, c_2)})}} (-1)^b \ketbra{0}{0} \otimes \ketbra{j}{j} \notag \\
        & \quad + \sum_{\substack{q:(1, q) \notin \\ V(G_\ell^{(c_1,c_2)})}} (-1)^b \ketbra{1}{1} \otimes \ketbra{q}{q},
    \end{align}
        \begin{align}
        \widetilde{V}_{\ell, 1, b, c_1, c_2}&:= \sum_{j=0}^{2^n-1} (-1)^b \ketbra{1}{0} \otimes \ketbra{j}{j} + (-1)^b \ketbra{0}{1} \otimes \ketbra{j}{j} \text{ for } \ell \neq 0 \text{ or } r,
    \end{align}
    \begin{align}
        \widetilde{V}_{\ell, 2, b, c_1, c_2} &:=  \sum_{\substack{j: (0,j) \in \\ V(G_\ell^{(c_1, c_2)})}} \bigg( g_2(m, j, c_1, b) e^{-i \frac{T}{r} \left( \frac{1}{2} \lambda_j (0) + \frac{1}{2} \lambda_j (1) + \sum_{\ell'=1}^{r-1} \lambda_j \left( \frac{\ell'}{r} \right) \right)} e^{i \arg (\zeta_{j, f_{ind, TD}(\ell, 0, j, c_1)}(\ell))} \ketbra{1}{0} \notag \\
        & \quad \otimes \ketbra{f_{ind,TD}(\ell, 0, j, c_1)}{j}  + (-1)^b \ketbra{0}{1} \otimes \ketbra{j}{f_{ind,TD}(\ell, 0, j, c_1)} \bigg) + \sum_{\substack{j: (0,j) \notin \\ V(G_\ell^{(c_1, c_2)})}} (-1)^b \ketbra{0}{0} \otimes \ketbra{j}{j} \notag \\
        & \quad + \sum_{\substack{q:(1, q) \notin \\ V(G_\ell^{(c_1,c_2)})}} (-1)^b \ketbra{1}{1} \otimes \ketbra{q}{q},
    \end{align}
    where
    \begin{align}
        g_1(\ell, j, c_1, b) &= \begin{cases}
        1 & b < \left[ 2^B |\eta_{j_0, f_{ind, TD}(\ell, 0, j_0, c_1)}(\ell)| \right]_B \\
        (-1)^b & \text{otherwise}
        \end{cases} \\
        g_2(\ell, j, c_1, b) &= \begin{cases}
        1 & b < \left[ 2^B |\zeta_{j_0, f_{ind, TD}(\ell, 0, j_0, c_1)}(\ell)| \right]_B \\
        (-1)^b & \text{otherwise}
        \end{cases}
    \end{align}
\end{definition}
\begin{lemma} \label{lemma:lt_prepsel}
We can implement PREP and SEL on an input state in $\mathcal{L} (\mathcal{H}_{r+1} \otimes \mathcal{H}_4 \otimes \mathcal{H}_{2^B} \otimes \mathcal{H}_d \otimes \mathcal{H}_d \otimes \mathcal{H}_2 \otimes \mathcal{H}_{2^n} )$ using a constant number of queries to the oracles defined above and a constant number of additional two-qubit gates, using $n + B + 2$ ancilla qubits.
\end{lemma}
\begin{proof}
    See below for full algorithm to implement \normalfont{PREP} and \normalfont{SEL}.
\end{proof}
The input state is shown in the table below:
\begin{center}
    \begin{tabular}{ |c|c|c|c| } 
        \hline
        Register \# & Initial State & Size (Qubits) & Description \\ 
        \hline
        0 & $\ket{0}$ & $\log (r+1)$ & stores time step number $\ell$ \\
        1 & $\ket{0}$ & 1 & stores whether state corresponds to the $\hat{H} \left( \frac{\ell-1}{r} \right)$ or the $\hat{H} \left( \frac{\ell}{r} \right)$ basis\\ 
        2 & $\ket{\psi_m}$ & $n$ & stores basis state indices for $\ket{\psi}$ \\
        3 & $\ket{0}$ & 2 & stores number of jumps \\
        4 & $\ket{0}$ & $B$ & stores $b$ index for alternating sign trick\\ 
        5 & $\ket{0}$ & $\log d$ & stores $c_1$ \\ 
        6 & $\ket{0}$ & $\log d$ & stores $c_2$ \\ 
        7 & $\ket{0}$ & 1 & stores whether vertex is in graph of colour $(c_1, c_2)$ \\
        8 & $\ket{0}$ & $n$ & stores $f_{ind,TD}(m, 0, j, c_1)$ \\ 
        9 & $\ket{0}$ & $B$ & stores value of transition amplitude magnitude in $B$ bits \\
        10 & $\ket{0}$ & 1 & stores whether $b \geq$ transition amplitude magnitude\\
        \hline
    \end{tabular}
\end{center}
This is very similar to the table for the time-independent input state shown in Appendix \ref{sec:lcu}, except that there is an extra register (register \#3), which stores the number of jumps in the operator we are implementing (0, 1 or 2). The algorithm in \ref{alg:prep2} describes how to implement the PREP$(r_0, r_3, r_4, r_5, r_6)$, and the algorithm in \ref{alg:sel2} describes how to implement SEL$(r_0, r_3, r_4, r_5, r_6, r_1, r_2)$.
\begin{algorithm}[H]
\begin{algorithmic}[1] \label{alg:prep2}
\State Use a single two-qubit gate to prepare $\frac{1}{\sqrt{\frac{1}{(r+1)d^2}+2}}\left( \frac{1}{\sqrt{(r+1)}d} \ket{0} + \ket{1} + \ket{2} \right)$ on $r_3$.
\State Apply Hadamard gates to $r_0, r_4, r_5$ and $r_6$
\end{algorithmic}
\caption{Prepare}
\end{algorithm}
The Select (SEL) operation is almost exactly the same as the time-independent one shown in \ref{alg:sel}, but with the operations also controlled on $r_3$ to implement the sum of the three different $\bar{U}_i$ operators.
\begin{algorithm}[H]
\caption{Select} \label{alg:sel2}
\begin{algorithmic}[1]
\item Apply $O_C(r_0, r_1, r_2, r_5, r_6, r_7)$, controlled on $r_3 = 2$ \Comment{checks whether state is in graph of colour $(c_1, c_2)$ for two-jump paths}
\item Apply $O_C(r_0, r_1, r_2, r_5, r_6, r_7)$, controlled on $r_3 = 1$ and $r_0 = 0$ or $r$ \Comment{checks whether state is in graph of colour $(c_1, c_2)$ for one-jump paths}
\item Apply $O_{ind,TD} (r_0, r_1, r_2, r_5, r_8)$, controlled on $r_7$ \Comment{computes index of state that $\ket{j}$ gets mapped to if it's in the graph of colour $(c_1, c_2)$}
\item Apply $O_{\eta}(r_0, r_2, r_8, r_9)$, controlled on $r_1 =0$, $r_3=1$ and $r_7=1$ \Comment{computes inner product magnitude for $\bar{V}_1$ if edge is in graph}
\item Apply $O_{\zeta}(r_0, r_2, r_8, r_9)$, controlled on $r_1=0$, $r_3=2$ and $r_7=1$ \Comment{computes inner product magnitude for $\bar{V}_2$ if edge is in graph}
\item Apply $U_{\leq} (r_4, r_9, r_{10})$, controlled on $r_3 = 1$ or 2 \Comment{marks states where $b \geq 2^B \times $ amplitude magnitude}
\item Apply $O_{TEP} (r_2)$, controlled on $r_3=0$ or $r_7=1$ \Comment{applies eigenvalue phase for all jumps}
\item Apply $O_{\eta P}(r_0, r_2, r_8)$ controlled on $r_1=0$, $r_3=1$, $r_7=1$ \Comment{computes $\eta$ phase for $\bar{V}_1$} 
\item Apply $O_{\gamma}(r_2, r_8)$, controlled on $r_0=\ell$, $r_1 = 0$, $r_3=1$ and $r_7=1$ \Comment{computes $\gamma$ phase for $\bar{V}_1$}
\item Apply $O_{\zeta P}(r_0, r_2, r_8)$ and $O_P(r_0, r_2)$, controlled on $r_1=0$, $r_3=2$, $r_7=1$ \Comment{computes phases for $\bar{V}_2$} 
\item Swap $r_2$ and $r_8$ and apply $X(r_1)$, controlled on $r_3=1$ and $r_7=1$ \Comment{implements $\bar{V}_1$ transition if edge is in graph of colour $(c_1, c_2)$}
\item Apply $(-1)^b$ to $r_0$, controlled on $r_{10}$ \Comment{cancels out states with 0 inner product, applies amplitudes}
\item Apply $\hat{U}_{\leq} (r_4, r_9, r_{10})$, controlled on $r_3=1$ or 2 \Comment{Sets $r_{10}$ back to 0}
\item Apply $O_{\eta}(r_0, r_8, r_2, r_9)$ controlled on $r_1=1$, $r_3=1$ and $r_7=1$ \Comment{sets $r_9$ back to 0 for 1-jump paths}
\item Apply $O_{\zeta}(r_0, r_2, r_8, r_9)$ controlled on $r_1=0$, $r_3=2$ and $r_7=1$ \Comment{sets $r_9$ back to 0 for 2-jump paths}
\item Apply $O_{ind,TD}(r_0, r_1, r_2, r_6, r_8)$ controlled on $r_3=1$, $r_7=1$ \Comment{sets $r_8$ back to 0 for $\bar{V}_1$}
\item Apply $O_{ind,TD}(r_0, r_1, r_2, r_5, r_8)$ controlled on $r_3=2$, $r_7=1$ \Comment{sets $r_8$ back to 0 for $\bar{V}_2$}
\item Apply $O_C(r_0, r_1, r_2, r_5, r_6, r_7)$ controlled on $r_3=1$ or 2  \Comment{Sets $r_7$ back to 0}
\item Apply $X(r_1)$ controlled on $r_3=1$\Comment{Sets $r_1$ back to 0}
\end{algorithmic}
\end{algorithm}
To see that the algorithm above implements the Select operator described in Def \ref{def:sel2}, consider how the algorithm above would act on an input state $\ket{\ell}_0 \ket{0}_1 \ket{j}_2 \ket{q}_3 \ket{b}_4 \ket{c_1}_5 \ket{c_2}_6 \ket{0}_7 \ket{0}_8 \ket{0}_9 \ket{0}_{10}$ for $q = 0, 1$ and 2:

$q=0$:
\begin{align}
    \text{Input: } & \ket{\ell}_0 \ket{0}_1 \ket{j}_2 \ket{0}_3 \ket{b}_4 \ket{c_1}_5 \ket{c_2}_6 \ket{0}_7 \ket{0}_8 \ket{0}_9 \ket{0}_{10}\\
    \xrightarrow{\text{Steps 1-7}} & e^{-i \frac{T}{r} \left( \frac{1}{2} \lambda_j(0) + \frac{1}{2} \lambda_j (1) + \sum_{\ell'=0}^{r-1} \lambda_j \left( \frac{\ell'}{r} \right) \right)} \ket{\ell}_0 \ket{0}_1 \ket{j}_2 \ket{0}_3 \ket{b}_4 \ket{c_1}_5 \ket{c_2}_6 \ket{0}_7 \ket{0}_8 \ket{0}_9 \ket{0}_{10} 
\end{align}

$q=1$:
\begin{align}
    & \text{Input: } \ket{\ell}_0 \ket{0}_1 \ket{j}_2 \ket{1}_3 \ket{b}_4 \ket{c_1}_5 \ket{c_2}_6 \ket{0}_7 \ket{0}_8 \ket{0}_9 \ket{0}_{10}\\
    \xrightarrow{\text{1-2}} &  \begin{cases} \ket{\ell}_0 \ket{0}_1 \ket{j}_2 \ket{1}_3 \ket{b}_4 \ket{c_1}_5 \ket{c_2}_6 \ket{1}_7 \ket{0}_8 \ket{0}_9 \ket{0}_{10} & \ell = 0 \text{ or } r, j \in G_\ell^{(c_1, c_2)}\\
     \ket{\ell}_0 \ket{0}_1 \ket{j}_2 \ket{1}_3 \ket{b}_4 \ket{c_1}_5 \ket{c_2}_6 \ket{0}_7 \ket{0}_8 \ket{0}_9 \ket{0}_{10} & \text{otherwise}
    \end{cases} \\
    \xrightarrow{\text{3}} &  \begin{cases} \ket{\ell}_0 \ket{0}_1 \ket{j}_2 \ket{1}_3 \ket{b}_4 \ket{c_1}_5 \ket{c_2}_6 \ket{1}_7 \ket{0}_8 \ket{0}_9 \ket{0}_{10} & \ell = 0 \text{ or } r, j \in G_\ell^{(c_1, c_2)}\\
     \ket{\ell}_0 \ket{0}_1 \ket{j}_2 \ket{1}_3 \ket{b}_4 \ket{c_1}_5 \ket{c_2}_6 \ket{0}_7 \ket{0}_8 \ket{0}_9 \ket{0}_{10} & \text{otherwise}
    \end{cases} \\
    \xrightarrow{\text{4}} &  \begin{cases} \ket{\ell}_0 \ket{0}_1 \ket{j}_2 \ket{1}_3 \ket{b}_4 \ket{c_1}_5 \ket{c_2}_6 \ket{1}_7 \ket{f_{ind,TD} (\ell,0,j,c_1)}_8 \ket{[2^B |\eta_{j,f_{ind,TD}(\ell,0,j,c_1)} |]_B}_9 \ket{0}_{10} & \ell = 0 \text{ or } r, j \in G_\ell^{(c_1, c_2)}\\
     \ket{\ell}_0 \ket{0}_1 \ket{j}_2 \ket{1}_3 \ket{b}_4 \ket{c_1}_5 \ket{c_2}_6 \ket{0}_7 \ket{0}_8 \ket{0}_9 \ket{0}_{10} & \text{otherwise}
    \end{cases} \\
    \xrightarrow{\text{5-6}} &  \begin{cases} \ket{\ell}_0 \ket{0}_1 \ket{j}_2 \ket{1}_3 \ket{b}_4 \ket{c_1}_5 \ket{c_2}_6 \ket{1}_7 \ket{f_{ind,TD} (\ell,0,j,c_1)}_8 \\ 
    \ket{[2^B |\eta_{j,f_{ind,TD}(\ell,0,j,c_1)} |]_B}_9 \ket{1}_{10} & \ell = 0 \text{ or } r, j \in G_\ell^{(c_1, c_2)}, b \geq [2^B |\eta_{j,f_{ind,TD}(\ell,0,j,c_1)} |]_B\\
    \ket{\ell}_0 \ket{0}_1 \ket{j}_2 \ket{1}_3 \ket{b}_4 \ket{c_1}_5 \ket{c_2}_6 \ket{1}_7 \ket{f_{ind,TD} (\ell,0,j,c_1)}_8 \\
    \ket{[2^B |\eta_{j,f_{ind,TD}(\ell,0,j,c_1)} |]_B}_9 \ket{0}_{10} & \ell = 0 \text{ or } r, j \in G_\ell^{(c_1, c_2)}, b < [2^B |\eta_{j,f_{ind,TD}(\ell,0,j,c_1)} |]_B\\
     \ket{\ell}_0 \ket{0}_1 \ket{j}_2 \ket{1}_3 \ket{b}_4 \ket{c_1}_5 \ket{c_2}_6 \ket{0}_7 \ket{0}_8 \ket{0}_9 \ket{1}_{10} & \text{otherwise}
    \end{cases} \\
    \xrightarrow{\text{7}} &  \begin{cases} e^{-i \frac{T}{r} \left( \frac{1}{2} \lambda_j(0) + \frac{1}{2} \lambda_j (1) + \sum_{\ell'=0}^{r-1} \lambda_j \left( \frac{\ell'}{r} \right) \right)} \ket{\ell}_0 \ket{0}_1 \ket{j}_2 \ket{2}_3 \ket{b}_4 \ket{c_1}_5 \ket{c_2}_6 \ket{1}_7 \\
    \ket{f_{ind,TD} (\ell,0,j,c_1)}_8 \ket{[2^B |\eta_{j,f_{ind,TD}(\ell,0,j,c_1)} |]_B}_9 \ket{1}_{10} & \ell = 0 \text{ or } r, j \in G_\ell^{(c_1, c_2)},\\
    & b \geq [2^B |\eta_{j,f_{ind,TD}(\ell,0,j,c_1)} |]_B\\
    e^{-i \frac{T}{r} \left( \frac{1}{2} \lambda_j(0) + \frac{1}{2} \lambda_j (1) + \sum_{\ell'=0}^{r-1} \lambda_j \left( \frac{\ell'}{r} \right) \right)} \ket{\ell}_0 \ket{0}_1 \ket{j}_2 \ket{2}_3 \ket{b}_4 \ket{c_1}_5 \ket{c_2}_6 \ket{1}_7 \\
    \ket{f_{ind,TD} (\ell,0,j,c_1)}_8 \ket{[2^B |\zeta_{j,f_{ind,TD}(\ell,0,j,c_1)} |]_B}_9 \ket{0}_{10} & \ell = 0 \text{ or } r, j \in G_\ell^{(c_1, c_2)},\\
    & b < [2^B |\zeta_{j,f_{ind,TD}(\ell,0,j,c_1)} |]_B\\
     \ket{\ell}_0 \ket{0}_1 \ket{j}_2 \ket{2}_3 \ket{b}_4 \ket{c_1}_5 \ket{c_2}_6 \ket{0}_7 \ket{0}_8 \ket{0}_9 \ket{1}_{10} & \text{otherwise}
    \end{cases} \\
    \xrightarrow{8} & \begin{cases} e^{-i \frac{T}{r} \left( \frac{1}{2} \lambda_j(0) + \frac{1}{2} \lambda_j (1) + \sum_{\ell'=0}^{r-1} \lambda_j \left( \frac{\ell'}{r} \right) \right)} e^{i \arg (\eta_{j, f_{ind,TD}(\ell,0,j,c_1)})} \ket{\ell}_0 \ket{0}_1 \ket{j}_2 \ket{1}_3 \ket{b}_4 \\
     \ket{c_1}_5 \ket{c_2}_6 \ket{1}_7 \ket{f_{ind,TD} (\ell,0,j,c_1)}_8 \ket{[2^B |\eta_{j,f_{ind,TD}(\ell,0,j,c_1)} |]_B}_9 \ket{1}_{10} & \ell = 0 \text{ or } r, j \in G_\ell^{(c_1, c_2)},\\
    & b \geq [2^B |\eta_{j,f_{ind,TD}(\ell,0,j,c_1)} |]_B\\
    e^{-i \frac{T}{r} \left( \frac{1}{2} \lambda_j(0) + \frac{1}{2} \lambda_j (1) + \sum_{\ell'=0}^{r-1} \lambda_j \left( \frac{\ell'}{r} \right) \right)} e^{i \arg (\eta_{j, f_{ind,TD}(\ell,r,j,c_1)})} \ket{\ell}_0 \ket{0}_1 \ket{j}_2 \ket{1}_3 \ket{b}_4 \\
    \ket{c_1}_5 \ket{c_2}_6 \ket{1}_7 \ket{f_{ind,TD} (\ell,0,j,c_1)}_8 \ket{[2^B |\eta_{j,f_{ind,TD}(\ell,0,j,c_1)} |]_B}_9 \ket{0}_{10} & \ell = 0 \text{ or } r, j \in G_\ell^{(c_1, c_2)},\\
    & b < [2^B |\eta_{j,f_{ind,TD}(\ell,0,j,c_1)} |]_B\\
     \ket{\ell}_0 \ket{0}_1 \ket{j}_2 \ket{1}_3 \ket{b}_4 \ket{c_1}_5 \ket{c_2}_6 \ket{0}_7 \ket{0}_8 \ket{0}_9 \ket{1}_{10} & \text{otherwise}
    \end{cases} \\
    \xrightarrow{9} & \begin{cases} e^{-i \frac{T}{r} \left( \frac{1}{2} \lambda_j(0) + \frac{1}{2} \lambda_j (1) + \sum_{\ell'=0}^{r-1} \lambda_j \left( \frac{\ell'}{r} \right) \right)} e^{i \arg (\eta_{j, f_{ind,TD}(\ell,0,j,c_1)})}  \\
     e^{-i \frac{\ell T}{r^2} \left( \frac{1}{2} \gamma_{j, f_{ind,TD} (r,0,j,c_1)} (0) + \frac{1}{2} \gamma_{j, f_{ind,TD} (r,0,j,c_1)} (1) + \sum_{\ell'=0}^{r-1} \gamma_{j, f_{ind,TD} (r,0,j,c_1)} \left( \frac{\ell'}{r} \right) \right)} \\
     \ket{\ell}_0 \ket{0}_1 \ket{j}_2 \ket{1}_3 \ket{b}_4 \ket{c_1}_5 \ket{c_2}_6 \ket{1}_7 \ket{f_{ind,TD} (\ell,0,j,c_1)}_8 \\
     \ket{[2^B |\eta_{j,f_{ind,TD}(\ell,0,j,c_1)} |]_B}_9 \ket{1}_{10} & \ell = 0 \text{ or } r, j \in G_\ell^{(c_1, c_2)},\\
    & b \geq [2^B |\eta_{j,f_{ind,TD}(\ell,0,j,c_1)} |]_B\\
    e^{-i \frac{T}{r} \left( \frac{1}{2} \lambda_j(0) + \frac{1}{2} \lambda_j (1) + \sum_{\ell'=0}^{r-1} \lambda_j \left( \frac{\ell'}{r} \right) \right)} e^{i \arg (\eta_{j, f_{ind,TD}(\ell,0,j,c_1)})} \\
    e^{-i \frac{\ell T}{r^2} \left( \frac{1}{2} \gamma_{j, f_{ind,TD} (r,0,j,c_1)} (0) + \frac{1}{2} \gamma_{j, f_{ind,TD} (r,0,j,c_1)} (1) + \sum_{\ell'=0}^{r-1} \gamma_{j, f_{ind,TD} (r,0,j,c_1)} \left( \frac{\ell'}{r} \right) \right)} \\ 
    \ket{\ell}_0 \ket{0}_1 \ket{j}_2 \ket{1}_3 \ket{b}_4 \ket{c_1}_5 \ket{c_2}_6 \ket{1}_7 \ket{f_{ind,TD} (\ell,0,j,c_1)}_8 \\
    \ket{[2^B |\eta_{j,f_{ind,TD}(\ell,0,j,c_1)} |]_B}_9 \ket{0}_{10} & \ell = 0 \text{ or } r, j \in G_\ell^{(c_1, c_2)},\\
    & b < [2^B |\eta_{j,f_{ind,TD}(\ell,0,j,c_1)} |]_B\\
     \ket{\ell}_0 \ket{0}_1 \ket{j}_2 \ket{1}_3 \ket{b}_4 \ket{c_1}_5 \ket{c_2}_6 \ket{0}_7 \ket{0}_8 \ket{0}_9 \ket{1}_{10} & \text{otherwise}
    \end{cases} \\
    \xrightarrow{10-11} & \begin{cases} e^{-i \frac{T}{r} \left( \frac{1}{2} \lambda_j(0) + \frac{1}{2} \lambda_j (1) + \sum_{\ell'=0}^{r-1} \lambda_j \left( \frac{\ell'}{r} \right) \right)} e^{i \arg (\eta_{j, f_{ind,TD}(\ell,0,j,c_1)})} \\
    e^{-i \frac{\ell T}{r^2} \left( \frac{1}{2} \gamma_{j, f_{ind,TD} (r,0,j,c_1)} (0) + \frac{1}{2} \gamma_{j, f_{ind,TD} (r,0,j,c_1)} (1) + \sum_{\ell'=0}^{r-1} \gamma_{j, f_{ind,TD} (r,0,j,c_1)} \left( \frac{\ell'}{r} \right) \right)} \\ \ket{\ell}_0 \ket{1}_1 \ket{f_{ind,TD} (\ell,0,j,c_1)}_2 \ket{1}_3 \ket{b}_4 \ket{c_1}_5 \ket{c_2}_6 \ket{1}_7 \ket{j}_8 \\
    \ket{[2^B |\eta_{j,f_{ind,TD}(\ell,0,j,c_1)} |]_B}_9 \ket{1}_{10} & \ell = 0 \text{ or } r, j \in G_\ell^{(c_1, c_2)},\\
    & b \geq [2^B |\eta_{j,f_{ind,TD}(\ell,0,j,c_1)} |]_B\\
    e^{-i \frac{T}{r} \left( \frac{1}{2} \lambda_j(0) + \frac{1}{2} \lambda_j (1) + \sum_{\ell'=0}^{r-1} \lambda_j \left( \frac{\ell'}{r} \right) \right)} e^{i \arg (\eta_{j, f_{ind,TD}(\ell,0,j,c_1)})} \\
    e^{-i \frac{\ell T}{r^2} \left( \frac{1}{2} \gamma_{j, f_{ind,TD} (r,0,j,c_1)} (0) + \frac{1}{2} \gamma_{j, f_{ind,TD} (r,0,j,c_1)} (1) + \sum_{\ell'=0}^{r-1} \gamma_{j, f_{ind,TD} (r,0,j,c_1)} \left( \frac{\ell'}{r} \right) \right)} \\    \ket{\ell}_0 \ket{1}_1 \ket{f_{ind,TD} (\ell,0,j,c_1)}_2 \ket{1}_3 \ket{b}_4 \ket{c_1}_5 \ket{c_2}_6 \ket{1}_7 \ket{j}_8 \\
    \ket{[2^B |\eta_{j,f_{ind,TD}(\ell,0,j,c_1)} |]_B}_9 \ket{0}_{10} & \ell = 0 \text{ or } r, j \in G_\ell^{(c_1, c_2)},\\
    & b < [2^B |\eta_{j,f_{ind,TD}(\ell,0,j,c_1)} |]_B\\
     \ket{\ell}_0 \ket{0}_1 \ket{j}_2 \ket{1}_3 \ket{b}_4 \ket{c_1}_5 \ket{c_2}_6 \ket{0}_7 \ket{0}_8 \ket{0}_9 \ket{1}_{10} & \text{otherwise}
    \end{cases} \\
    \xrightarrow{12} & \begin{cases} (-1)^b e^{-i \frac{T}{r} \left( \frac{1}{2} \lambda_j(0) + \frac{1}{2} \lambda_j (1) + \sum_{\ell'=0}^{r-1} \lambda_j \left( \frac{\ell'}{r} \right) \right)} e^{i \arg (\eta_{j, f_{ind,TD}(\ell,0,j,c_1)})} \\
    e^{-i \frac{\ell T}{r^2} \left( \frac{1}{2} \gamma_{j, f_{ind,TD} (r,0,j,c_1)} (0) + \frac{1}{2} \gamma_{j, f_{ind,TD} (r,0,j,c_1)} (1) + \sum_{\ell'=0}^{r-1} \gamma_{j, f_{ind,TD} (r,0,j,c_1)} \left( \frac{\ell'}{r} \right) \right)} \\ \ket{\ell}_0 \ket{1}_1 \ket{f_{ind,TD} (\ell,0,j,c_1)}_2 \ket{1}_3 \ket{b}_4 \ket{c_1}_5 \ket{c_2}_6 \ket{1}_7 \ket{j}_8 \\
    \ket{[2^B |\eta_{j,f_{ind,TD}(\ell,0,j,c_1)} |]_B}_9 \ket{1}_{10} & \ell = 0 \text{ or } r, j \in G_\ell^{(c_1, c_2)},\\
    & b \geq [2^B |\eta_{j,f_{ind,TD}(\ell,0,j,c_1)} |]_B\\
    e^{-i \frac{T}{r} \left( \frac{1}{2} \lambda_j(0) + \frac{1}{2} \lambda_j (1) + \sum_{\ell'=0}^{r-1} \lambda_j \left( \frac{\ell'}{r} \right) \right)} e^{i \arg (\eta_{j, f_{ind,TD}(\ell,0,j,c_1)})} \\
    e^{-i \frac{\ell T}{r^2} \left( \frac{1}{2} \gamma_{j, f_{ind,TD} (r,0,j,c_1)} (0) + \frac{1}{2} \gamma_{j, f_{ind,TD} (r,0,j,c_1)} (1) + \sum_{\ell'=0}^{r-1} \gamma_{j, f_{ind,TD} (r,0,j,c_1)} \left( \frac{\ell'}{r} \right) \right)} \\    \ket{\ell}_0 \ket{1}_1 \ket{f_{ind,TD} (\ell,0,j,c_1)}_2 \ket{1}_3 \ket{b}_4 \ket{c_1}_5 \ket{c_2}_6 \ket{1}_7 \ket{j}_8 \\
    \ket{[2^B |\eta_{j,f_{ind,TD}(\ell,0,j,c_1)} |]_B}_9 \ket{0}_{10} & \ell = 0 \text{ or } r, j \in G_\ell^{(c_1, c_2)},\\
    & b < [2^B |\eta_{j,f_{ind,TD}(\ell,0,j,c_1)} |]_B\\
    (-1)^b \ket{\ell}_0 \ket{0}_1 \ket{j}_2 \ket{1}_3 \ket{b}_4 \ket{c_1}_5 \ket{c_2}_6 \ket{0}_7 \ket{0}_8 \ket{0}_9 \ket{1}_{10} & \text{otherwise}
    \end{cases} \\
    \xrightarrow{13} & \begin{cases} g_1(\ell,j,c_1,b) e^{-i \frac{T}{r} \left( \frac{1}{2} \lambda_j(0) + \frac{1}{2} \lambda_j (1) + \sum_{\ell'=0}^{r-1} \lambda_j \left( \frac{\ell'}{r} \right) \right)} e^{i \arg (\eta_{j, f_{ind,TD}(\ell,0,j,c_1)})} \\
    e^{-i \frac{\ell T}{r^2} \left( \frac{1}{2} \gamma_{j, f_{ind,TD} (r,0,j,c_1)} (0) + \frac{1}{2} \gamma_{j, f_{ind,TD} (r,0,j,c_1)} (1) + \sum_{\ell'=0}^{r-1} \gamma_{j, f_{ind,TD} (r,0,j,c_1)} \left( \frac{\ell'}{r} \right) \right)} \\\ket{\ell}_0 \ket{1}_1 \ket{f_{ind,TD} (\ell,0,j,c_1)}_2 \ket{1}_3 \ket{b}_4 \ket{c_1}_5 \ket{c_2}_6 \ket{1}_7 \ket{j}_8 \\
    \ket{[2^B |\eta_{j,f_{ind,TD}(\ell,0,j,c_1)} |]_B}_9 \ket{0}_{10} & \ell = 0 \text{ or } r, j \in G_\ell^{(c_1, c_2)}\\
    (-1)^b \ket{\ell}_0 \ket{0}_1 \ket{j}_2 \ket{1}_3 \ket{b}_4 \ket{c_1}_5 \ket{c_2}_6 \ket{0}_7 \ket{0}_8 \ket{0}_9 \ket{0}_{10} & \text{otherwise}
    \end{cases} \\
    \xrightarrow{14} & \begin{cases} g_1(\ell,j,c_1,b) e^{-i \frac{T}{r} \left( \frac{1}{2} \lambda_j(0) + \frac{1}{2} \lambda_j (1) + \sum_{\ell'=0}^{r-1} \lambda_j \left( \frac{\ell'}{r} \right) \right)} e^{i \arg (\eta_{j, f_{ind,TD}(\ell,0,j,c_1)})} \\
    e^{-i \frac{\ell T}{r^2} \left( \frac{1}{2} \gamma_{j, f_{ind,TD} (r,0,j,c_1)} (0) + \frac{1}{2} \gamma_{j, f_{ind,TD} (r,0,j,c_1)} (1) + \sum_{\ell'=0}^{r-1} \gamma_{j, f_{ind,TD} (r,0,j,c_1)} \left( \frac{\ell'}{r} \right) \right)} \\
    \ket{\ell}_0 \ket{1}_1 \ket{f_{ind,TD} (\ell,0,j,c_1)}_2 \ket{1}_3 \ket{b}_4 \ket{c_1}_5 \ket{c_2}_6 \ket{1}_7 \ket{j}_8 \ket{0}_9 \ket{0}_{10} & \ell = 0 \text{ or } r, j \in G_\ell^{(c_1, c_2)}\\
    (-1)^b \ket{\ell}_0 \ket{0}_1 \ket{j}_2 \ket{1}_3 \ket{b}_4 \ket{c_1}_5 \ket{c_2}_6 \ket{0}_7 \ket{0}_8 \ket{0}_9 \ket{0}_{10} & \text{otherwise}
    \end{cases} \\
    \xrightarrow{15-16} & \begin{cases} g_1(\ell,j,c_1,b) e^{-i \frac{T}{r} \left( \frac{1}{2} \lambda_j(0) + \frac{1}{2} \lambda_j (1) + \sum_{\ell'=0}^{r-1} \lambda_j \left( \frac{\ell'}{r} \right) \right)} e^{i \arg (\eta_{j, f_{ind,TD}(\ell,0,j,c_1)})} \\
    e^{-i \frac{\ell T}{r^2} \left( \frac{1}{2} \gamma_{j, f_{ind,TD} (r,0,j,c_1)} (0) + \frac{1}{2} \gamma_{j, f_{ind,TD} (r,0,j,c_1)} (1) + \sum_{\ell'=0}^{r-1} \gamma_{j, f_{ind,TD} (r,0,j,c_1)} \left( \frac{\ell'}{r} \right) \right)} \\
    \ket{\ell}_0 \ket{1}_1 \ket{f_{ind,TD} (\ell,0,j,c_1)}_2 \ket{1}_3 \ket{b}_4 \ket{c_1}_5 \ket{c_2}_6 \ket{1}_7 \ket{0}_8 \ket{0}_9 \ket{0}_{10} & \ell = 0 \text{ or } r, j \in G_\ell^{(c_1, c_2)}\\
    (-1)^b \ket{\ell}_0 \ket{0}_1 \ket{j}_2 \ket{1}_3 \ket{b}_4 \ket{c_1}_5 \ket{c_2}_6 \ket{0}_7 \ket{0}_8 \ket{0}_9 \ket{0}_{10} & \text{otherwise}
    \end{cases} \\
    \xrightarrow{17-18} & \begin{cases} g_1(\ell,j,c_1,b) e^{-i \frac{T}{r} \left( \frac{1}{2} \lambda_j(0) + \frac{1}{2} \lambda_j (1) + \sum_{\ell'=0}^{r-1} \lambda_j \left( \frac{\ell'}{r} \right) \right)} e^{i \arg (\eta_{j, f_{ind,TD}(\ell,0,j,c_1)})} \\
    e^{-i \frac{\ell T}{r^2} \left( \frac{1}{2} \gamma_{j, f_{ind,TD} (r,0,j,c_1)} (0) + \frac{1}{2} \gamma_{j, f_{ind,TD} (r,0,j,c_1)} (1) + \sum_{\ell'=0}^{r-1} \gamma_{j, f_{ind,TD} (r,0,j,c_1)} \left( \frac{\ell'}{r} \right) \right)} \\
    \ket{\ell}_0 \ket{1}_1 \ket{f_{ind,TD} (\ell,0,j,c_1)}_2 \ket{1}_3 \ket{b}_4 \ket{c_1}_5 \ket{c_2}_6 \ket{0}_7 \ket{0}_8 \ket{0}_9 \ket{0}_{10} & \ell = 0 \text{ or } r, j \in G_\ell^{(c_1, c_2)}\\
    (-1)^b \ket{\ell}_0 \ket{0}_1 \ket{j}_2 \ket{1}_3 \ket{b}_4 \ket{c_1}_5 \ket{c_2}_6 \ket{0}_7 \ket{0}_8 \ket{0}_9 \ket{0}_{10} & \text{otherwise}
    \end{cases} \\
    \xrightarrow{19} & \begin{cases} g_1(\ell,j,c_1,b) e^{-i \frac{T}{r} \left( \frac{1}{2} \lambda_j(0) + \frac{1}{2} \lambda_j (1) + \sum_{\ell'=0}^{r-1} \lambda_j \left( \frac{\ell'}{r} \right) \right)} e^{i \arg (\eta_{j, f_{ind,TD}(\ell,0,j,c_1)})} \\
    e^{-i \frac{\ell T}{r^2} \left( \frac{1}{2} \gamma_{j, f_{ind,TD} (r,0,j,c_1)} (0) + \frac{1}{2} \gamma_{j, f_{ind,TD} (r,0,j,c_1)} (1) + \sum_{\ell'=0}^{r-1} \gamma_{j, f_{ind,TD} (r,0,j,c_1)} \left( \frac{\ell'}{r} \right) \right)} \\
    \ket{\ell}_0 \ket{0}_1 \ket{f_{ind,TD} (\ell,0,j,c_1)}_2 \ket{1}_3 \ket{b}_4 \ket{c_1}_5 \ket{c_2}_6 \ket{0}_7 \ket{0}_8 \ket{0}_9 \ket{0}_{10} & \ell = 0 \text{ or } r, j \in G_\ell^{(c_1, c_2)}\\
    (-1)^b \ket{\ell}_0 \ket{1}_1 \ket{j}_2 \ket{1}_3 \ket{b}_4 \ket{c_1}_5 \ket{c_2}_6 \ket{0}_7 \ket{0}_8 \ket{0}_9 \ket{0}_{10} & \text{otherwise}
    \end{cases}
\end{align}

$q=2$:
\begin{align}
    & \text{Input: } \ket{\ell}_0 \ket{0}_1 \ket{j}_2 \ket{2}_3 \ket{b}_4 \ket{c_1}_5 \ket{c_2}_6 \ket{0}_7 \ket{0}_8 \ket{0}_9 \ket{0}_{10}\\
    \xrightarrow{1} &  \begin{cases} \ket{\ell}_0 \ket{0}_1 \ket{j}_2 \ket{2}_3 \ket{b}_4 \ket{c_1}_5 \ket{c_2}_6 \ket{1}_7 \ket{0}_8 \ket{0}_9 \ket{0}_{10} & j \in G_\ell^{(c_1, c_2)}\\
     \ket{\ell}_0 \ket{0}_1 \ket{j}_2 \ket{2}_3 \ket{b}_4 \ket{c_1}_5 \ket{c_2}_6 \ket{0}_7 \ket{0}_8 \ket{0}_9 \ket{0}_{10} & \text{otherwise}
    \end{cases} \\
    \xrightarrow{2-3} &  \begin{cases} \ket{\ell}_0 \ket{0}_1 \ket{j}_2 \ket{2}_3 \ket{b}_4 \ket{c_1}_5 \ket{c_2}_6 \ket{1}_7 \ket{0}_8 \ket{0}_9 \ket{0}_{10} & j \in G_\ell^{(c_1, c_2)}\\
     \ket{\ell}_0 \ket{0}_1 \ket{j}_2 \ket{2}_3 \ket{b}_4 \ket{c_1}_5 \ket{c_2}_6 \ket{0}_7 \ket{0}_8 \ket{0}_9 \ket{0}_{10} & \text{otherwise}
    \end{cases} \\
    \xrightarrow{\text{4-5}} &  \begin{cases} \ket{\ell}_0 \ket{0}_1 \ket{j}_2 \ket{2}_3 \ket{b}_4 \ket{c_1}_5 \ket{c_2}_6 \ket{1}_7 \ket{f_{ind,TD} (\ell,0,j,c_1)}_8 \ket{[2^B |\zeta_{j,f_{ind,TD}(\ell,0,j,c_1)} |]_B}_9 \ket{0}_{10} & j \in G_\ell^{(c_1, c_2)}\\
     \ket{\ell}_0 \ket{0}_1 \ket{j}_2 \ket{2}_3 \ket{b}_4 \ket{c_1}_5 \ket{c_2}_6 \ket{0}_7 \ket{0}_8 \ket{0}_9 \ket{0}_{10} & \text{otherwise}
    \end{cases} \\
    \xrightarrow{\text{6}} &  \begin{cases} \ket{\ell}_0 \ket{0}_1 \ket{j}_2 \ket{2}_3 \ket{b}_4 \ket{c_1}_5 \ket{c_2}_6 \ket{1}_7 \ket{f_{ind,TD} (\ell,0,j,c_1)}_8 \\ 
    \ket{[2^B |\zeta_{j,f_{ind,TD}(\ell,0,j,c_1)} |]_B}_9 \ket{1}_{10} & j \in G_\ell^{(c_1, c_2)}, b \geq [2^B |\zeta_{j,f_{ind,TD}(\ell,0,j,c_1)} |]_B\\
    \ket{\ell}_0 \ket{0}_1 \ket{j}_2 \ket{2}_3 \ket{b}_4 \ket{c_1}_5 \ket{c_2}_6 \ket{1}_7 \ket{f_{ind,TD} (\ell,0,j,c_1)}_8 \\
    \ket{[2^B |\zeta_{j,f_{ind,TD}(\ell,0,j,c_1)} |]_B}_9 \ket{0}_{10} & j \in G_\ell^{(c_1, c_2)}, b < [2^B |\zeta_{j,f_{ind,TD}(\ell,0,j,c_1)} |]_B\\
     \ket{\ell}_0 \ket{0}_1 \ket{j}_2 \ket{2}_3 \ket{b}_4 \ket{c_1}_5 \ket{c_2}_6 \ket{0}_7 \ket{0}_8 \ket{0}_9 \ket{1}_{10} & \text{otherwise}
    \end{cases} \\
    \xrightarrow{\text{7}} &  \begin{cases} e^{-i \frac{T}{r} \left( \frac{1}{2} \lambda_j(0) + \frac{1}{2} \lambda_j (1) + \sum_{\ell'=0}^{r-1} \lambda_j \left( \frac{\ell'}{r} \right) \right)} \ket{\ell}_0 \ket{0}_1 \ket{j}_2 \ket{2}_3 \ket{b}_4 \ket{c_1}_5 \ket{c_2}_6 \ket{1}_7 \\
    \ket{f_{ind,TD} (\ell,0,j,c_1)}_8 \ket{[2^B |\zeta_{j,f_{ind,TD}(\ell,0,j,c_1)} |]_B}_9 \ket{1}_{10} & j \in G_\ell^{(c_1, c_2)},\\
    & b \geq [2^B |\zeta_{j,f_{ind,TD}(\ell,0,j,c_1)} |]_B\\
    e^{-i \frac{T}{r} \left( \frac{1}{2} \lambda_j(0) + \frac{1}{2} \lambda_j (1) + \sum_{\ell'=0}^{r-1} \lambda_j \left( \frac{\ell'}{r} \right) \right)} \ket{\ell}_0 \ket{0}_1 \ket{j}_2 \ket{2}_3 \ket{b}_4 \ket{c_1}_5 \ket{c_2}_6 \ket{1}_7 \\
    \ket{f_{ind,TD} (\ell,0,j,c_1)}_8 \ket{[2^B |\zeta_{j,f_{ind,TD}(\ell,0,j,c_1)} |]_B}_9 \ket{0}_{10} & j \in G_\ell^{(c_1, c_2)},\\
    & b < [2^B |\zeta_{j,f_{ind,TD}(\ell,0,j,c_1)} |]_B\\
     \ket{\ell}_0 \ket{0}_1 \ket{j}_2 \ket{2}_3 \ket{b}_4 \ket{c_1}_5 \ket{c_2}_6 \ket{0}_7 \ket{0}_8 \ket{0}_9 \ket{1}_{10} & \text{otherwise}
    \end{cases} \\
    \xrightarrow{8-10} & \begin{cases} e^{-i \frac{T}{r} \left( \frac{1}{2} \lambda_j(0) + \frac{1}{2} \lambda_j (1) + \sum_{\ell'=0}^{r-1} \lambda_j \left( \frac{\ell'}{r} \right) \right)} e^{i \arg (\zeta_{j, f_{ind,TD}(\ell,0,j,c_1)})} \ket{\ell}_0 \ket{0}_1 \ket{j}_2 \ket{2}_3 \ket{b}_4 \\
     \ket{c_1}_5 \ket{c_2}_6 \ket{1}_7 \ket{f_{ind,TD} (\ell,0,j,c_1)}_8 \ket{[2^B |\zeta_{j,f_{ind,TD}(\ell,0,j,c_1)} |]_B}_9 \ket{1}_{10} & j \in G_\ell^{(c_1, c_2)},\\
    & b \geq [2^B |\zeta_{j,f_{ind,TD}(\ell,0,j,c_1)} |]_B\\
    e^{-i \frac{T}{r} \left( \frac{1}{2} \lambda_j(0) + \frac{1}{2} \lambda_j (1) + \sum_{\ell'=0}^{r-1} \lambda_j \left( \frac{\ell'}{r} \right) \right)} e^{i \arg (\zeta_{j, f_{ind,TD}(\ell,r,j,c_1)})} \ket{\ell}_0 \ket{0}_1 \ket{j}_2 \ket{2}_3 \ket{b}_4 \\
    \ket{c_1}_5 \ket{c_2}_6 \ket{1}_7 \ket{f_{ind,TD} (\ell,0,j,c_1)}_8 \ket{[2^B |\zeta_{j,f_{ind,TD}(\ell,0,j,c_1)} |]_B}_9 \ket{0}_{10} & j \in G_\ell^{(c_1, c_2)},\\
    & b < [2^B |\zeta_{j,f_{ind,TD}(\ell,0,j,c_1)} |]_B\\
     \ket{\ell}_0 \ket{0}_1 \ket{j}_2 \ket{2}_3 \ket{b}_4 \ket{c_1}_5 \ket{c_2}_6 \ket{0}_7 \ket{0}_8 \ket{0}_9 \ket{1}_{10} & \text{otherwise}
    \end{cases} \\
    \xrightarrow{11-12} & \begin{cases} (-1)^b e^{-i \frac{T}{r} \left( \frac{1}{2} \lambda_j(0) + \frac{1}{2} \lambda_j (1) + \sum_{\ell'=0}^{r-1} \lambda_j \left( \frac{\ell'}{r} \right) \right)} e^{i \arg (\zeta_{j, f_{ind,TD}(\ell,0,j,c_1)})} \ket{\ell}_0 \ket{0}_1 \ket{j}_2 \ket{2}_3 \ket{b}_4 \\
     \ket{c_1}_5 \ket{c_2}_6 \ket{1}_7 \ket{f_{ind,TD} (\ell,0,j,c_1)}_8 \ket{[2^B |\zeta_{j,f_{ind,TD}(\ell,0,j,c_1)} |]_B}_9 \ket{1}_{10} & j \in G_\ell^{(c_1, c_2)},\\
    & b \geq [2^B |\zeta_{j,f_{ind,TD}(\ell,0,j,c_1)} |]_B\\
    e^{-i \frac{T}{r} \left( \frac{1}{2} \lambda_j(0) + \frac{1}{2} \lambda_j (1) + \sum_{\ell'=0}^{r-1} \lambda_j \left( \frac{\ell'}{r} \right) \right)} e^{i \arg (\zeta_{j, f_{ind,TD}(\ell,r,j,c_1)})} \ket{\ell}_0 \ket{0}_1 \ket{j}_2 \ket{2}_3 \ket{b}_4 \\
    \ket{c_1}_5 \ket{c_2}_6 \ket{1}_7 \ket{f_{ind,TD} (\ell,0,j,c_1)}_8 \ket{[2^B |\zeta_{j,f_{ind,TD}(\ell,0,j,c_1)} |]_B}_9 \ket{0}_{10} & j \in G_\ell^{(c_1, c_2)},\\
    & b < [2^B |\zeta_{j,f_{ind,TD}(\ell,0,j,c_1)} |]_B\\
     (-1)^b \ket{\ell}_0 \ket{0}_1 \ket{j}_2 \ket{2}_3 \ket{b}_4 \ket{c_1}_5 \ket{c_2}_6 \ket{0}_7 \ket{0}_8 \ket{0}_9 \ket{1}_{10} & \text{otherwise}
    \end{cases} \\
    \xrightarrow{13} & \begin{cases} g_2 (\ell, j, c_1, b) e^{-i \frac{T}{r} \left( \frac{1}{2} \lambda_j(0) + \frac{1}{2} \lambda_j (1) + \sum_{\ell'=0}^{r-1} \lambda_j \left( \frac{\ell'}{r} \right) \right)} e^{i \arg (\zeta_{j, f_{ind,TD}(\ell,0,j,c_1)})} \ket{\ell}_0 \ket{0}_1 \ket{j}_2 \ket{2}_3 \ket{b}_4 \\
     \ket{c_1}_5 \ket{c_2}_6 \ket{1}_7 \ket{f_{ind,TD} (\ell,0,j,c_1)}_8 \ket{[2^B |\zeta_{j,f_{ind,TD}(\ell,0,j,c_1)} |]_B}_9 \ket{0}_{10} & j \in G_\ell^{(c_1, c_2)} \\
     (-1)^b \ket{\ell}_0 \ket{0}_1 \ket{j}_2 \ket{2}_3 \ket{b}_4 \ket{c_1}_5 \ket{c_2}_6 \ket{0}_7 \ket{0}_8 \ket{0}_9 \ket{0}_{10} & \text{otherwise}
    \end{cases} \\
    \xrightarrow{14-15} & \begin{cases} g_2 (\ell, j, c_1, b) e^{-i \frac{T}{r} \left( \frac{1}{2} \lambda_j(0) + \frac{1}{2} \lambda_j (1) + \sum_{\ell'=0}^{r-1} \lambda_j \left( \frac{\ell'}{r} \right) \right)} e^{i \arg (\zeta_{j, f_{ind,TD}(\ell,0,j,c_1)})} \ket{\ell}_0 \ket{0}_1 \ket{j}_2 \ket{2}_3 \ket{b}_4 \\
     \ket{c_1}_5 \ket{c_2}_6 \ket{1}_7 \ket{f_{ind,TD} (\ell,0,j,c_1)}_8 \ket{0}_9 \ket{0}_{10} & j \in G_\ell^{(c_1, c_2)} \\
     (-1)^b \ket{\ell}_0 \ket{0}_1 \ket{j}_2 \ket{2}_3 \ket{b}_4 \ket{c_1}_5 \ket{c_2}_6 \ket{0}_7 \ket{0}_8 \ket{0}_9 \ket{0}_{10} & \text{otherwise}
    \end{cases} \\
    \xrightarrow{16-17}& \begin{cases} g_2 (\ell, j, c_1, b) e^{-i \frac{T}{r} \left( \frac{1}{2} \lambda_j(0) + \frac{1}{2} \lambda_j (1) + \sum_{\ell'=0}^{r-1} \lambda_j \left( \frac{\ell'}{r} \right) \right)} e^{i \arg (\zeta_{j, f_{ind,TD}(\ell,0,j,c_1)})} \ket{\ell}_0 \ket{0}_1 \ket{j}_2 \ket{2}_3 \ket{b}_4 \\
     \ket{c_1}_5 \ket{c_2}_6 \ket{1}_7 \ket{0}_8 \ket{0}_9 \ket{0}_{10} & j \in G_\ell^{(c_1, c_2)} \\
     (-1)^b \ket{\ell}_0 \ket{0}_1 \ket{j}_2 \ket{2}_3 \ket{b}_4 \ket{c_1}_5 \ket{c_2}_6 \ket{0}_7 \ket{0}_8 \ket{0}_9 \ket{0}_{10} & \text{otherwise}
    \end{cases} \\
    \xrightarrow{18} & \begin{cases} g_2 (\ell, j, c_1, b) e^{-i \frac{T}{r} \left( \frac{1}{2} \lambda_j(0) + \frac{1}{2} \lambda_j (1) + \sum_{\ell'=0}^{r-1} \lambda_j \left( \frac{\ell'}{r} \right) \right)} e^{i \arg (\zeta_{j, f_{ind,TD}(\ell,0,j,c_1)})} \ket{\ell}_0 \ket{0}_1 \ket{j}_2 \ket{2}_3 \ket{b}_4 \\
     \ket{c_1}_5 \ket{c_2}_6 \ket{0}_7 \ket{0}_8 \ket{0}_9 \ket{0}_{10} & j \in G_\ell^{(c_1, c_2)} \\
     (-1)^b \ket{\ell}_0 \ket{0}_1 \ket{j}_2 \ket{2}_3 \ket{b}_4 \ket{c_1}_5 \ket{c_2}_6 \ket{0}_7 \ket{0}_8 \ket{0}_9 \ket{0}_{10} & \text{otherwise}
    \end{cases}
\end{align}

Thus, the algorithm described in \ref{alg:sel2} implements the Select operation defined in Def. \ref{def:sel2} on the input state $\ket{\ell}_0 \ket{0}_1 \ket{j}_2 \ket{p}_3 \ket{b}_4 \ket{c_1}_5 \ket{c_2}_6 \ket{0}_7 \ket{0}_8 \ket{0}_0 \ket{0}_{10}$. 

\begin{claim} \label{clm:lt_clm1}
    Let
    \begin{align}
        \widetilde{V} = \sum_{\ell=0}^r \sum_{b=0}^{2^B-1} \sum_{c_1=0}^{d-1} \sum_{c_2=0}^{d-1} \frac{1}{2^B} (\hat{X} \otimes \normalfont{\openone} ) \left( \frac{1}{(r+1)d^2} \widetilde{V}_{\ell, 0, b, c_1, c_2} + \widetilde{V}_{\ell, 1, b, c_1, c_2} + \widetilde{V}_{\ell, 2, b, c_1, c_2}  \right)
    \end{align}
Then,
    \begin{align}
        \| (\bra{0} \otimes \normalfont{\openone}) \widetilde{V} (\ket{0} \otimes \normalfont{\openone}) - \bar{V} \| \in O \left( \frac{d^2}{2^B} \right)
    \end{align}
\end{claim}
\begin{proof}
The claim above is equivalent to Claim \ref{clm:ast} in the short-time Hamiltonian path integral algorithm, and can be proven in a very similar manner. We will highlight the differences in this proof, without repeating the parts that are effectively identical. Using the triangle inequality, we can break the error into the sum of errors from implementing each $\bar{V}_i$ term:
\begin{align}
    & \| (\bra{0} \otimes \normalfont{\openone}) \widetilde{V} (\ket{0} \otimes \normalfont{\openone}) - \bar{V} \| \\
    & \quad = \| (\bra{0} \otimes \normalfont{\openone}) \left( \sum_{\ell=0}^r \sum_{b=0}^{2^B-1} \sum_{c_1=0}^{d-1} \sum_{c_2=0}^{d-1} \frac{1}{2^B} (\hat{X} \otimes \normalfont{\openone} ) \left( \widetilde{V}_{\ell, 0, b, c_1, c_2} + \widetilde{V}_{\ell, 1, b, c_1, c_2} + \widetilde{V}_{\ell, 2, b, c_1, c_2} \right) \right) (\ket{0} \otimes \normalfont{\openone}) - (\bar{V}_0 + \bar{V}_1 + \bar{V}_2) \| \\
    & \quad \leq \| (\bra{0} \otimes \normalfont{\openone}) \left( \sum_{\ell=0}^r \sum_{b=0}^{2^B-1} \sum_{c_1=0}^{d-1} \sum_{c_2=0}^{d-1} \frac{1}{2^B} (\hat{X} \otimes \normalfont{\openone} ) \widetilde{V}_{\ell, 0, b, c_1, c_2} \right) (\ket{0} \otimes \normalfont{\openone}) - \bar{V}_0 \| \notag \\
    & \quad +  \| (\bra{0} \otimes \normalfont{\openone}) \left( \sum_{\ell=0}^r \sum_{b=0}^{2^B-1} \sum_{c_1=0}^{d-1} \sum_{c_2=0}^{d-1} \frac{1}{2^B} (\hat{X} \otimes \normalfont{\openone} ) \widetilde{V}_{\ell, 1, b, c_1, c_2} \right) (\ket{0} \otimes \normalfont{\openone}) - \bar{V}_1 \| \notag \\
    & \quad +  \| (\bra{0} \otimes \normalfont{\openone}) \left( \sum_{\ell=0}^r \sum_{b=0}^{2^B-1} \sum_{c_1=0}^{d-1} \sum_{c_2=0}^{d-1} \frac{1}{2^B} (\hat{X} \otimes \normalfont{\openone} ) \widetilde{V}_{\ell, 2, b, c_1, c_2} \right) (\ket{0} \otimes \normalfont{\openone}) - \bar{V}_2 \| \label{LT_error1}
\end{align}
$(\bra{0} \otimes \normalfont{\openone}) \left( \sum_{\ell,b,c_1,c_2} \frac{1}{2^B} (\hat{X} \otimes \normalfont{\openone} ) \widetilde{V}_{\ell, 0, b, c_1, c_2} \right) (\ket{0} \otimes \normalfont{\openone}) = \bar{V}_0$ exactly, so the first term is zero. As we saw in Section \ref{sec:ast_full_proof}, the error incurred from implementing the amplitudes using the bitstring encoding and alternating sign trick is $O \left( \frac{d^2}{2^B} \right)$. Thus, the error from the $\bar{V}_1$ and $\bar{V}_2$ terms (the second term and third terms in Eq. \eqref{LT_error1}) is $O \left( \frac{d^2}{2^B} \right)$. Thus, the entire expression in Eq. \eqref{LT_error1} is $O \left( \frac{d^2}{2^B} \right)$.
\end{proof}
\begin{claim}
    Let
    \begin{align}
        \hat{W} = \mathrm{PREP}^\dagger (r_0, r_3, r_4, r_5, r_6) \; \mathrm{SEL} (r_0, r_1, r_3, r_4, r_5, r_6) \; \mathrm{PREP} (r_0, r_3, r_4, r_5, r_6).
    \end{align}
    Let $\hat{\Pi} := (\ketbra{0}{0})^{\otimes (\log (r+1) + 2 + B + 2 \log d)}$ and $\hat{R} := - (\normalfont{\openone} - 2 \hat{W} \hat{\Pi} \hat{W}^\dagger)$. Then,
    \begin{align}
        & \Bigg\| \left( \hat{\Pi} \hat{R}^p \hat{W} \hat{\Pi} \right) \ket{0}^{\otimes (\log (r+1) + 2 + B + 2 \log d)} \ket{\psi} \notag \\
        & \quad - ( \ketbra{0}{0} )^{\otimes (\log (r+1) + 2 + B + 2 \log d)} \otimes \left( \hat{B}_2^\dagger \mathcal{T} e^{-i T \int_0^1 \hat{H}(s) \mathrm{d} s} \hat{B}_1 \right) \ket{0}^{\otimes (\log (r+1) + 2 + B + 2 \log d)} \ket{\psi} \Bigg\| \nonumber \\
        & \quad \in O \left( \frac{\Gamma^4}{\gamma_{\min}^2 T^2} \right),
    \end{align}
    where $\hat{B}_1$ and $\hat{B}_2$ are change of basis operators from the computational basis to the $\hat{H}(0)$ and to the $\hat{H}(T)$ bases, respectively.
\end{claim}

\begin{proof}
    From Eq. \eqref{eq:lt_error2}, we have
    \begin{align}
        \left\| \mathcal{T} e^{-i T \int_0^1 \hat{H}(s) \mathrm{d} s} - \bar{U} \right\| \in O \left( \frac{\Gamma^4}{\gamma^2_{\rm min} T^2} + \frac{T}{r^3} \Lambda\right), \label{eq:lt_clm_proof1}
    \end{align}
    where
    \begin{align}
        \Lambda := \max_j ( \max | \lambda_{j}''(s)| ),
    \end{align}
    so
    \begin{align}
        \left\| \left( \hat{B}_2^\dagger \mathcal{T} e^{-i T \int_0^1 \hat{H}(s) \mathrm{d} s} \hat{B}_1 \right)  - \bar{V} \right\| \in O \left( \frac{\Gamma^4}{\gamma^2_{\rm min} T^2} + \frac{T}{r^3} \Lambda \right), \label{eq:lt_clm_proof2}
    \end{align}
    since \eqref{eq:lt_clm_proof2} is just \eqref{eq:lt_clm_proof1} in a different basis. By Claim \ref{clm:lt_clm1}, the rounding error from approximating the amplitudes in $B$ bits is $O \left( \frac{d^2}{2^B} \right)$. Thus, using the triangle inequality, we have
    \begin{align}
         & \quad \left\| \hat{B}_2^\dagger \mathcal{T} e^{-i T \int_0^1 \hat{H}(s) \mathrm{d} s} \hat{B}_1 - (\bra{0} \otimes \normalfont{\openone}) \widetilde{V} (\ket{0} \otimes \normalfont{\openone}) \right\| \\
         &= \left\| \hat{B}_2^\dagger \mathcal{T} e^{-i T \int_0^1 \hat{H}(s) \mathrm{d} s} \hat{B}_1 - \bar{V}  - \left( (\bra{0} \otimes \normalfont{\openone}) \widetilde{V} (\ket{0} \otimes \normalfont{\openone}) - \bar{V} \right) \right\| \\
         &\leq \left\| \hat{B}_2^\dagger \mathcal{T} e^{-i T \int_0^1 \hat{H}(s) \mathrm{d} s} \hat{B}_1 - \bar{V} \right\| + \left\| \left( (\bra{0} \otimes \normalfont{\openone}) \widetilde{V} (\ket{0} \otimes \normalfont{\openone}) - \bar{V} \right) \right\|
    \end{align}
    Since $\left\| \hat{B}_2^\dagger \mathcal{T} e^{-i T \int_0^1 \hat{H}(s) \mathrm{d} s} \hat{B}_1 - \bar{V} \right\| \in O \left( \frac{\Gamma^4}{\gamma^2_{\rm min} T^2} + \frac{T}{r^3} \Lambda \right) $ and $\left\| \left( (\bra{0} \otimes \normalfont{\openone}) \widetilde{V} (\ket{0} \otimes \normalfont{\openone}) - \bar{V} \right) \right\| \in O \left( \frac{d^2}{2^B} \right)$, if we choose
    \begin{align}
        r \in \Theta \left( T \sqrt[3]{\frac{\gamma_{\min}^2\Lambda}{\Gamma^4}} \right)
    \end{align}
    and
    \begin{align}
        B \in \Theta \left( \log \left( \frac{\gamma_{\min}^2 T^2 d^2}{\Gamma^4}\right) \right),
    \end{align} then the overall combined error from the truncation, rounding and Riemann sum is
    \begin{align}
        \left\| \hat{B}_2^\dagger \mathcal{T} e^{-i T \int_0^1 \hat{H}(s) \mathrm{d} s} \hat{B}_1 - (\bra{0} \otimes \normalfont{\openone}) \widetilde{V} (\ket{0} \otimes \normalfont{\openone}) \right\| \in O \left( \frac{\Gamma^4}{\gamma_{\min}^2 T^2} \right).
    \end{align} Then, the linear combination of unitaries $\widetilde{V} = \sum_{\ell=0}^r \sum_{p=0}^2 \sum_{b=0}^{2^B-1} \sum_{c_1=0}^{d-1} \sum_{c_2=0}^{d-1} \frac{1}{2^B} (\hat{X} \otimes \normalfont{\openone} ) \widetilde{V}_{\ell, p, b, c_1, c_2}.$ is $\delta$-close to the exact time evolution $\hat{B}_2^\dagger \mathcal{T} e^{-i T \int_0^1 \hat{H}(s) \mathrm{d} s} \hat{B}$ where $\delta \in O \left( \frac{\Gamma^4}{\gamma_{\min}^2 T^2} \right)$.

    By Lemma 6 from \cite{bck2015}, if we have some linear combination of unitaries $\widetilde{V}$ that is within error $\delta$ of a unitary (in this case the unitary is the exact time evolution operator $\hat{B}_2^\dagger \mathcal{T} e^{-i T \int_0^1 \hat{H} (s) \mathrm{d} s} \hat{B}_1$), we can approximate $\widetilde{V}$ to within error $\delta$ by applying $\hat{\Pi} \hat{R}^p \hat{W} \hat{\Pi}$, where $p \in \mathbb{N}$ such that $\sin \left( \frac{\pi}{2 (2p+1)} \right) = \frac{1}{s}$ for $s$ greater than the sum of the absolute value of the coefficients in the linear combination of unitaries. In this case, we have $s \in \Theta \left( d^2 r \right) = \Theta \left( d^2 T \sqrt[3]{\frac{\gamma_{\min}^2 \Lambda}{\Gamma^4}} \right)$. The robust oblivious amplitude amplification implements an operator that is equal to the linear combination of unitaries $\widetilde{V}$ within error $O \left( \frac{\Gamma^4}{\gamma_{\min}^2 T^2} \right)$, and $\widetilde{V}$ is equal to the exact time evolution operator $\hat{B}_2^\dagger \mathcal{T} e^{-i T \int_0^1 \hat{H}(s) \mathrm{d}s} \hat{B}$ up to error $O \left( \frac{\Gamma^4}{\gamma^2_{\rm min} T^2} \right)$, so the robust oblivious amplitude amplification algorithm implements $\hat{B}_2^\dagger \mathcal{T} e^{-i T \int_0^1 \hat{H}(s) \mathrm{d}s} \hat{B}$ to within error $O \left( \frac{\Gamma^4}{\gamma_{\min}^2 T^2} \right)$.

\end{proof}
\begin{theorem}
    Let $\hat{H}(t)$ be a time-dependent Hamiltonian satisfying the $d$-sparsity definition in Def. \ref{def:dsparse2}, with a minimum gap of $\gamma_{\min}$ between any two eigenvalues, and let $\Gamma := O \left( \frac{\max \{ \| \dot{\hat{H}} \|, \| \ddot{\hat{H}} \|, \| \dddot{\hat{H}} \| \} }{\gamma_{\min}} \right)$. Then, we can approximately implement the time evolution operator $\mathcal{T} e^{-i T \int_0^1 \hat{H}(s) \mathrm{d} s}$ up to error $\epsilon$ for $T \in \Theta  \left( \frac{\Gamma^2}{\gamma_{\min} \sqrt{\epsilon}} \right)$ using $O \left( \frac{d^2}{\sqrt{\epsilon}} \sqrt[3]{\frac{\Gamma^2 \Lambda}{\gamma_{\min}}} \right)$ queries to the oracles defined earlier.
\end{theorem}
\begin{proof}
dEach iteration of $\hat{R}$ in the ROAA process requires a constant number of PREP and SEL operations and their inverses to implement, and each iteration of PREP and SEL requires a constant number of oracle queries by Lemma \ref{lemma:lt_prepsel}, so the overall number of queries required is $O (p) \subseteq O (s) \subseteq O \left( d^2 T \sqrt[3]{\frac{\gamma_{\min}^2 \Lambda}{\Gamma^4}} \right)$. Since ROAA implements the exact operator up to $O \left( \frac{\Gamma^4}{\gamma_{\min}^2 T^2} \right)$ error, if we have some overall allowed error tolerance $\epsilon$, we will need the evolution time to be at least $\Omega \left( \frac{\Gamma^2}{\gamma_{\min} \sqrt{\epsilon}} \right)$. If we have $T \in \Theta \left( \frac{\Gamma^2}{\gamma_{\min} \sqrt{\epsilon}} \right)$, then the query complexity is $O \left( \frac{d^2}{\sqrt{\epsilon}} \sqrt[3]{\frac{\Gamma^2 \Lambda}{\gamma_{\min}}} \right)$.
\end{proof}

\section{Conventions and Notation}
We use natural units where $\hbar = 1$ throughout, except for Section \ref{sec:lagrange} where the $\hbar$ is explicitly written. All logarithms are taken base 2. $\| \cdot \|_p$ for a vector denotes the $p$-norm, and $\| \cdot \|_p$ for an operator denotes the induced $p$-norm. $\oplus$ denotes bitwise addition. $\mathcal{H}_d$ denotes the complex $d$-dimensional Hilbert space of $\log d$ qubits ($\hat{H}_d \cong \mathbb{C}^d$). $\mathcal{L}(\mathcal{H})$ denotes a linear operator on the Hilbert space $\mathcal{H}$. All states written as $\ket{x}$ or $\bra{x}$ are assumed to be normalized. $\openone$ refers to the identity operator acting on an appropriately sized space. If the definition of an operator only specifies its action on part of the full Hilbert space, it is assumed that it acts as the identity on the rest of the space. All products of operators written with the product symbol $\prod$ are assumed to be time-ordered, i.e.,
\begin{align}
    \prod_{\ell=0}^{L-1} \hat{U}_\ell := \hat{U}_{L-1} \hat{U}_{L-2} \cdots \hat{U}_0,
\end{align}
and operators in the opposite order will be written as
\begin{align}
    \prod_{\ell=L-1}^{0} \hat{U}_\ell := \hat{U}_0 \hat{U}_{2} \cdots \hat{U}_{L-1}.
\end{align}
The time-ordered exponential of a time-dependent operator $\hat{A}(t)$ is defined as
\begin{align}
    \mathcal{T} e^{-i \int_0^T A(t) t} := \lim_{r \rightarrow \infty} \prod_{k=0}^{r} e^{-i \hat{A} \left( \frac{kT}{r} \right) \frac{T}{r} }.
\end{align}
The time derivative of a function $x(t)$ is denoted as $\dot{x}(t) := \frac{\partial x(t)}{\partial t}$, with more dots denoting higher order time derivatives ($\ddot{x}(t) := \frac{\partial^2 x(t)}{\partial t^2}$, etc.).

\end{document}